\documentclass[a4paper,10pt]{article}
\pdfoutput=1 

\usepackage{jheppub} 

\usepackage{amsmath, amsthm, amssymb}
\DeclareFontFamily{U}{mathx}{}
\DeclareFontShape{U}{mathx}{m}{n}{<-> mathx10}{}
\DeclareSymbolFont{mathx}{U}{mathx}{m}{n}
\DeclareMathAccent{\widehat}{0}{mathx}{"70}
\DeclareMathAccent{\widecheck}{0}{mathx}{"71}
\usepackage{dsfont}
\usepackage{wasysym, verbatim}
\usepackage{color}
\usepackage{tensor}
\usepackage{mathtools}
\usepackage{enumitem}
\usepackage{physics}
\usepackage{subcaption}
\usepackage{graphicx}
\usepackage[utf8]{inputenc}
\usepackage{booktabs,multirow, quotes,tablefootnote}
\usepackage{nicematrix}
\usepackage{tikz}
\NiceMatrixOptions
{
	custom-line =
	{
		letter = : ,
		command = dashedline ,
		ccommand = cdashedline ,
		tikz = dashed
	},
	custom-line =
	{
		letter = . ,
		command = dotrule ,
		ccommand = cdotrule ,
		tikz = dotted
	}
}

\usepackage{cancel}
\usepackage[normalem]{ulem}

\usepackage{braket}
\usepackage{tcolorbox}
\usepackage{float}

\usepackage{notoccite}

\usepackage{cleveref}
\creflabelformat{equation}{(#2\textup{#1}#3)}
\crefname{equation}{Eq.}{Eqns.}
\crefname{figure}{Fig.}{Figs.}
\crefname{section}{Section}{Sections}
\crefname{proposition}{Proposition}{Propositions}
\crefname{table}{Table}{Tables}
\crefname{chapter}{Chapter}{Chapters}
\crefname{appendix}{Appendix}{Appendices}
\crefname{subsection}{Section}{Sections}
\crefname{remark}{Remark}{Remarks}
\crefname{lemma}{Lemma}{Lemmas}
\crefname{theorem}{Theorem}{Theorems}
\crefname{definition}{Definition}{Definitions}
\crefname{footnote}{footnote}{footnotes}

\usepackage{cases}

\usepackage{xfrac}

\usepackage{nicefrac}

\usepackage{xcolor}
\usepackage{hyperref}

\numberwithin{equation}{section}

\theoremstyle{theorem}
\newtheorem*{conjecture*}{Conjecture}

\newtheorem{theorem}{Theorem}\numberwithin{theorem}{section}
\newtheorem{definition}{Definition}\numberwithin{definition}{section}
\newtheorem{lemma}[theorem]{Lemma}\numberwithin{theorem}{section}
\newtheorem{proposition}[theorem]{Proposition}\numberwithin{theorem}{section}
\theoremstyle{remark}
\newtheorem{remark}{Remark}\numberwithin{remark}{section}

\usepackage{adjustbox}

\def\eps{\epsilon}
\newcommand{\beq}{\begin{equation}} 
\newcommand{\eeq}{\end{equation}}

\def\bR {\mathbb{R}} 
\def\bC {\mathbb{C}}

\def\calN {{\cal N}} 
\def\cN {{\cal N}}

\def\ge{\geqslant}
\def\le{\leqslant}
\def\geq{\geqslant}
\def\leq{\leqslant}
\def\<{\langle}
\def\>{\rangle}

\def\s{\sigma}

\newcommand{\myinclude}[2][]{\raisebox{0.6ex}{\raisebox{-0.5\height}{\includegraphics[#1]{#2}}}}

\newcommand{\RG}{\mathcal{R}}

\newcommand{\red}{\color{red}}

\usepackage{enumitem}

\newcommand\x{\times}

\renewcommand{\ge}{\geqslant}
\renewcommand{\le}{\leqslant}

\renewcommand{\x}{\mathbf{x}}

\renewcommand{\o}{\mathbf{o}}
\renewcommand{\u}{\mathbf{u}}
\renewcommand{\d}{\mathbf{d}}
\renewcommand{\r}{\mathbf{r}}
\renewcommand{\s}{\mathbf{s}}
\newcommand{\AHT}{A_*}
\newcommand{\hatAHT}{\widehat{A}_*}
\renewcommand{\kappa}{\varkappa}
\renewcommand{\hat}{\widehat}

\newcommand{\vvvert}{\lvert\hspace{-1 pt}\lvert\hspace{-1 pt}\lvert}

\makeatletter
\def\@fpheader{\ }
\makeatother

\title{Tensor Renormalization Group Meets Computer Assistance}

\author[a]{Nikolay Ebel,}
\author[b]{Tom Kennedy}
\author[a]{and Slava Rychkov}

\affiliation[a]{Institut des Hautes \'Etudes Scientifiques, 91440 Bures-sur-Yvette, France}
\affiliation[b]{Department of Mathematics, University of Arizona,
	Tucson, AZ 85721, USA}

\emailAdd{ebelnikola@gmail.com}
\emailAdd{tgk@arizona.edu}
\emailAdd{slava@ihes.fr}

\usepackage{marginnote}

\newcounter{codeeq}
\newcommand*\codetag{%
  \stepcounter{codeeq}%
  \marginnote{\normalfont $\langle$code~\thecodeeq$\rangle$}%
}

\abstract{Tensor renormalization group, originally devised as a numerical technique, is emerging as a rigorous analytical framework for studying lattice models in statistical physics. Here we introduce a new renormalization map - the 2x1 map - which coarse-grains the lattice anisotropically by a factor of two in one direction followed by a 90-degree rotation. We develop a novel graphical language that translates the action of the 2x1 map into a system of inequalities on tensor components, with rigorous estimates in the Hilbert-Schmidt norm. We define a finite-dimensional “bounding box” called the hat-tensor, and a master function governing its RG flow. Iterating this function numerically, we establish convergence to the high-temperature fixed point for tensors lying within a quantifiable neighborhood. Our main theorem shows that tensors with deviations bounded by 0.02 in 63 orthogonal sectors flow to the fixed point. We also apply the method to specific models - the 2D Ising and XY models - obtaining explicit bounds on their high-temperature phase. This work brings the Tensor RG program closer towards a rigorous, computer-assisted construction of critical fixed points.
}

\begin{document}

\maketitle

\flushbottom

\section{Introduction}

Tensor renormalization group, born as a numerical tool \cite{Levin:2006jai,Evenbly-review}, is becoming a rigorous method to study statistical physics of lattice models \cite{paper1, paper2,Ebel:2024jbz}. The stability of the high-T phase was proved for 2D lattice models in Ref.~\cite{paper1} and for the 3D case in Ref.~\cite{Ebel:2024jbz}. Ref.~\cite{paper2} studied the low-T phase of $\mathbb{Z}_2$ symmetric models in 2D. One grand goal of this project, which appears within reach, is a rigorous computer-assisted construction of the nontrivial fixed points describing the critical point of the 2D Ising and other universality classes. This paper will bring us a few steps closer to that goal, as follows:

\begin{enumerate}
  \item
        We will introduce a novel tensor RG map which we call the 2x1 map (see Fig.~\ref{fig:bigfig}), which shrinks the lattice by a factor of 2 in only the horizontal direction, the vertical length being unchanged, after which the lattice is rotated by 90 degrees.\footnote{The idea that 90 degree rotation is useful for RG maps first appeared in our previous work~\cite{Ebel:2024nof,paper-DSO}.} The 2x1 map is remarkably simple and as such interesting both numerically and for rigorous studies. {In particular, the map is truncation-free and does not rely on singular value decomposition (SVD) as many numerical maps do. The latter is important as SVD would make it hard to obtain estimates for the tensor elements of the intermediate tensors and to establish analyticity of the map.}
  \item
    We will use the 2x1 map for a new rigorous look at the high-T phase, the problem previously studied in \cite{paper1} using a more complicated map. Our study here will rely on a new efficient graphical language for describing the structure of rigorous tensor RG maps involving disentangling, like the 2x1 map.

  \item Our graphical language allows us to seamlessly translate the definition of the map into a set of bounds on how various parts of the tensor change under the RG map. We can thus easily check that the 2x1 map is a contraction (in a weighted Hilbert-Schmidt norm) in a sufficiently small neighborhood of the high-T fixed point, as is the map from \cite{paper1}.

  \item
  Most importantly, the new graphical language is also ideal for implementing the bounds on a computer. This leads to concrete results about the size of the basin of attraction of the high-T fixed point. The main idea is to use a ``bounding box'' (called hat-tensor later on). The box has a large but finite number of dimensions corresponding to various "sectors" of the tensor. We define a \emph{master function} $\mathfrak{M}$ which computes how the bounding box varies from one RG step to the next. The master function is iterated numerically until the bounding box starts shrinking. It will then keep shrinking indefinitely, establishing convergence to the high-T fixed point. Using interval arithmetic, such computer-assisted results are fully rigorous.

  \item
  We are thus able to show that \emph{any initial tensor deviating from the high-T fixed point by at most 0.02 in the Hilbert-Schmidt norm in every one of $N=63$ orthogonal sectors that we will define, converges to the high-T fixed point under RG} (Theorem \ref{th:stability}). While 0.02 may seem like a small number, the full Hilbert-Schmidt norm of the initial deviation can be as large as $ 0.02\sqrt{N}\approx 0.16$.

  \item We also study a couple of specific lattice models (the nearest-neighbor Ising model and the XY model), translating them into tensor network representation. For those models we also obtain rigorous bounds on the size of the high-T phase. See Table \ref{tab:results} for a summary.

  \item {Tensor RG yields an expansion for the free energy density, alternative to the cluster expansion (Proposition~\ref{prop:free}). This is used to show that the free energy density is analytic in the basin of attraction of the high-T fixed point. This can be also used to compute rigorous bounds on the free energy density for specific models, as we illustrate for the Ising model, see \cref{fig:isingfbound}.}
\end{enumerate}

\begin{figure}[H]
	\centering
	\includegraphics[width=\textwidth]{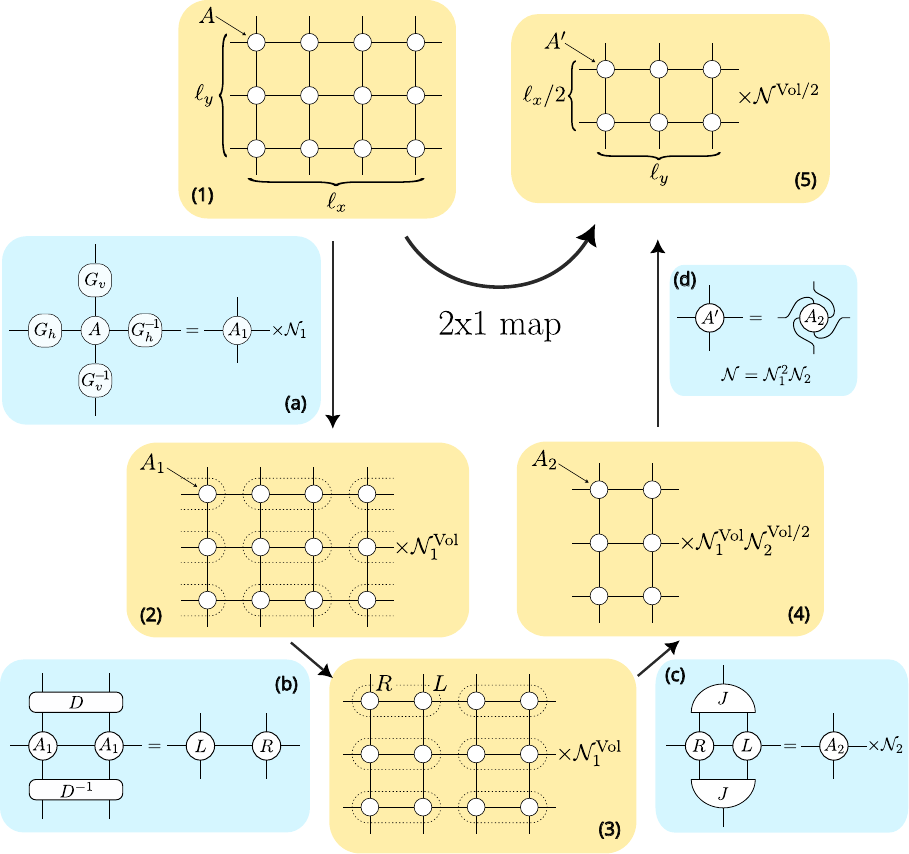}
	\caption{The 2x1 map (see Sec.~\ref{sec:2x1overview} for details) transforms an $\ell_x\times\ell_y$ tensor network made of tensor $A$ into an $\ell_y\times\ell_x/2$ tensor network made of tensor $\cN A'$. It is decomposed into 4 steps (blue panels): (a) gauge transformation, (b) disentangling and splitting;  (c) reconnection; (d) rotation. The tensor network partition function (assuming periodic boundary conditions) is kept invariant, but is represented differently after each step, as illustrated in yellow panels (1)-(5). $\mathrm{Vol}=\ell_x\ell_y$ is the volume of the original lattice.}
	\label{fig:bigfig}
\end{figure}

\begin{table}[h]
	\centering
	\begin{tabular}{ccc}
		\hline
		& 2x1 map result & $\beta_c$                               \\
		\hline
		\rule{0pt}{2.5ex}
		Ising                                           & $\beta \leq 0.12$          & $\frac{1}{2} \log(1 + \sqrt{2}) \approx0.44$  \\
		XY                                              & $\beta \leq 0.18995$          & $\approx 1.12$ \\
		General $A$ ($A_{0000}=1$) &
			$\|A_{abcd}\| \leq 0.02$
		&        ---                                    \\
		\hline
	\end{tabular}

	\caption{High-T phase regions proven here using the 2x1 map for the Ising model, the XY model, and general tensor networks built of a normalized tensor $A$ ($A_{0000}=1$). For Ising and XY, our bounds may be compared with known $\beta_c$ (third column). For general tensors, see \Cref{sec:sectors} for the definition of sectors $abcd$.}

\label{tab:results}
\end{table}

 Our paper builds on the intuition from \cite{paper1,paper2}, but it is self-contained. Section \ref{sec:construction} is the main technical section which defines the 2x1 map and the master function.
  After a brief reminder of tensor network notation in \Cref{sec:tensors}, we give an overview of the basic ingredients of the 2x1 map in \Cref{sec:2x1overview}. Then in \Cref{sec:lin} we describe the map at the linearized level, i.e.~for small perturbations around the high-T fixed point tensor. In particular in \Cref{sec:sectors} we define "sectors," which are restrictions of tensor to various combinations of subspaces on horizontal and vertical legs. We subdivide tensors into sectors and follow the norm of each sector separately. \Cref{sec:lin} culminates in an explanation why the map is a contraction at the linearized level, in an appropriate weighted generalization of the Hilbert-Schmidt norm, and for sufficiently large values of certain parameters $w_\x,w_\o>0$ called the reweighting parameters. We then proceed with the nonlinear analysis. In \Cref{sec:hat-tensors} we define the hat-tensors, which are central for our construction. They play the role of bounding boxes for tensors split into sectors. Hat-tensors reduce the control of the 2x1 map to a finite-dimensional computation of the master function. This basic strategy is explained in \Cref{sec:control-of-the-2x1-map-via-hat-tensors}. We then proceed to define the 2x1 map at the nonlinear level (Sections \ref{sec:gauge}-\ref{sec:reconnection-and-rotation}). This is put together in Section \ref{sec:masterf} where the master function is defined. This is a function from (a subset of) $\mathbb{R}_{\ge0}^{63}$ to itself. In \Cref{sec:leading} we come back to the question of how large the reweighting parameters $w_\x,w_\o$ must be, and derive a necessary condition. We conclude the discussion of the 2x1 map and of its master function in \Cref{sec:symmetries}, describing how they behave with respect to the symmetries that the lattice model might have.

In Section \ref{sec:results} we apply the 2x1 map to prove bounds on the high-T phase of the various models in Table \ref{tab:results}. The proofs are computer assisted. We provide a short computer code which evaluates the master function. One runs this code for a few iterations, until a condition is verified which implies convergence of subsequent iterations. A short description of the code is given in App.~\ref{sec:accompanying-code}, and the full documentation is on the GitHub repository \cite{our-code}. The code follows the paper closely. The equations referenced from the code are given margin tags of the form $\<\text{code }\#\>$; see e.g.~Eq.~\eqref{eq:tensornot}.

To conclude, in Section \ref{sec:conclusions} we discuss various future improvements of the method, which will further enlarge its domain of applicability. There are a few trivial improvements, like optimization of the  reweighting parameters $w_\x,w_\o$. We also describe a strategy for a more dramatic improvement, which will lead to accurate and systematically improvable results for the free energy (in this paper we are content with providing some crude estimates). Finally, we outline how we plan to use a similar improvement of the method to recover the critical point.

\section{RG map and the master function}\label{sec:construction}

\subsection{Tensors}
\label{sec:tensors}
The readers familiar with the tensor notation and terminology (see e.g.~\cite{paper2}, Sec.~2) may skip this short section.

In this paper $V$ will be an infinite-dimensional separable complex Hilbert space with an orthonormal basis $\{e_k\}_{k=0}^\infty$.
An $n$-leg tensor $T$ is a complex multilinear map on $V \cross V  \cross \cdots \cross V$. 
Equivalently, given an orthonormal basis $\{e_k\}_{k=0}^\infty$, we can think of $T$ as a table of complex numbers $T_{i_1,i_2\cdots,i_n} = T(e_{i_1}, e_{i_2}, \dots, e_{i_n})$. We will sometimes split the legs of the tensor into two groups, say the first $k$ and the last $n-k$, and think of it as a linear operator from $V^{\otimes k}$ to $V^{\otimes n-k}$ acting by $e_{i_1}\otimes\ldots\otimes e_{i_k} \mapsto T_{i_1,i_2\cdots,i_n} e_{i_{k+1}}\otimes\ldots\otimes e_{i_n}$.

The Hilbert-Schmidt (HS) norm of a tensor is
$
  \|T\| = ( \sum_{i_1,i_2,\cdots,i_n} |T_{i_1,i_2,\cdots,i_n}|^2 )^{1/2}
$.
Throughout this paper we will encounter three types of tensors:
\begin{enumerate}
  \item
        Hilbert-Schmidt tensors, i.e.~tensors having a finite HS norm;
  \item
        Tensors given by exponentiating a HS tensor with $2n$ legs, viewed as a linear operator on $V^{\otimes n}$. Those will be a sum of an identity operator and of a HS tensor;
  \item
        Isometric tensors which just permute basis elements.
\end{enumerate}

We will study models on a square lattice, and our renormalization group map will act on a HS tensor $A$ with four legs.\footnote{The letters $A$ and $b$ (see below) will always denote 4-leg tensors. Generic tensors having an arbitrary number of legs will be denoted by $T$.} In the usual graphical notation, we draw it as a blob with four legs representing the tensor indices. The first, second, third, and fourth index will always corresponds to legs pointing left, up, right, and down, respectively, e.g.:

\begin{equation}\label{eq:tensornot}
  \myinclude{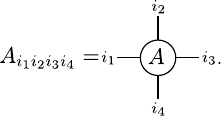} \codetag
\end{equation}
The graphical notation is convenient to specify tensor contractions, drawn as tensors sharing one or more legs. Such contracted legs are referred to as ``bonds''. We will also refer to left and right legs as ``horizontal'' and to up and down legs as ``vertical''.

{A comment is in order about the real vs complex. We work with complex tensors for generality, and also because our RG map will be analytic. However, some of our building blocks---various isometries---will be naturally real. Once we specify the Hilbert space basis $e_k$, the vectors $\sum c_k e_k$ with real $c_k$ form the real subspace of $V$. This defines the real structure on $V$.}

\subsection{The 2x1 RG map overview}
\label{sec:2x1overview}

The tensor index 0 will play a special role. We will say that a four-leg tensor is normalized if $A_{0000}=1$. The 2x1 map will act on a four-leg normalized HS tensor $A$. It is the RG map in the sense that it preserves the tensor network partition function. The partition function $Z(A,\ell_x\times \ell_y)$ is computed by placing a copy of $A$ at each vertex of the square $\ell_x\times \ell_y$ lattice and contracting legs along the edges, with periodic boundary conditions (see \cite{paper2}, Eq.~(2.7)). The result is a complex number, finite if $\ell_x,\ell_y\ge 2$ due to $A$ being Hilbert-Schmidt  (\cite{paper2}, Prop.~2.3). The 2x1 map will transform $A$ to another tensor $\cN A'$, where $\cN\in \bC$ and $A'$ is normalized, so that the partition function of $A$ on the $\ell_x\times \ell_y$ lattice equals that of $\cN A'$ on the $\ell_y\times (\ell_x/2)$ lattice:
\beq
\boxed{Z(A,\ell_x\times \ell_y)=Z(\cN A',\ell_y\times \ell_x/2)}\,.\label{eq:Zpreserved}
\eeq

(We stress that $\cN A'$ is independent of $\ell_x,\ell_y$.) That $Z$ remains the same while the total number of tensors in the network gets reduced is the defining property of any RG map. The unusual feature of the 2x1 map is how the dimensions of the lattice are transformed: $(\ell_x,\ell_y)\mapsto (\ell_y,\ell_x/2)$. This happens because the 2x1 map will first shrink the lattice by a factor of 2 in the horizontal direction, and then rotate by 90 degrees. Applying the 2x1 map twice, the lattice dimensions are both reduced by a factor of 2. One may thus say that the 2x1 map has ``effective scale factor $\sqrt{2}$''.

Let $A_*$ be the four-leg tensor with only one nonzero component $(A_*)_{0000}=1$. The infinite-temperature limit of 2D lattice models such as the Ising model corresponds to the tensor network built of $A_*$ \cite{paper1,paper2}. Furthermore, $A_*$ will be a fixed point of the 2x1 RG map. For these reasons we refer to $A_*$ as the high-T fixed point tensor.

We will study the 2x1 map acting on $A$
in a neighborhood of $A_*$.  Since $A$ is normalized, we can write
$A = A_* + b$, where $b_{0000} = 0$. Writing also $A'=A_*+b'$, the 2x1 map is thus fully specified by a function $\cN(b)$ ("$\cN$-factor") and a map $b'(b)$. These will be defined in a neighborhood of 0 of the Hilbert space of complex tensors equipped with the HS norm, and will depend on $b$ analytically.\footnote{See \cite{paper2}, App.~A for a few basic facts about analytic functions on Banach spaces.} Although physical models often correspond to real tensors, considering complex $b$ is convenient, as it allows to take advantage of this analyticity.

The $\cN$-factor trivially factors out of the tensor network partition function. The map $A\mapsto A'$ is called the normalized RG map, and its iterates generate the normalized RG flow. One of our main results (``stability of the high-T fixed point'', Theorem \ref{th:stability}  below) will say that the normalized RG flow converges to the high-T fixed point for any small initial $b$. This result is the key to showing that the infinite volume free energy is analytic in the same neighborhood (Proposition \ref{prop:free1}).

The general structure of the 2x1 map is similar but simpler than in \cite{paper1, paper2}. Like the maps in those works, the 2x1 map naturally splits in several steps. The first step (gauge transformation, Fig.~\ref{fig:bigfig}(a)) multiplies $A$ by $G_v$ and $G_v^{-1}$ on the opposite vertical legs, and by $G_h$ and $G_h^{-1}$ on the opposite horizontal legs. Here $G_v$ and $G_h$ are bounded linear operators with bounded inverses. The partition function is preserved for any $G_v$, $G_h$ (\cite{paper2}, Sec.2.4.3). We will use $G_v$ and $G_h$ depending on $b$ in a way specified in Sec.~\ref{sec:gauge}. This step produces tensor $A_{g} = \cN_1 A_1$ with $A_1$ normalized.

The second step (disentangling and splitting, Fig.~\ref{fig:bigfig}(b), Sec.~\ref{sec:disent}) groups $A_1$ tensors into horizontally contracted pairs, multiplies the upper vertical legs of the pair by a "disentangler" map $D:V\otimes V\to V\otimes V$, and multiplies the lower vertical legs by its inverse $D^{-1}$. The resulting 6-leg tensor is then split into a contraction of two 4-leg HS tensors $L$ and $R$. The partition function is preserved (\cite{paper2}, Sec.2.4.4). All these tensors will be expressed as power series in $b_1=A_1-A_*$. 

The third step (reconnection, Fig.~\ref{fig:bigfig}(c), Sec.~\ref{sec:reconnection-and-rotation}) contracts tensors $R$ and $L$ horizontally, in the opposite order from how these tensors were obtained during disentangling. The upper vertical legs are then contracted with an appropriately chosen real isometry\footnote{See Definition \ref{def:isometry}.} $J:V\otimes V\to V$, and the lower vertical legs with its transpose. The resulting tensor is denoted $\cN_2 A_2$ with $A_2$ normalized. Since $J^T J=\mathds{1}_{V\otimes V}$, the partition function is preserved. This is the step which halves the horizontal lattice size .

The final step (rotation, Fig.~\ref{fig:bigfig}(d), Sec.~\ref{sec:reconnection-and-rotation}) defines the tensor $A'$ as the counterclockwise rotation of $A_2$ by 90 degrees. This completes the overview of the 2x1 map. The $\cN$-factor is given by $\cN=\cN_1^2\cN_2$.

\subsection{Linearized analysis}\label{sec:lin}

We start in this section by considering the linearized 2x1 map (i.e.~to first order in $b$). The map was designed so that this linearization is a contraction (with respect to a norm defined below). We would like that the reader be convinced about that, before confronting the somewhat technical discussion of the full nonlinear map in the subsequent sections.

\subsubsection{Sectors}\label{sec:sectors}

We will decompose the tensor leg Hilbert spaces into a direct sum of orthogonal subspaces called "sectors". This is needed both at the linear and, later, at the nonlinear level. On the vertical legs, we decompose $V=\o\oplus\x$ where $\o$ is the one-dimensional sector spanned by $e_0$, and $\x$ is the infinite-dimensional orthogonal complement of $\o$, spanned by $\{e_k\}_{k=1}^\infty$. On the horizontal legs we decompose $V=\o\oplus\d \oplus \u \oplus \r$ where $\o$ is as before, and $\d$, $\u$, $\r$ are three infinite-dimensional sectors. (Their names stand for down, up, rest.) Why we use two vertical sectors and four horizontal sectors will become clear in the course of our analysis, see in particular \cref{sec:lin-reconnection-and-rotation}. We will sometimes use $\x=\d \oplus \u \oplus \r$ on the horizontal legs.

{The exact numbering of basis vectors which span $\d,\u,\r$ will not play any role. For definiteness, let them be spanned by $e_k$, $k\ge 1$, with $k\!\mod 3 =1,2,0$, respectively.}

Furthermore, we will decompose the tensors $A$ on which the 2x1 map acts into a collection of tensors obtained by restricting tensor legs to the above-defined sectors.
Namely, for a four-leg tensor $A$, given $a,c \in \{\o,\x\}$ and $b,d \in \{\o, \d, \u, \r\}$, we define the restricted tensor $A_{abcd}$ as the tensor whose components are equal to those of $A$ when all four indices are in the respective sectors and equal to zero otherwise. These restricted tensors will also be referred to as "sectors" of the original tensor. So there are $4 \times 2 \times 4 \times 2 =64$ sectors for $A$. Obviously, the sum of $A_{abcd}$ over all $64$ possibilities for $abcd$ equals $A$.  It will be sometimes convenient to set $b,d=\x$ on the horizontal legs. This will correspond to the sum of sectors over $b,d \in \{\d, \u, \r\}$. In graphical notation $A_{abcd}$ will be denoted by putting sector labels \emph{next} to the corresponding legs (unlike indices of tensor elements in \eqref{eq:tensornot}, put \emph{at the extremities} of the legs):
\begin{equation}\label{eq:tensorrestr}
	\myinclude{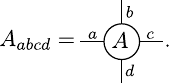}
\end{equation}
This definition extends to general tensors.
If $T$ is an $n$-leg tensor and $a_1, a_2,\cdots,a_n$ are sectors for each
of the legs, then the restricted tensor $T_{a_1, a_2, \cdots, a_n}$ is the tensor whose components are equal to those of $T$ when all of the indices are in the respective sectors and equal to zero otherwise.

\subsubsection{Gauge transformation}\label{sec:lin-gauge-transformation}
We now proceed with the linearized analysis of the 2x1 map. We write $A=A_*+b$, $A'=A_*+b'$. Our purpose is to compute $b'$ to first order in $b$.

The first step of the 2x1 map (Fig.~\ref{fig:bigfig}(a)) is a gauge transformation. Since our final goal is to reduce $b$, let us try to somewhat reduce it via the gauge transformation before doing something more sophisticated. We choose the gauge transformation matrices $G_h$ and $G_v$ so that\footnote{In this section, $\approx$ means modulo $b^2$ terms.}
\beq
\label{eq:Glin}
G_h\approx \mathds{1}_V+g_h,\qquad G_v\approx \mathds{1}_V+g_v\,,
\eeq
with $g_h$, $g_v$ linear in $b$. We apply $G_h,G_v$ and their inverses to the horizontal, vertical legs of $A$, respectively, as shown in Fig.~\ref{fig:bigfig}(a). We denote the resulting tensor by $A_{g}$. It is easy to see that
\beq
\label{eq:lingauge}
\myinclude{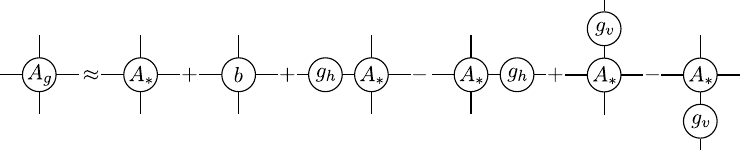}.
\eeq
We define the two-leg tensors $g_h$ and $g_v$  by the following expressions:
\begin{equation}
	\label{eq:ghgv}
	g_h = b_{\o 0 \x 0} - b_{\x 0 \o 0},  \quad  g_v = b_{0 \o 0 \x} - b_{0 \x 0 \o}.
\end{equation}
A word about the notation here. The tensor $b$ has four legs. We can form two-leg tensors from it by setting the indices on two of the legs to $0$. For example, $b_{\x 0 \o 0}$ means that the vertical legs are fixed to be $0$, the left leg is restricted to the $\x$ sector\footnote{Here we are using the convention that $\x$ on horizontal legs stands for $\d \oplus \u \oplus \r$.} and the right leg is restricted to the $\o$ sector. (We recall that restricting the leg to the $\o$ sector is not the same as setting the index of the leg to $0$. The former leaves the leg, the latter removes it.) $b_{\o 0 \x 0}$ is similar, except that now the left leg is restricted to the $\o$ sector and the right leg to the $\x$ sector. So $b_{\x 0 \o 0}$ and $b_{\o 0 \x 0}$ are tensors with only horizontal legs. Similar definitions for $b_{0 \x 0 \o}$ and $b_{0 \o 0 \x}$ give tensors with only vertical legs. Eq.~\eqref{eq:ghgv} defines $g_h$ and $g_v$ as linear combinations of these tensors. Now, with this definition, it is easy to see that the last four terms in \eqref{eq:lingauge} cancel precisely the sectors
\beq
\label{eq:set0}
b_{\x \o \o \o}, b_{\o \x \o \o}, b_{\o \o \x \o}, b_{\o \o \o \x}
\eeq
 of the $b$ tensor (the second diagram), while the other sectors are left unchanged at first order. This accomplishes the stated goal of "somewhat reducing $b$".\footnote{This is the best (at first order in $b$) one can do with a gauge transformation.} We also have $\calN_1=(A_{g})_{0000}\approx 1$. So the normalized tensor $A_1=A_{g}/\calN_1$ is written as $A_*+b_1$ where $b_1\approx b$ with sectors \eqref{eq:set0} set to zero. When we describe the subsequent steps, we simply rename $A_1$ and $b_1$ as $A$ and $b$ and assume that the sectors \eqref{eq:set0} are zero.

\subsubsection{Disentangling}\label{sec:lin-disentangling}
The second step of our RG map is disentangling and splitting, shown in Fig.~\ref{fig:bigfig}(b) which we copy here (recall that we renamed $A_1$ as $A$ and $b_1$ as $b$):
\beq
\myinclude{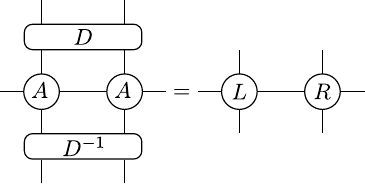}\,.
\label{eq:AADD2LR}
\eeq
Our goal will be to compute $L$ and $R$ to first order in $b$. For this we will need to include some diagrams which are second and third order in $b$. This is different from the previous section on gauge transformation where all considered diagrams were first order in $b$.

When we contract the two $A$ tensors in \eqref{eq:AADD2LR} and expand in powers of $b$, we get four diagrams:
\beq
\label{eq:TT}
  \myinclude{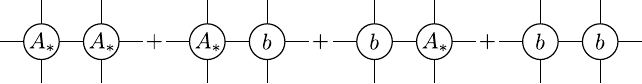}\,.
\eeq
The disentangler will be chosen as
\beq
\label{eq:Dlin}
D=\mathds{1}_{V\otimes V}+X+O(b^4),
\eeq
where $X=O(b^2)$ is given by
 \begin{equation}
 	\label{eq:Xdeflin}\,.
 	\myinclude{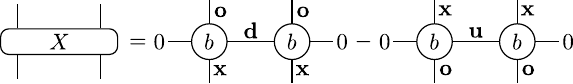}\,.
 \end{equation}
When we act with $D$ and with $D^{-1}=\mathds{1}-X+O(b^4)$ on the contraction of two $A$'s, the important term to the considered order in $b$ is when $X$ is contracted to the first diagram in \eqref{eq:TT}. This term cancels the following part of the last diagram in \eqref{eq:TT}:
\beq
	\myinclude{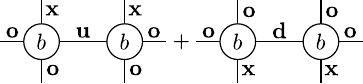}\,.
	\label{eq:aa-prob}
\eeq
The 4 sectors of $b$ involved here have one horizontal leg $\u$ or $\d$, one vertical leg $\x$, and two other legs $\o$. Furthermore, the position of the $\x$ leg (up or down) is correlated with whether it's $\u$ or $\d$ on the horizontal leg. This correlation is the raison d'\^etre of the $\u$ and $\d$ sectors.
	We will see below how it is generated naturally during the reconnection and rotation steps of the map.

We still have to perform the splitting, i.e.~find $L$ and $R$ so that Eq.~\eqref{eq:AADD2LR} holds. We write
\beq
\label{eq:LR1storder}
L=A_*+b_L, \quad R= A_*+b_R,
\eeq
and we have to find $b_L,b_R$ to linear order in $b$.
Let us postpone this task and see first what happens when we perform the reconnection and rotation.

\subsubsection{Reconnection and rotation}\label{sec:lin-reconnection-and-rotation}

Reconnection, Fig.~\ref{fig:bigfig}(c), contracts $L$ and $R$ in the opposite order, which at the linearized level amounts to the three diagrams:
\beq
\label{eq:threediag}
\myinclude{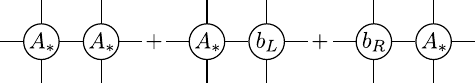}\,.
\eeq
{Then, a real isometry $J:V\otimes V \to V$ acts on the pairs of horizontal legs. Taking into account the subsequent rotation, we need to specify
\beq
J:(\o\oplus \x)\otimes (\o\oplus \x) \to \o \oplus \d \oplus \u \oplus \r.
\eeq
We make the following choices for how $J$ acts on various subsectors (see Section \ref{sec:isometries} below):
\beq
\label{eq:J_who_where}
J: \o \otimes \o \to \o,\quad J:\x\otimes \o \to \d,\quad J:\o\otimes \x \to \u, \quad  J:\x\otimes \x \to \r.
\eeq
By the first of these assignments, the first diagram in \eqref{eq:threediag} is mapped to $A_*$.}

The following remarks help understand the other assignments. We call ``important'' the sectors of $b_L$, $b_R$ which contribute to \eqref{eq:threediag}. These are:
\beq
\label{eq:lrimp}
\myinclude{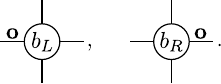}
\eeq
The second and third diagrams in \eqref{eq:threediag} will have an interesting structure: the former has $\o\x$ or $\o\o$ on the vertical pairs of legs (but not both $\o\o$ for the reason stated immediately after), while the latter satisfies a similar rule with  $\o\x$ replaced by $\x\o$. "Both $\o\o$" are excluded because, see \cref{sec:l-r-and-the-norm-of-kss},
\beq
\label{eq:lr}
(b_L)_{\o\o\o\o},(b_L)_{\o\o\x\o},(b_R)_{\o\o\o\o},(b_R)_{\x\o\o\o}\approx 0\,.
\eeq
The described asymmetry between $\x\o$ and $\o\x$ suggests that we should treat these sectors asymmetrically, which is why we map them by $J$ into two orthogonal sectors $\d$ and $\u$. Sector $\r$ is then introduced for the isometric image of $\x\otimes\x$.\footnote{The sector $\x\otimes\x$ will not actually occur in the linearized analysis; we are discussing it here for completeness.}

Then, after counterclockwise 90 degree rotation, as in Fig.~\ref{fig:bigfig}(d), we will obtain diagrams in the following 8 sectors, referred below as "special":
\beq
\label{eq:special}
\o\x\u\o, {\u\x\o\o}, {\o\o\d\x}, {\d\o\o\x}, {\u\x\u\o}, {\u\o\u\o}, {\d\o\d\x}, {\d\o\d\o}\,.
\eeq
The first four special sectors are the ones involved in the canceled diagram \eqref{eq:aa-prob}.

\subsubsection{Infinite-dimensional isometries}\label{sec:isometries}

{ In several places in this paper we will need to specify isometries between infinite-dimensional Hilbert spaces. We recall
\begin{definition}\label{def:isometry}An isometry $U:Y\to Z$ between complex Hilbert spaces $Y$ and $Z$ is a linear map which preserves the norm; it satisfies $U^\dagger U=\mathds{1}_{Y}$. For complex Hilbert spaces with real structures, a real isometry is one which is real with respect to those structures; it satisfies $U^TU=\mathds{1}_{Y}$.
	\end{definition}
The isometries will be used as follows. Given a contraction of two tensors $L\text{\bf ---}R$ where the contracted bond lives in $Y$, and a real isometry $U:Y\to Z$, we ``insert $U^TU=\mathds{1}_{Y}$ on the contracted bond''. That is, we define two new tensors $L^U = L\text{\bf ---}U$ and $R^U = U\text{\bf ---}R$. Then the contraction is preserved: $L^U\text{\bf ---}R^U=L\text{\bf ---}R$ where the first contraction is over $Z$.

The isometry is, of course, injective, but it is not required to be surjective. If it's not surjectve the contraction is still preserved, but some part of the $Z$ space is not used. If we want, we can make the isometry surjective replacing $Z$ by the image of $U$.

\begin{remark} Let us compare with the numerical tensor RG algorithms \cite{Evenbly-review}, which work with finite-dimensional $Y$ and $Z$, ${\rm dim} Y>{\rm dim} Z$. One chooses an orthogonal decomposition $Y=Y_1\oplus (Y_1)^\perp$, ${\rm dim} Y_1={\rm dim} Z$ and considers a ``partial isometry'' $U:Y\to Z$ which maps $Y_1$ isometrically onto $Z$, while $(Y_1)^\perp$ is projected to zero. Then in general $L^U\text{\bf ---}R^U\ne L\text{\bf ---}R$, i.e. the contraction cannot be exactly preserved, and one tries to choose $Y_1$ to minimize the ``truncation error''. (Ref.~\cite{Evenbly-review} formulates the discussion in terms of the transpose $U^T: Z\to Y$ which is an isometry, but because it acts ``in the wrong direction'', the contraction is not preserved.)
	\end{remark}

Let us state a simple general result about the isometries of the kind we will need.

For two countable sets of basis vectors $e \in \mathcal{A}$ and $f\in \mathcal{B}$, disjoint unions of finitely many subsets,
\beq
\mathcal{A} = \bigsqcup_{i=0}^I \mathcal{A}_i,\quad \mathcal{B} = \bigsqcup_{k=0}^K \mathcal{B}_k,
\eeq
let $\mathbf{y}_i$ (resp. $\mathbf{z}_k$) be complex Hilbert spaces with orthonormal basis vectors $e\in \mathcal{A}_i$ (resp. $f \in \mathcal{B}_k$), and let $Y$ and $Z$ be the direct sums
\beq
Y = \bigoplus_{i=0}^I \mathbf{y}_i, \quad Z = \bigoplus_{k=0}^K \mathbf{z}_k\,.
\eeq

\begin{lemma}\label{lem:isometry} Suppose that $\mathbf{y}_0$ and $\mathbf{z}_0$ are one-dimensional and all the other $\mathbf{y}_i$, $\mathbf{z}_k$ are infinite-dimensional. Then, given any map $p:\{1,\ldots,I\}\to \{1,\ldots,K\}$, there exists an isometry $J:Y\to Z$ such that $J:\mathbf{y}_0\to  \mathbf{z}_0$ and $J:\mathbf{y}_i\to  \mathbf{z}_{p(i)}$ for $i=1,\ldots,I$. This isometry can be chosen real with respect to the specified basis.
\end{lemma}

\begin{proof}
We will define $J$ on the basis vectors of $Y$ by the formula:
	\beq
	J(e) = \phi(e)\,,
	\eeq
	where $\phi:\mathcal{A}\to \mathcal{B}$ is an injective function defined below such that $\phi(\mathcal{A}_0)=\mathcal{B}_0$ and $\phi(\mathcal{A}_i) \subset \mathcal{B}_{k}$ where $k=p(i)$.
In words, every basis vector of $Y$ will be mapped to a distinct basis vector of $Z$, with coefficient 1. By linearity, this extends to a real isometry from $Y$ to $Z$ with needed properties.

It remains to define $\phi$.

On $\mathcal{A}_0$, $\phi$ simply maps its unique element to the unique element of $\mathcal{B}_0$.

On the infinite-dimensional sectors, $\phi$ should map $\mathcal{A}_i$ into $\mathcal{B}_k$ where $k=p(i)$, and it should be injective. If $p$ is injective, we enumerate the elements of $\mathcal{A}_i$ and $\mathcal{B}_k$ in some arbitrary order, and $\phi$ maps the $n$-th element of $\mathcal{A}_i$ to the $n$-th element of $\mathcal{B}_k$. If $p$ is not injective, then a simple modification is needed. As an example, suppose that $p^{-1}(1)=\{1,2\}$. Then we arrange so that $\phi$ maps elements of $\mathcal{A}_1$ (resp. $\mathcal{A}_2$) to odd-numbered (resp. even-numbered) elements of $\mathcal{B}_1$. In the general case when $|p^{-1}(k)|=r$ we split the sequence of elements of $\mathcal{B}_k$ modulo $r$.
\end{proof}


For example, the isometry specified in Eq.~\eqref{eq:J_who_where} is obtained by applying the lemma with
\beq
(\mathbf{y}_{0}, \mathbf{y}_{1}, \mathbf{y}_{2}, \mathbf{y}_{3})=(\o\otimes \o, \x\otimes \o, \o\otimes\x ,\x\otimes \x),\qquad
(\mathbf{z}_{0}, \mathbf{z}_{1}, \mathbf{z}_{2}, \mathbf{z}_{3})=(\o, \d, \u ,\r)\,.
\eeq
To follow through the details explicitly, we need to enumerate basis elements of $\x\otimes \x$. They come naturally as tensor products $e_i\otimes e_j$ of basis elements of $\x$, so they are indexed by pairs $(i,j)$, $i,j\in \mathbb{N}$. To enumerate them, one can use any pairing function \cite{wikipedia_pairing_function}, e.g.~the Cantor pairing function.}



\subsubsection{When is the linearized map a contraction?}\label{sec:lin-norm}
We just showed that after the RG step $b'$ has only the special components \eqref{eq:special} nonzero (to first order). This leads to some simplification of the analysis when the linearized RG map is a contraction, and with respect to which norm. Let $P_s$ ($P_n$) be the orthogonal projectors on the sum of special (non-special) sectors.  Let $\RG: A \mapsto A'$ be the RG map and $\nabla \RG$ be its Jacobian at $A_*$. Then $P_n \nabla \RG=0$ and we can restrict our attention to $P_s \nabla \RG$. We decompose $b=b_s+b_{n}$ where $b_s=P_s b$, $b_{n}=P_n b$. The important part of the Jacobian is $K_{ss} = P_s(\nabla \RG)P_s$, which shows how $b_s'$ depends on $b_s$ (to first order). The less important part is $K_{sn} = P_s(\nabla \RG)P_n$ which shows how $b_s'$ depends on $b_n$. The $K_{ss}$ and $K_{sn}$ are some bounded linear operators (with respect to the HS norm).  Consider the following weighted norm:
\beq
\label{eq:weighted}
\vvvert b \vvvert = \max(\|b_s\|, C\|b_n\|),
\eeq
where $\|\cdot\|$ is the HS norm and $C>0$ is a constant to be determined. As $P_n \nabla \RG=0$ we have
\beq
\label{eq:KK}
\vvvert (\nabla \RG)b \vvvert= \| (P_s\nabla \RG) b\|\le \|K_{ss}\|\|b_s\|+ \|K_{sn}\|\|b_n\| \le
\nu \vvvert b \vvvert,\quad \nu=\|K_{ss}\|+ C^{-1}\|K_{sn}\|\,.
\eeq
We see that if $\|K_{ss}\|<1$ and $\|K_{sn}\|$ is finite, we can always choose $C$ so that $\nu<1$ and $\nabla \RG$ is a contraction with respect to $\vvvert\cdot \vvvert$. That we only need to worry about $\|K_{ss}\|<1$ is the above-mentioned simplification.

\subsubsection{$b_L$, $b_R$ and the norm of $K_{ss}$}\label{sec:l-r-and-the-norm-of-kss}

We come back to the task of defining $L,R$ as in \eqref{eq:LR1storder}. We will also see how the condition $\|K_{ss}\|<1$ comes about.

The idea to define $L$ and $R$ is simple: we have to "cut" the l.h.s.~of \cref{eq:AADD2LR} vertically into two halves, assigning the left half to $L$, and the right half to $R$. The contraction of two $A$'s and the diagrams defining $X$ can both naturally be cut in two halves.

We expand in powers of $b$. Let $n_L$ (resp. $n_R$) denote the order of the left (resp. right) half of a diagram in $b$. Since here we only need $b_L$ and $b_R$ to first order, we have to keep diagrams where $n_L\le 1$ or $n_R\le 1$.

We start by discussing the diagrams obtained by expanding \cref{eq:AADD2LR} to $O(b^2)$. These are the diagrams in \eqref{eq:TT}, with diagrams from \eqref{eq:aa-prob} cancelled out. These are the most important diagrams. A few remaining diagrams, which are $O(b^3)$, will be dealt with below.

{ We will have to pay attention to two things. First, some diagrams cancel, so they drop out and do not have to be reproduced. Second, gluing of $L$ and $R$ must reproduce exactly the remaining diagrams; we don't want to introduce ``cross terms.''

The second issue will be handled by the method of ``auxiliary vectors''. Let us give an example. Suppose we need to reproduce the sum of two diagrams:
\beq
\label{eq:2diags-ex}
\myinclude{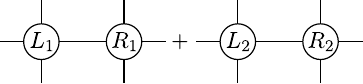},
\eeq
where the contracted bond lives in $V$. Let $V_{\rm aux}$ be an auxiliary two-dimensional Hilbert space spanned by the orthonormal basis $(e_1,e_2)$. We define $L$ and $R$ whose right (resp. left) leg lives in $V\otimes V_{\rm aux}$:
\beq
L = L_1 \otimes e_1+L_2 \otimes e_2= \myinclude{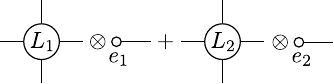},
\eeq
\beq
R = R_1 \otimes e_1+R_2 \otimes e_2 = \myinclude{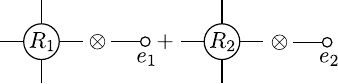}.
\eeq
(The auxiliary vector is tensored with the tensor's leg which points in the same direction.)
Then the contraction of $L$ and $R$ reproduces \eqref{eq:2diags-ex} without unwanted cross terms. This finishes the example.

For our real problem, we will need a larger, seven-dimensional auxiliary
Hilbert space $V_{\rm aux}$ spanned by the following orthonormal basis:
\beq
\label{eq:Vaux0}
V_{\rm aux} = {\rm span}(e_z, e_{\x\x}, e_{\o\x}, e_{\x\o}, e_{\o\o}, e_{\o\o}', e_{\o\o}^{\prime \prime}).
\eeq
Each of these auxiliary vectors will be used to reproduce a sum of certain group of diagrams called a ``channel''.
The vectors $e_z, e_{\x\x}, e_{\o\x}, e_{\x\o}$ will be used for the $z, \x\x, \o\x, \x\o$ channels,
respectively, while the last three vectors will be used for the $\o\o$ channel.

Because we introduced the auxiliary vectors, Eq.~\eqref{eq:LR1storder} needs to be modified. The following modification is appropriate:
\beq
L=A_*\otimes e_z+\tilde b_L, \quad R=  A_* \otimes e_z+\tilde b_R,
\eeq
where $e_z$ is tensored with the right (resp. left) leg of $A_*$ in the first (resp. second) equation. We will define $\tilde b_L$ and $\tilde b_R$, and then will pass to $b_L$ and $b_R$ by applying an isometry on the contracted leg in the $L \text{\bf ---} R$ contraction (see below).
}


The "$z$ channel" ($z$ for zero) consists of the first three diagrams in \eqref{eq:TT}, and of the fourth diagram with the contracted leg restricted to $\o$. This channel is reproduced by defining
\beq
\label{eq:zchannel-lin}
 \myinclude{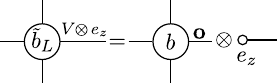}\,, \qquad \myinclude{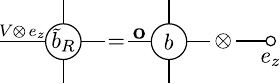}.
\eeq

Consider next the fourth diagram in \eqref{eq:TT} with the contracted leg restricted to $\x$. There $\x\x$, $\x\o$, $\o\x$, $\o\o$ channels are defined by restricting the the external horizontal legs of this diagram to $\x\x$, $\x\o$, $\o\x$, $\o\o$. The first three channels are reproduced by defining (see below for the meaning of boxes around some equations)
\begin{equation}\label{eq:linchan}
  \begin{array}{@{}l@{\hspace{1cm}}l@{}}
    \myinclude{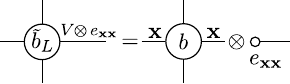}, & \myinclude{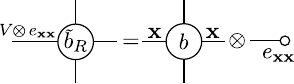}, \\
    \noalign{\vspace{10pt}}
    \myinclude{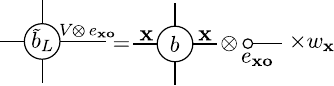}, & \boxed{\myinclude{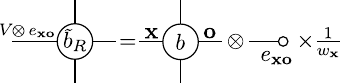}}, \\
    \noalign{\vspace{10pt}}
    \boxed{\myinclude{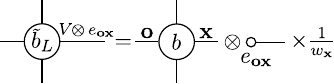}}, & \myinclude{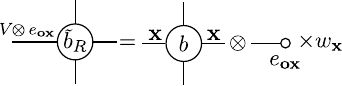}.
  \end{array}
  \end{equation}
  We also introduced a ``reweighting parameter'' $w_\x>0$ for the $\x\o$ and $\o\x$ channels. Playing with this parameter will help us achieve $\|K_{ss}\|<1$ (see below).

Finally consider the $\o\o$ channels. In this sector we have cancellations: the diagrams \eqref{eq:aa-prob} were canceled by the disentangler. We are left with the five diagrams
\begin{gather}
  \myinclude{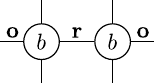}\, +\,   \myinclude{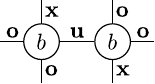}\,  +\,  \myinclude{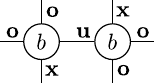}
\, +\,   \myinclude{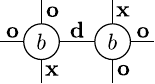}\, +\,  \myinclude{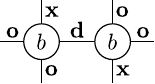}\label{eq:00diags}
\end{gather}
(Note that we omitted the diagrams involving $b_{\x \o\o\o}$ and $b_{\o\o\x\o}$. Those diagrams are higher order, since $b_{\x 000},b_{00\x0}\approx 0$ as a result of the gauge transformation.) The first three diagrams are reproduced by setting:
\beq
\label{eq:diag1split}
\myinclude{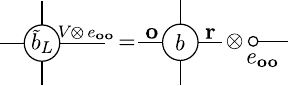}\,, \qquad\qquad \myinclude{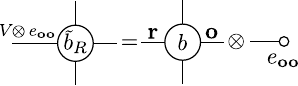}.
\eeq
\begin{equation}\label{eq:diag23split}
  \begin{array}{@{}l@{\hspace{1cm}}l@{}}
    \boxed{\myinclude{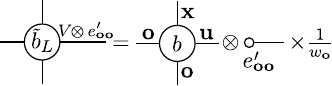}}, & \myinclude{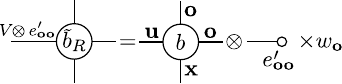}, \\
    \noalign{\vspace{6pt}}
    \myinclude{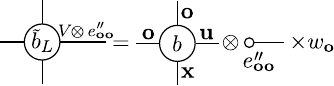}, & \boxed{\myinclude{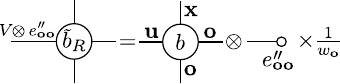}}.
  \end{array}
\end{equation}
The last two diagrams are reproduced by \cref{eq:diag23split} with $\u$ replaced by $\d$ in the r.h.s and all labels on vertical legs swapped. Here we introduced another reweighting parameter $w_\o>0$, needed for $\|K_{ss}\|<1$.

This concludes the treatment of the diagrams up to $O(b^2)$. Consider next the following diagram
\beq
\label{eq:cubic}
\myinclude{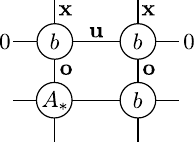}\,,
\eeq
which appears when we expand the l.h.s.~of \cref{eq:AADD2LR} at $O(b^3)$. It has $n_L=1$, $n_R=2$, so when we cut it vertically, it contributes to $\tilde b_L$ at linear order. The left leg of $A_*$ being in $\o$, this $\tilde b_L$ will contribute to $b'$. To suppress it, we use reweighting, which works here since $\tilde b_R$ is second order.
For simplicity let us use the same reweighting parameters: $w_\o$ if the right external leg is also in $\o$, and $w_\x$ if it's in $\x$. There are 3 more $O(b^3)$ diagram like \eqref{eq:cubic}, all treated in the same way.

{ At this stage we constructed $\tilde b_L$ and $\tilde b_R$ so that the contraction $L\text{\bf ---} R$ reproduces all the needed diagrams. The contracted bond lives in the space $V'=(V\otimes V_{\rm aux}) \oplus (\o \otimes \u) \oplus (\d \otimes \o)$. The first term in this direct sum comes from the splitting of the $O(b^2)$ diagrams, and the second and third terms come from the splitting of the $O(b^3)$ diagrams. We would like instead that the contracted bond lives in $V$, so that after reconnection we get a tensor with horizontal legs in $V$. This is easily achieved as follows. Using \cref{lem:isometry}, we pick a real isometry $\varkappa:V' \to V$ such that
\beq\label{eq:lin-kappa}
	\kappa: e_{0} \otimes e_z \mapsto e_0,\quad \kappa: (e_{0} \otimes e_z)^\perp \to \x.
\eeq
We define tensors $L^\kappa$ and $R^\kappa$ by contracting the $V'$ leg of $L$ and $R$ with $\kappa$, i.e.
\beq
L^\kappa=L\text{\bf ---}\kappa,\quad R^\kappa=\kappa \text{\bf ---} R\,.
\eeq
Because $\kappa^T \kappa=\mathds{1}_{V'}$, we have $L^\varkappa\stackrel{V}{\text{\bf ---} }R^\kappa = L\stackrel{V'}{\text{\bf ---}} R$. We have
\beq
\label{eq:lin-L-kappa}
L^\kappa=A_*+b_L,\quad R^\kappa=A_*+b_R
\eeq
where $A_*$ is obtained by contracting the zeroth-order terms in $L$ and $R$ with $\kappa$, while $b_L$ and $b_R$ are defined as contractions of $\tilde b_L$ and $\tilde b_R$ with $\kappa$. This finishes the definition of $b_L$ and $b_R$ to first order.}

Now let us see how $\|K_{ss}\|<1$ is obtained given these $b_L$ and $b_R$, with the help of the reweighting parameters. We only have  to pay attention to the sectors of $b_L$ and $b_R$ that satisfy both of the following two conditions:
\begin{itemize}
	\item
	They are the "important" sectors from Eq.~\eqref{eq:lrimp}, i.e.~have $\o$ on the "external" horizontal leg. Only these sectors contribute to $b'$ at linear order.
	\item
	They are proportional to $b$ from the special sectors. Only these sectors contribute to $K_{ss}$.
\end{itemize}
There are only four assignments in Eqs.~\eqref{eq:zchannel-lin},\eqref{eq:linchan},\eqref{eq:diag1split},\eqref{eq:diag23split} which give rise to $b_L$, $b_R$ sectors satisfying both these conditions; we put boxes around them.
The key point is that all these equations contain rescalings by $1/w_\x$ or $1/w_\o$. Choosing $w_\x$ and $w_\o$ large we will suppress them and achieve $\|K_{ss}\|<1$. The remaining sectors do not matter since they either do not contribute to $b'$ at linear order, or contribute only to $K_{sn}$ and not to $K_{ss}$.

This is also a good moment to see why we needed to disentangle. If we did not cancel the diagrams  \eqref{eq:aa-prob} by the disentangler, we would need additional $b_L$, $b_R$ sectors to reproduces them. But both halves of those diagrams belong to special sectors, hence they both contribute to $K_{ss}$. Therefore, there would be no way to achieve $\|K_{ss}\|<1$ by reweighting.\footnote{A related comment is that tensor RG maps without disentangling ("simple RG maps") typically have eigenvalue 1 perturbations at the high-T fixed point, associated with the so called CDL tensors; see \cite{paper1}.}

We still owe a check that the $b_L, b_R$ sectors mentioned in \eqref{eq:lr} vanish to first order. This follows by the inspection of \eqref{eq:zchannel-lin},\eqref{eq:linchan},\eqref{eq:diag1split},\eqref{eq:diag23split}, and the diagrams obtained by cutting the cubic terms like \eqref{eq:cubic}. There are only a few diagrams which contribute to \eqref{eq:lr}, and they are all proportional to the corresponding sectors of $b$. These sectors are either zero ($b_{\o\o\o\o}$), or have been set to zero to first order by the gauge transformation.


This finishes the discussion of $b_L$, $b_R$ up to linear order. As explained we can make $\|K_{ss}\|<1$  by choosing $w_\x$, $w_\o$ sufficiently large. An explicit condition on $w_\x$, $w_\o$ for this to happen will be worked out in \cref{sec:leading} below. We then choose $C$ sufficiently large so that $\nu<1$ in \eqref{eq:KK} and the linearized RG map is a contraction with respect to the $\vvvert\cdot\vvvert$ norm.

Hopefully, it is now less mysterious how we designed our RG map and why we think it's a good map. We will now proceed to the discussion of the map at the full nonlinear level. This will necessarily be more complicated, but the main ideas will be the same. We will also introduce a graphical language which will automatize the task of keeping track of various diagrams and making sure that no diagram is forgotten.

\subsection{Hat-tensors}
\label{sec:hat-tensors}

As we have already seen in the linearized discussion, various tensors will be separated into sectors. An essential ingredient of this paper will be to control the HS norm of each sector separately. (We have already seen a simple example of this in the weighted norm \eqref{eq:weighted}.) This leads to more refined estimates than using the overall HS norm of $b$ as in \cite{paper1, paper2,Ebel:2024jbz}.

\begin{definition}\label{def:hat}
  Let $T$ be a HS tensor with $n$ legs. We will say that an $n$-leg tensor
  $\widehat{T}$ is a hat-tensor for $T$ if
  \begin{equation}\label{eq:hat-condition}
    \| T_{a_1 a_2\cdots a_n} \| \le \widehat{T}_{a_1 a_2\cdots a_n} \,,
  \end{equation}
  where each index $a_i$ ranges over the sectors for the index space for that
  leg. (So $\widehat{T}$ is a finite dimensional tensor with real nonnegative components.)
\end{definition}

\begin{remark} \label{hattensorremark} A hat-tensor always exists for a HS
  tensor since we can simply define $\widehat{T}_{a_1, a_2,\cdots,a_n}$
  to be $\| T_{a_1, a_2,\cdots,a_n} \|$. We call this hat-tensor the minimal hat-tensor for $T$.

\end{remark}

We now have two types of tensors: the original tensors for which $V$ is infinite dimensional and the hat-tensors for which the index spaces are finite dimensional. When we want to emphasize that we are referring to the former, we will call them ``full-tensors.'' We will label indices for the full-tensors with letters such as $i_1,i_2,\ldots$. For the hat-tensors the possible values of the indices are the sectors for the corresponding leg. We will label indices for the hat-tensors with letters such as $a,b,c, \ldots$ or $a_1,a_2,\ldots$.

It is immediate that if $\widehat{A}$ and $\widehat{B}$ are hat-tensors for $A$ and $B$, then $\widehat{A} + \widehat{B}$ is a hat-tensor for $A+B$. We need to consider hat-tensors for contractions of tensors. Note that we will only contract legs of different tensors (perhaps even several legs simultaneously), never two legs of the same tensor ("no self-contraction").

\begin{lemma}\label{hatlemma}
  Let $T_1,T_2,\ldots,T_n$ be HS tensors and let $C$ be some contraction of
  these tensors without self-contractions.
  Let $\widehat{T}_i$ be hat-tensors for $T_i$. Then the contraction of
  $\widehat{T}_1, \widehat{T}_2, \ldots, \widehat{T}_n$ is a hat-tensor for $C$. (The $\widehat{T}_i$ are contracted in the same way as the $T_i$.)
\end{lemma}

\begin{proof}[Sketch of the proof:] Proposition 2.4 in \cite{paper2} is a partial case of this lemma with one sector per leg, when the hat-tensor is just a scalar bound for the full HS norm. Like that proposition, the lemma is a simple consequence of the Cauchy-Schwarz inequality.
  A simple induction argument reduces the proof to the case of $n=2$. For this case we group together the legs of $A_1$ that are not involved in the contraction, and group together the legs of $A_1$ that are involved. We do a similar leg grouping for $A_2$. This reduces the $n=2$ case to the case of two 2-leg tensors. This case can then be proved by several applications of the Cauchy-Schwarz inequality, paying attention to different sectors.
\end{proof}

As mentioned in \cref{sec:tensors}, we will also have some tensors which are not HS and so do not have
hat-tensors. To handle these we make the following definition.

\begin{definition}\label{def:check}
  Let $S$ be a tensor (not necessarily HS) with some non-empty subset of its
  legs distinguished. We say $\widecheck{S}$ is a check-tensor for $S$ with
  this subset of distinguished legs if for any HS tensor $T$ with hat-tensor
  $\widehat{T}$, the contraction of the distinguished legs of $S$ with
  any subset of the legs of $T$ produces a HS tensor, and a hat-tensor
  for this contraction may be obtained by contracting $\widecheck{S}$
  and $\widehat{T}$ in the same manner.
\end{definition}

\begin{remark} \label{checkdiagramremark} In our diagrams we will put a tick mark across the
	distinguished legs for tensors that have a check-tensor.
\end{remark}

\begin{remark} \label{checkremark} The idea behind this definition is that $S$ is a bounded operator from the (Hilbert space of) distinguished legs to the undistinguished legs, and $\widecheck{S}$ bounds the operator norm of restrictions of $S$ to various sectors. This generalizes Prop.~2.2(b) from \cite{paper2}.
\end{remark}

Using \cref{hatlemma} and \cref{def:check} repeatedly we can compute hat-tensors for a wide class of \emph{allowed} contractions defined as follows:

\begin{definition}\label{def:allowed} Let $n \ge 1$ and let
	$T_1,T_2,\ldots,T_n$ be HS tensors with hat-tensors $\widehat{T}_i$.
	Let $m \ge 0$ and let $S_1,S_2,\ldots S_m$ be tensors each of which
	is assigned a set of distinguished legs and has a check-tensor
	$\widecheck{S}_j$ for that set. A contraction $C$ of all these tensors is called \emph{allowed} if
	there are no self-contractions, and for each $S_j$ the set of distinguished legs is
	contracted with some subset of the legs of a single tensor $T_i$, while the undistinguished legs of $S_j$ may be either contracted with legs of tensors $T_{i'}$, $i'\ne i$, or left uncontracted.
\end{definition}

The following lemma will be used extensively.
\begin{lemma}\label{checklemma} Let $C$ be an allowed contraction from the previous definition.
  Then $C$ is HS and if we contract the $\widehat{T}_i$ and
  the   $\widecheck{S}_j$ in the same manner as the contraction that formed $C$,
  then we obtain a hat-tensor for $C$.
\end{lemma}

\begin{proof} Let $i$ be such that the distinguished legs
  of $S_1$ are contracted with some legs of $T_i$.
  Then by the definition of check-tensor, we may replace the contraction
  between $S_1$ and $T_i$ by a single HS tensor for which the contraction of
  $\widecheck{S}_1$ and $\widehat{T}_i$ is a hat-tensor. This reduces $m$ by
  $1$ and leaves $n$ unchanged. Continuing this process we are reduced to
  the case of $m=0$.  Then we can apply the previous lemma.
\end{proof}

\begin{remark} \label{checkhatremark} Whether one uses a hat-tensor or a check-tensor is sometimes a matter of necessity and sometimes of convenience. In particular, in this paper:
\begin{itemize}
	\item We use a hat-tensor for any tensor which is a contraction of $b$ tensors. This hat-tensor will be obtained from $\widehat{b}$ by Lemma \ref{hatlemma}.

	\item We use a check-tensor for any two-leg tensor $S$ which is an orthogonal projector on the direct sum of
	one or more sectors. For example, if $S$ is the projector on
	$\d \oplus \u \oplus \r$, then its check-tensor $\widecheck{S}_{aa}=1, \, a = \d, \u, \r$, and all other components equal to zero. We find it convenient to use a check-tensor even for projectors on finite-dimensional subspaces, even though these tensors are HS and a hat-tensor exists. For example, if $S$ is the projector on $\o$ then $\widecheck{S}_{\o\o}=1$ is the only nonzero component.
	\end{itemize}
\end{remark}

\subsection{Control of the 2x1 map via hat-tensors}\label{sec:control-of-the-2x1-map-via-hat-tensors}

Let us describe the main ideas of how we will control the 2x1 RG map on the nonlinear level. We denote by $\mathbb{H}_0$ the Hilbert space of complex tensors $b$ equipped with the HS norm and satisfying $b_{0000}=0$. The RG map acts on the tensor $A=A_*+b$, where $b\in \mathbb{H}_0$.

Let $\widehat b$ be a hat-tensor for $b$ (which may or may not be its minimal hat-tensor). It is a tensor with $2\times4\times 2\times 4=64$ real nonnegative components, one of which is zero: $b_{\o\o\o\o}=0$. We denote the set of such hat-tensors $\hat{\mathbb{H}}_0$. We will perform the four steps of the RG map one by one, and after every step we will estimate the arising new tensors by their hat-tensors, which will be expressed as functions of previous hat-tensors:
\begin{itemize}
	\item After the gauge transformation step we will obtain $A_1=A_*+b_1$ with a hat-tensor $\widehat{b}_1$ which is a function of $\widehat b$;
	\item After disentangling and splitting we'll get $L=A_*+b_L$, $R=b_R$ with hat-tensors $\widehat{b}_L$, $\widehat{b}_R$ which are functions of $\widehat{b}_1$;
	\item After reconnection we'll get $A_2=A_*+b_2$ with a hat-tensor $\widehat{b}_2$ a function of $\widehat{b}_L$ and $\widehat{b}_R$;
	\item And finally after rotation we have $A'=A_*+b'$ with $\widehat{b}'$ a function of  $\widehat{b}_2$.
\end{itemize}
All these functions will be given by explicit expressions which will be easily inferred from the diagrams defining the RG map, given below, using Lemma \ref{checklemma}. Composing the functions from the four steps of the 2x1 map, we get $\widehat{b'}$ as a function of $\widehat{b}$, which we call the master function:
\beq
\widehat{b}' = \mathfrak{M}(\widehat{b})\,.
\eeq
The master function acts from $\hat{\mathbb{H}}_0$ to itself. Because of the large dimensionality of this space, there is no question of computing it by hand. We realize it in a computer code accompanying the paper.

Our construction will define the 2x1 map for $b$ belonging to an open neighborhood $\Omega$ of $0\in \mathbb{H}_0$. Moreover $b'$ will be an analytic function of $b\in\Omega$. The neighborhood $\Omega$ will be defined in terms of a condition on a hat-tensor for $b$:
\beq
\label{eq:Omega}
\Omega = \{b: \exists\ \widehat{b}\text{ for }b\text{ such that } \widehat{b} \in \widehat\Omega\}\,.
\eeq
where $\widehat\Omega$ is a subset of $\hat{\mathbb{H}}_0$. For any given tensor $\widehat{b}$, it will be easy to check
whether $\widehat{b} \in \widehat{\Omega}$ or not.\footnote{The $\widehat{\Omega}$
will be downward closed (Def.~\ref{def:downward} below). So to check if $b\in \Omega$ it's enough to check if its minimal hat-tensor belongs to $\widehat\Omega$. Still, it's convenient to write the definition for $\Omega$ as above.}

Now suppose we are given a tensor $b$ whose hat-tensor $\widehat{b}$, specified numerically, belongs to $\widehat\Omega$, so the 2x1 map is defined. We apply the map and we get $b'$. The $b'$ will be defined explicitly by some diagrams in terms of $b$. Of course $b'$, like $b$, is an infinite-dimensional tensor, so we don't attempt to evaluate those diagrams. Instead we evaluate the master function and obtain a numerical hat-tensor $\widehat{b}'$, which gives us information about the size of $b'$. In particular we can check if $\widehat{b}'\in \widehat\Omega$. If this condition holds, this implies that $b'\in \Omega$, and we can iterate the 2x1 map. We so obtain a sequence of tensors $b=b^{(0)}, b'=b^{(1)}, b^{(2)}, b^{(3)}, \ldots$ (which are just defined, not evaluated), and the corresponding hat-tensors, which are all evaluated numerically.

We will find that when we iterate the map one of the following two possibilities is realized:
\begin{align}
  &\parbox{.8\linewidth}{After a certain number of iterations we get a hat-tensor $\widehat{b}^{(i_0)}\notin\widehat{\Omega}$.
  } \\[6pt]
  &\codetag \parbox{.8\linewidth}{The iterations enter a trajectory where all components $\widehat{b}^{(i)}$ start decreasing exponentially with $i$, in the sense that for some $i_0$ the condition $\widehat{b}^{(i_0+1)}\le \lambda \widehat{b}^{(i_0)}$ is satisfied componentwise with some $\lambda<1$.} \label{eq:bhatcase2}
\end{align}
In the first case we have to stop the iterations as the RG map is no longer defined. In the second case, as we will argue now, the exponential decrease will continue indefinitely.

The argument is based on some properties of functions of hat-tensors. We need a few definitions. If $x,y\in \bR_{\geq 0 }^n$ (e.g.~they are both hat-tensors), then we say $x\le y$ if this inequality holds componentwise. Used subsets of $\bR_{\geq 0 }^n$ will satisfy the following definition:
\begin{definition}\label{def:downward}
 A set $\widehat{\Lambda} \subset \bR_{\geq 0 }^n$ is called {\bf downward closed} if for any $x\in \widehat{\Lambda}$ and for any $y$ such that $0\le y\le x$ we have $y\in \widehat{\Lambda}$.
 \end{definition}
 The following properties will hold naturally for the functions of hat-tensors that we will consider.
\begin{definition}\label{def:mondandhom}
  Let $\widehat{\Lambda} \subset \bR_{\geq 0 }^n$ be downward closed. A map $f:\widehat{\Lambda} \to \bR_{\geq 0}^m$ is called {\bf monotonic} if the inequality $f(y) \le f(x)$ holds componentwise for all $y, x \in \widehat{\Lambda}$ such that $y \le x$ componentwise. The map is called {\bf subhomogeneous} if the inequality $f(\mu x)\le \mu f(x)$ holds componentwise for all $x\in \widehat{\Lambda}$ and all $0 \leq \mu \le 1$.
\end{definition}
In the course of \cref{sec:gauge,sec:disent,sec:reconnection-and-rotation,sec:masterf}, we will present explicit definitions of $\mathfrak{M}$ and $\widehat{\Omega}$, together with a proof that $\mathfrak{M}$ is monotonic and subhomogeneous on $\widehat{\Omega}$. The following key lemma then implies the convergence of $\widehat{b}^{(i)}$ satisfying \eqref{eq:bhatcase2}.
\begin{lemma}[Key lemma]\label{lem:key}\codetag
  Let $\widehat{b}^{(i)}$ be the sequence of hat-tensors generated by iterating the master function, assumed monotonic and subhomogeneous. Suppose \eqref{eq:bhatcase2} is realized. Then for all $i > i_0$,
  \begin{equation}
    \widehat{b}^{(i)} \leq \lambda^{i - i_0} \widehat{b}^{(i_0)},
  \end{equation}
  and thus $\widehat{b}^{(i)}$ converges to zero exponentially fast.
\end{lemma}

\begin{proof}
  We proceed by induction on $i$. For $i=i_0+1$ the inequality holds by \eqref{eq:bhatcase2}. The induction step $i\to i+1$ is a consequence of monotonicity and subhomogeneity:
  \begin{align}
    \widehat{b}^{(i+1)} & = \mathfrak{M}(\widehat{b}^{(i)}) \leq \mathfrak{M}(\lambda^{i- i_0} \widehat{b}^{(i_0)})    &  & \text{(monotonicity)} \\
                    & \leq \lambda^{i - i_0} \mathfrak{M}(\widehat{b}^{(i_0)})                             &  & \text{(subhomogeneity)} \\
                    & \equiv \lambda^{i - i_0} \widehat{b}^{(i_0+1)} \leq \lambda^{i+1 - i_0} \widehat{b}^{(i_0)},
  \end{align}
  completing the induction.
\end{proof}

The following simple facts about monotonic, subhomogeneous maps will be useful below.
\begin{lemma} \label{lem:monsub}
	\begin{enumerate}
		 \item[(a)]  Any monomial of degree $\ge 1$ in components of $\widehat{b}$ is monotonic and subhomogeneous.
		 \item[(b)] The sum of two monotonic, subhomogeneous maps is monotonic and subhomogeneous.
		 \item[(c)] The composition of two monotonic, subhomogeneous maps is monotonic and subhomogeneous.
		 \item[(d)] The product of a nonnegative monotonic, subhomogeneous map and a nonnegative monotonic real-valued function is monotonic and subhomogeneous.
		 \end{enumerate}
		 \end{lemma}
We omit the proof. We now proceed to the construction of the 2x1 map on the nonlinear level, and of the associated master function.

\subsection{Gauge transformation}
\label{sec:gauge}
We repeat the analysis of \cref{sec:lin-gauge-transformation} at the full nonlinear level. We keep the definitions of $g_h$ and $g_v$ given in \eqref{eq:ghgv}. In addition, we define the scalars
\begin{equation}
	\beta_h = \myinclude[scale=0.9]{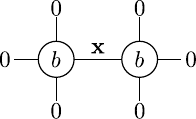}, \quad
	\beta_v = \myinclude[scale=0.9]{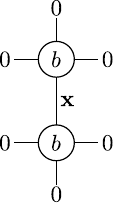}.
\end{equation}
It is easily checked that with $\beta = \beta_h, \beta_v$
\begin{equation}
	g^3 = - \beta g.
	\label{gcubed}
\end{equation}
At the nonlinear level, Eqs.~\eqref{eq:Glin} for $G_h,G_v$ are replaced by the exact equalities:
\begin{equation}
	G = \mathds{1}_V + g + \frac{1}{1+\sqrt{1-\beta}} \, g^2, \quad
	G^{-1} = \mathds{1}_V - g + \frac{1}{1+\sqrt{1-\beta}} \, g^2. \quad
	\label{gauge-trans}
\end{equation}
The first equation is the definition of $G$. The second equation is easily checked using \eqref{gcubed}. We will be assuming that $|\beta|<1$, and so $G,G^{-1}$ are analytic functions of $b$.

Let us briefly discuss how this is related to \cite{paper1}. There, the gauge transformation was done with
$G_{\rm there} = \exp(g)$,
while our definition corresponds to taking $G=\exp(c g)$ with $c=\arcsin(\sqrt{\beta})/\sqrt{\beta}=1+O(\beta)$. The two expressions agree to linear order in $g$ and give the same needed cancellations. We will use Eq.~\eqref{gauge-trans} as it is simpler for the quantitative analysis of higher-order terms.

We will apply $G_h,G_v$ and their inverses to the horizontal, vertical legs of $A$, respectively, as in Fig.~\ref{fig:bigfig}(a). We denote the resulting tensor by $A_{g}$. Partition function is preserved (\cite{paper2}, Sec.2.4.3):
\beq
\boxed{Z(A,\ell_x\times \ell_y)=Z(A_{g},\ell_x\times \ell_y)}\,.\label{eq:Zgauge}
\eeq
Note that $A_{g}$ is not normalized; we will worry about normalization at the end. As discussed in \cref{sec:lin-gauge-transformation}, if we expand $A_{g}$ in $g$, the first order in $g$ part will cancel the sectors of $b$ given in \eqref{eq:set0}. We now discuss the main improvement with respect to \cite{paper1}: a graphical language which allows to make this first-order cancellation explicit, and to seamlessly express the bounds on the remainders. The same language will be later used for disentangling, in a more complicated setting.

We start by writing each of $G_h,G_v$ as sums: $G_h = \sum_{i=1}^5 h_i$, $G_v = \sum_{i=1}^5 v_i$. The tensors $h_i$, $v_i$, are given in \cref{fig:vtens,fig:htens}; they represent various terms in \eqref{gauge-trans}. Note that we find it convenient to split the identity and $g$ as a sum of two terms each. Furthermore we define:
\begin{equation}
  h_i' = h_i, \, v_i = v_i' \,\, {\rm for} \,\, i=1,2,5, \quad \quad
  h_i' = -h_i, \, v_i = -v_i' \,\, {\rm for} \,\, i=3,4.
\end{equation}
Then $G_h^{-1} = \sum_{i=1}^5 h_i'$, and $G_v^{-1} = \sum_{i=1}^5 v_i'$.

\begin{figure}[H]
	\centering
	\codetag\includegraphics[scale=0.8]{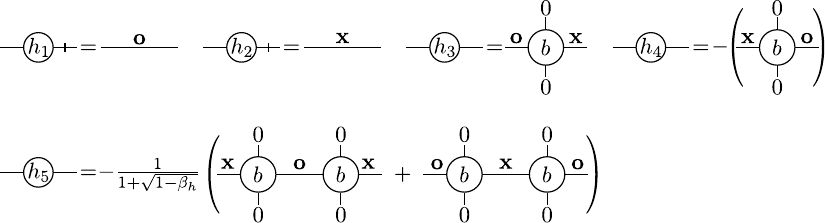}
	\caption{Definition of $h_i$ tensors. The lines labeled $\o$ and $\x$ in the definitions of $h_1$ and $h_2$ represent orthogonal projectors onto the corresponding subspaces. The right legs of $h_1,h_2$ are distinguished, hence the tick mark.}
	\label{fig:htens}
\end{figure}

\begin{figure}[H]
  \centering
  \codetag\includegraphics[scale=0.8]{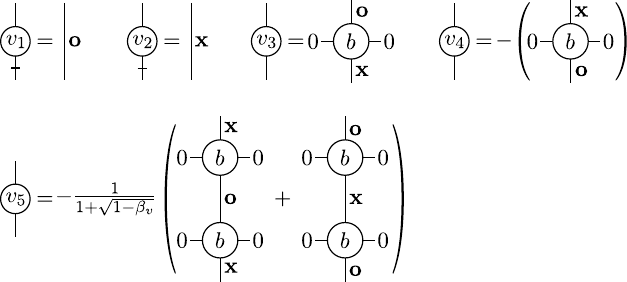}
  \caption{Definition of $v_i$ tensors. The bottom leg of $v_1,v_2$ is distinguished.}
  \label{fig:vtens}
\end{figure}

Now $A_{g}$ is given by the sum over $i,j,k,l$ ranging from $1$ to $5$ of the following diagrams:
\begin{equation}
	\label{eq:Ag0}
  A_{g} = \sum_{ijkl} D^{(ijkl)},\qquad D^{(ijkl)} :=\myinclude{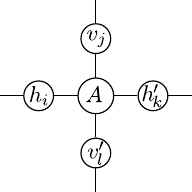}\,.\codetag
\end{equation}
As mentioned, there is some cancellation in $A_{g}$: terms first-order in $g$ cancel a part of $b$. We make this explicit by defining a set of index combinations:
\begin{equation}
  \mathcal{C} = \{2111, 4111, 1211, 1411, 1121, 1131, 1112, 1113 \}. \codetag
\end{equation}
The cancellation is now phrased by saying that the part of the sum in \eqref{eq:Ag0} over $ijkl\in \mathcal{C}$ is zero. For example the 2111 diagram picks up $b_{\x \o \o \o}$ while the 4111 diagram gives $h_4$ contracted with $A_*$ (recall that $b_{0000}=0$) which is exactly $-b_{\x \o \o \o}$. Similarly the other three pairs of diagrams cancel pairwise. We can therefore drop the canceled terms and restrict the sum (\ref{eq:Ag0}) to $ijkl\notin \mathcal{C}$.

We now proceed to bounding $A_{g}$ i.e.~finding its hat-tensor $\widehat{A}_g$. It will be expressed in terms of a hat-tensor for $A$, which we write in the form
\beq
\widehat A = \widehat A_*+\widehat b \codetag
\eeq
where $(\widehat{A}_*)_{\o \o \o \o}=1$ and all other components are zero, while $(\widehat b)_{\o \o \o \o}=0$.  We define hat or check-tensors for all other tensors in the diagrams $D^{(ijkl)}$:
\begin{itemize}
\item
 For $i=3,4,5$ the tensors $h_i, h_i'$,$v_i, v_i'$ are Hilbert-Schmidt. We define their hat-tensors from their expressions in Figs.~\ref{fig:htens},\ref{fig:vtens}. Using Lemma \ref{hatlemma}, these hat-tensors are expressed in terms of (contractions of) various components of $\widehat b$. Furthermore, the coefficients $1/(1+\sqrt{1-\beta})$ that appear in $h_5, h'_5, v_5, v'_5$ need to be replaced by appropriate upper bounds on their absolute values in terms of $\widehat{b}$. We have $|1/(1+\sqrt{1-\beta})|\le 1/(1+\sqrt{1-|\beta|})$ for $|\beta|<1$, and the r.h.s.~is an increasing function of $|\beta|$. So the needed upper bounds are obtained by replacing $\beta_h, \beta_v$ by upper bounds  $\widehat{\beta}_h,\widehat{\beta}_v$ on the absolute values of $\beta_h, \beta_v$, given by\footnote{Recall that $\x=\d\oplus\u\oplus \r$ on the horizontal legs; see Sec.~\ref{sec:sectors}. Contraction over such $\x$ implies a sum over $\u,\d, \r$. E.g.~$\widehat{b}_{\o \o \x \o} \ \widehat{b}_{\x \o \o \o}$ stands for $\sum_{a =\u,\d, \r}\widehat{b}_{\o \o a \o} \ \widehat{b}_{a\o \o \o}$.}
 \begin{equation}
 	\widehat{\beta}_h = \widehat{b}_{\o \o \x \o} \ \widehat{b}_{\x \o \o \o}, \quad \quad
 	\widehat{\beta}_v = \widehat{b}_{\o \o \o \x} \ \widehat{b}_{\o \x \o \o}\,, \codetag
 \end{equation}
 as long as $\widehat \beta_h<1$ and $\widehat \beta_v<1$, which we assume to be the case.
 \item
 Tensors $h_i, h'_i, v_i, v'_i$, $i=1,2$ are orthogonal projectors, so we define their check-tensors (Remark \ref{checkhatremark}).
 We declare the legs that are contracted with $A$ to be
 distinguished, denoted with a tick mark in Figs.~\ref{fig:htens},\ref{fig:vtens} for $h_1,v_1,h_2,v_2$.

On the horizontal legs $V=\o\oplus \d\oplus\u\oplus \r$. The check-tensors $\widecheck{h}_2, \widecheck{h}'_2$ are $4\times 4$ matrices with the $\u\u$, $\d\d$, $\r\r$ elements equal to 1, and all other elements 0. The check-tensors $\widecheck{h}_1, \widecheck{h}'_1$ are $4\times 4$ matrices with a single nonzero element $\o\o$ equal to 1.

 On the vertical legs $V=\o\oplus \x$. The check-tensors $\widecheck{v}_2$ $\widecheck{v}'_2$ are $2\times 2$ matrices with a single nonzero element $\x\x$ equal to 1. The check-tensors $\widecheck{v}_1$ $\widecheck{v}'_1$ are $2\times 2$ matrices with a single nonzero element $\o\o$ equal to 1.
\end{itemize}
With the given definition of distinguished legs, all diagrams $D^{(ijkl)}$ are allowed contractions (Def.~\ref{def:allowed}). So by Lemma \ref{checklemma} we obtain the hat-tensor for each diagram by replacing every tensor in the diagram by its hat or check-tensor, as defined above, and contracting these tensors. Summing the hat-tensors $\widehat{D}^{(ijkl)}$ over $ijkl\notin\mathcal{C}$, we obtain the hat-tensor $\widehat{A}_g$.

Here is a summary of our new graphical language. \emph{We represent a full-tensor as a sum of diagrams. We render cancellations explicit by identifying and removing the set of diagrams which sum to zero. After that a hat-tensor is obtained by simply putting hats, or checks, on all the full-tensors involved. If there are coefficients in the diagrammatic expression for the full-tensor, we also need to replace these coefficients by bounds on their absolute values.}

There is one final step in our gauge transformation. The gauge transformations can change the $0000$ component of our tensor by a small amount. To restore our normalization condition that the $0000$ component is $1$, we write
$A_{g}=\cN_1 A_1$,
where $\cN_1 = (A_{g})_{0000}$ and $A_1$ is normalized, $A_{1}=A_*+b_1$. Factoring out $\calN_1$ from the partition function, we can rewrite \eqref{eq:Zgauge} as:
\beq
\boxed{Z(A,\ell_x\times \ell_y)=\calN_1^{\rm Vol} Z(A_1,\ell_x\times \ell_y)\,,\quad{\rm Vol}=\ell_x\ell_y}\,.\label{eq:Zgauge1}
\eeq

We now turn to finding a hat-tensor $\widehat{A}_{1}$, which we write as
\beq
\widehat A_{1}=\widehat A_*+\widehat b_1 \codetag
\eeq
where $(\widehat{A}_*)_{\o\o\o\o}=1$ and has all other components are zero, while $(\widehat b_1)_{\o\o\o\o}=0$. We also write $\widehat{A}_g$ in a similar form, separating the $\o\o\o\o$ component:
\begin{equation}
	\widehat{A}_g =\widehat{\calN}_1 \, \widehat{A}_* + \widehat{b}_g,\qquad \widehat{\calN}_1:=(\widehat{A}_g)_{\o \o \o \o}\,. \codetag
\end{equation}
Since $A_1=A_{g}/\calN_1$, we can set $\widehat b_1=\widehat{b}_g/|\calN_1|_{plb}$ where  $|\calN_1|_{plb}$ is a \emph{positive lower} bound for $|\calN_1|$. Note that until now we only considered upper bounds on absolute values (or norms); e.g.~$\widehat{\calN}_1$ is an upper bound on $|\calN_1|$. However a positive lower bound on $|\calN_1|$ can be inferred without additional work, by the following argument. Our construction represents $\calN_1=1+ n_1$ where $n_1$ is a sum of various contractions of $b$ times various factors. The upper bound has the form $\widehat{\calN}_1=1+\widehat n_1$, where $\widehat n_1$ is an upper bound on $|n_1|$. Therefore, the needed positive lower bound has the form $|\calN_1|_{plb}=1 -\widehat{n}_1= 2-\widehat{\calN}_1$. (It will be positive since we will impose $\widehat{\calN}_1<2$.)

To summarize, the considerations of this section define the following functions:
\begin{align}
	&\widehat{\beta}_h,\widehat{\beta}_v:\hat{\mathbb{H}}_0\to \mathbb{R}_{\ge 0}, \\
	&\widehat{\calN}_1:\{ \hat b\in \hat{\mathbb{H}}_0: \widehat{\beta}_h,\widehat{\beta}_v<1\}\to \mathbb{R}_{\ge 0},\\
	& \widehat b_1:\hat{\Omega}_1\to \hat{\mathbb{H}}_0,\quad \widehat\Omega_1 := \{ \widehat{b}\in \hat{\mathbb{H}}_0: \widehat{\beta}_h < 1,\ \widehat{\beta}_v < 1,\ \widehat{\mathcal{N}}_1 < 2 \},\codetag \\
	&\calN_1:\Omega_1\to \mathbb{C},\qquad b_1:\Omega_1\to \mathbb{H}_0,\qquad \Omega_1= \{b \in \mathbb{H}_0: \exists\ \widehat{b}\text{ for }b\text{ such that } \widehat{b} \in \widehat\Omega_1\}.
\end{align}

\begin{proposition}\label{prop:summary-gauge}
(a) The sets on which $\widehat{\beta}_h$, $\widehat{\beta}_v$, $\widehat{\calN}_1$, $\widehat{b}_1$ are defined are downward closed. $\widehat{\beta}_h$, $\widehat{\beta}_v$, $\widehat b$ are monotonic and subhomogeneous. $\widehat{\calN}_1$ is monotonic.\\[5pt]
\noindent (b)  The functions $\calN_1$, $b_1$ are analytic on $\Omega_1$. \\[5pt]
\noindent (c) We have $b_1(0)=0$, $\calN_1(b)=1+O(b^2)$.\\[5pt]
\noindent (d) If $\widehat{b}$ is a hat-tensor for $b$, then $\widehat{b}_1(\widehat b)$ is a hat-tensor for $b_1(b)$. We have $\calN_1\ne 0$ on $\Omega_1$, and the bound:
	\beq
	|\calN_1(b)-1|\le \widehat{\calN}_1(\widehat b)-1\,.
	\eeq
\end{proposition}

\begin{proof}
	(a) follows with the help of Lemma \ref{lem:monsub}. Note that $1/(1+\sqrt{1-\widehat\beta_h})$, $1/(1+\sqrt{1-\widehat\beta_v})$ are nonnegative monotonic real functions of $\hat b$ as long as $\hat\beta_h,\hat\beta_v<1$, so that Lemma \ref{lem:monsub}(d) can be applied. (b) follows from the fact that on $\Omega_1$ we have $|\beta_h|,|\beta_v|<1$, $|\calN_1|>0$, so that $1/(1+\sqrt{1-\beta_h})$, $1/(1+\sqrt{1-\beta_v})$, $1/\calN_1$ are analytic. Multiplying by these factors were the only non-polynomial operations in the definitions of $\calN_1$, $b_1$. (c) follows from definitions. (d) is just a restatement of things already discussed.
\end{proof}

\subsection{Disentangling and splitting}
\label{sec:disent}
Here we will repeat the linearized discussion in \cref{sec:lin-disentangling} at the fully nonlinear level. Recall that disentangling and splitting is applied to the tensor $A_1=A_*+b_1$ which is the end result of the gauge transformation step. As in \cref{sec:lin-disentangling}, we rename $A_1$ as $A$ and $b_1$ as $b$. This will avoid writing index 1 on many tensors below. We will revert to $A_1$ and $b_1$ in \cref{sec:hat-tensor-estimates}.

Disentangling and splitting is represented in \cref{eq:AADD2LR}, which is a copy of Fig.~\ref{fig:bigfig}(b) up to the described renaming. So, we divide the lattice into horizontal $2 \times 1$ blocks and contract the two $A$'s in each block. Our disentangler $D$ will be a tensor with two legs on top and two on the bottom. We think of it as a matrix acting from the space corresponding to the two bottom legs to the space corresponding to the top two legs. In our network we insert $I= D^{-1} D$ on each pair of vertical legs between two blocks. We think of the $D$ above a block and the $D^{-1}$ below the same block as being in that block. The resulting block is shown in the left side of  \cref{eq:AADD2LR}. (The right side of the figure is for later use.)

We want to choose the disentangler $D$ to disentangle the contraction of two $A$'s in the block. The major part of this contraction comes from the 0 index, but $\x$ indices also contribute. Crudely speaking, disentangling means reducing the contribution of $\x$ indices to the contraction. Concretely, we will construct $D$ to cancel the diagrams shown in \cref{eq:aa-prob}. In this figure the contracted legs are never in $\o$, so $\AHT$ does not contribute and so we have replaced $A$ by $b$. The reason for cancelling these particular diagrams was explained in the linearized analysis. A discussion of the type of disentangler we are using can be found in Section 2.4.4 of \cite{paper2}. In particular Remark 2.6 discusses the origin of the name disentangler.\footnote{Ref.~\cite{Evenbly-Vidal} introduced disentanglers into numerical algorithms of tensor network RG. Ref.~\cite{paper1} first used disentanglers in a rigorous tensor RG study. For a closely related use of disentanglers in quantum lattice models see \cite{Vidal:2007hda}.}

If we group the top two legs together and the bottom two legs together in \cref{eq:aa-prob}, then the problem of cancelling these diagrams is analogous to how we cancelled certain diagrams in the original $A$ with a gauge transformation. So our disentangler $D$ is quite similar to the gauge transformation $G$ defined in Eq.~\eqref{gauge-trans}. It is defined by (cf. \eqref{eq:Dlin}):
\begin{equation}
	\label{eq:Ddef}
  D = \mathds{1}_{V\otimes V} + X + \frac{1}{1 + \sqrt{1 - \alpha}} \, X^2, \qquad D^{-1} = \mathds{1}_{V\otimes V} - X + \frac{1}{1 + \sqrt{1 - \alpha}} \, X^2,
\end{equation}
where $X$ is given in \cref{eq:Xdeflin}, and
\begin{equation}
  \alpha=\myinclude[scale=0.9]{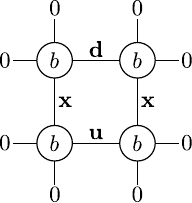}.
\end{equation}
Below we will make sure that $|\alpha|<1$, so that \eqref{eq:Ddef} are well defined and analytic in $b$. Here, $X$ should be considered as a matrix acting from the space of its bottom legs to the space of its top legs. Similarly to the gauge transformation case, the formula for $D^{-1}$ works because $X^3 = -\alpha X$. If we compute the contraction of $D$, $D^{-1}$ and the two $A$'s in \cref{eq:AADD2LR} to first order in $X$, we see that the diagrams from \cref{eq:aa-prob} are cancelled.

\begin{figure}[h]
  \centering
  \codetag\includegraphics[scale=0.8]{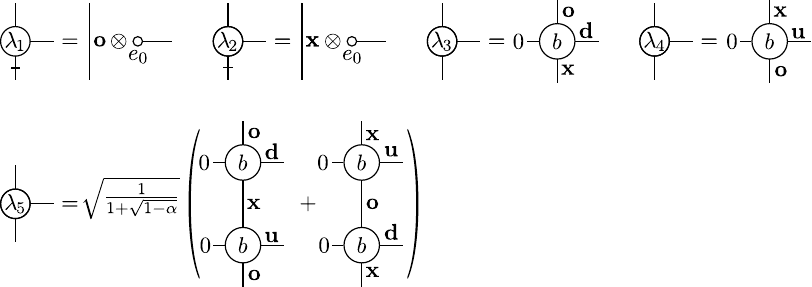}
  \caption{Definition of $\lambda_i$ tensors. Here, $\lambda_1$ and $\lambda_2$ are Kronecker products of the shown projectors with $e_0$, the basis vector of $\o$, a 1-tensor denoted by the circle with label $e_0$.}
  \label{fig:lambdas}
\end{figure}

\subsubsection{Decomposition of disentangler}\label{sec:decomposition-of-disentangler}
The next step is to decompose the disentangler $D$ into a sum of terms. This will be extremely helpful to accomplish two goals. One goal is to represent the contraction of two $A$'s with $D$ and $D^{-1}$ as the contraction of two new tensors $L$ and $R$ as shown in \cref{eq:AADD2LR}. We refer to this step as ``splitting.''
The second goal is to explicitly realize the discussed cancellation (see Sec.~\ref{sec:the-oo-channel} below). The l.h.s.~of \cref{eq:AADD2LR} will become a sum of terms, some of which will cancel each other exactly.

So we introduce $(\lambda_i)_{i=1}^5$ tensors in \cref{fig:lambdas}. The right leg of these tensors lives in the Hilbert space $V_1:=V\oplus (V\otimes V)$. The sectors of $V_1$ have the form $a$ or $a\otimes b$ where $a,b$ are sectors of $V$. All tensor elements not shown in \cref{fig:lambdas} are assumed zero. In particular, tensors $\lambda_i$ with $i=1,2,3,4$ vanish when the right leg is in the $V\otimes V$ subspace of $V_1$, while $\lambda_5$ vanishes when the right leg is in the $V$ subspace of $V_1$.

We also define $\rho_i$ tensors in an analogous way but with all sector labels reflected with respect to a vertical line and tensors $\rho_4, \rho_5$ {\bf acquiring minus signs} (this is needed to reproduce minus signs of some diagrams in $X$ and $X^2$). With these definitions it's easy to check that
\begin{equation}
	\label{eq:Ddec}
	\myinclude{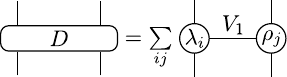},
\end{equation}
where we put $V_1$ label on the contracted leg to emphasize that it lives in a different Hilbert space.

Some comments are in order:
\begin{itemize}
	\item
Since $\lambda_5$ and $\rho_5$ are the only ones which have horizontal legs in the $V\otimes V$ part of $V_1$, they contract to each other but not with any other tensor. Their contraction reproduces the $X^2$ term in \eqref{eq:Ddef} including the prefactor.
\item
Since $\lambda_i,\rho_j$ with $i,j=1,2$ are the only ones which have $e_0$ horizontal legs, they only contract among themselves. Their contraction reproduces exactly the $\mathds{1}_{V\otimes V}$ term in $D$.
\item
Finally, $\lambda_3$ contracts only with $\rho_3$, and $\lambda_4$ only with $\rho_4$. These contractions reproduce the $X$ term in \eqref{eq:Ddef}.
\end{itemize}

We also define $\lambda'_i$ and $\rho'_i$ in an analogous way so that:
\begin{equation}
  \myinclude{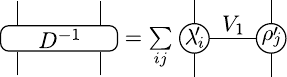}.
\end{equation}
They only differ from $\lambda_i$ and $\rho_i$ by some signs to account for the opposite sign of $X$ in $D^{-1}$; we skip the exact definition.

Substituting these expressions for $D$ and $D^{-1}$ into the left side of \cref{eq:AADD2LR} we get:
\begin{equation}\label{eq:dis-as-sum}
  \myinclude{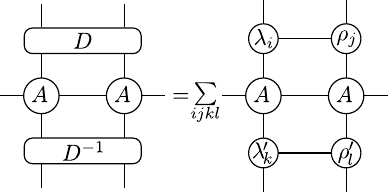}\codetag
\end{equation}

\subsubsection{Channels}

We now proceed to the splitting, i.e.~to defining the tensors $L$, $R$ in the r.h.s.~of \eqref{eq:AADD2LR}.
  Crudely speaking, splitting will be performed by cutting each diagram in \eqref{eq:dis-as-sum} across the three horizontal bonds and putting the left side into $L$ and the right side into $R$.
We represented the disentangler $D$ as a sum of contractions of $\lambda_i$ and $\rho_j$ so that we could perform this cutting.
Further tricks will accompany the cutting, whose role will be explained in due course.

As the first step we write the r.h.s.~of \cref{eq:dis-as-sum} as
$\sum_q H_q$ with the five terms $H_q$, $q\in\{z, \x \x,\o \x,\x \o,\o \o\}$ defined in
Fig.~\ref{fig:channels}, called channels. For the $\x \x,\o \x,\x \o,\o \o$, the name refers to the sectors on the horizontal legs. The $z$ ("zero") channel is different in that it is the horizontal contracted bonds which are restricted to $\o$. These channels are the nonlinear counterparts of the five channels with the same names considered in the linearized analysis in \cref{sec:l-r-and-the-norm-of-kss}.

\begin{figure}[h]
	\centering
	\includegraphics[scale=0.8]{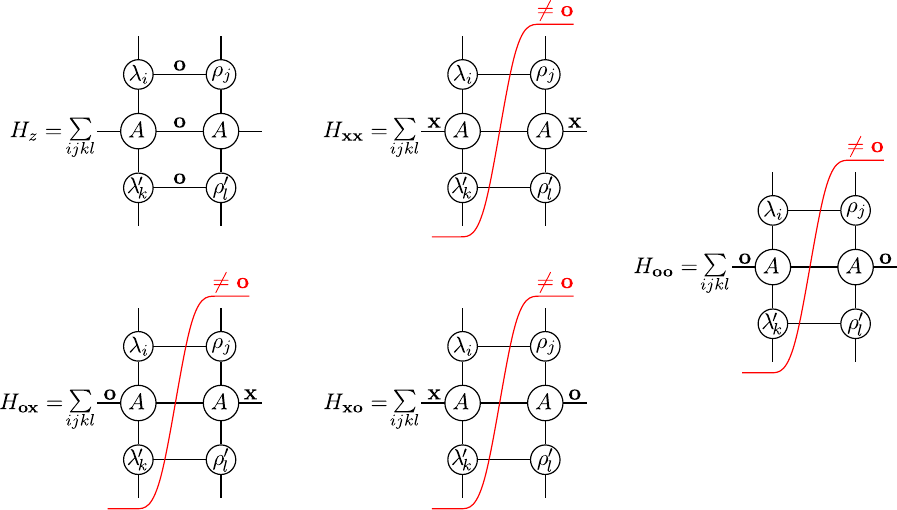}
	\caption{ Definition of channels. A line cutting through the bonds or legs with ``$\neq \o$'' label means that we forbid all these bonds or legs to be simultaneously in the $\o$ sector. In this particular case, this is the same as inserting there a projector to the orthogonal complement of $\o \otimes \o \otimes \o$ in $V_1\otimes V\otimes V_1$.
	}
	\label{fig:channels}
\end{figure}

We write each channel $H_q$ in the natural way as a contraction of two tensors denoted $L^{(0)}_q$ and $R^{(0)}_q$:\footnote{Extra manipulations will be needed before $L^{(0)}_q$ and $R^{(0)}_q$ are assembled into $L$ and $R$, hence the superscript $(0)$.}
\beq\label{eq:Hq}
\myinclude[scale=0.9]{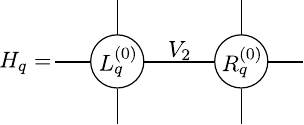}
\eeq
\begin{figure}[h]
	\centering
	\codetag\includegraphics[scale=0.8]{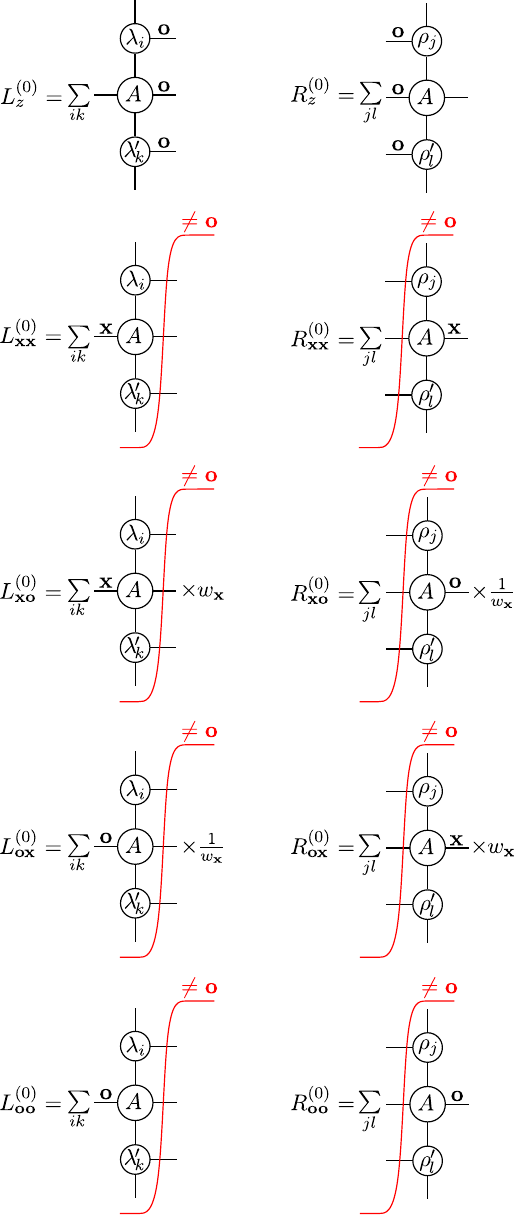}
	\caption{The definition of $L^{(0)}_q$, $R^{(0)}_q$ tensors. }
	\label{fig:LR}
\end{figure}
The involved steps are as follows (see Fig.~\ref{fig:LR} for the explicit definitions):
\begin{itemize}
	\item
	  We cut the three horizontal bonds in $H_q$. The cut bonds become horizontal legs of $L^{(0)}_q$ and $R^{(0)}_q$, living in $V_2 := V_1\otimes V\otimes V_1$. (We assume the ordering convention for the triple of bonds top, middle, bottom.) We emphasized this by $V_2$ label in \cref{eq:Hq}.
	\item
	We factor $\sum_{i,j,k,l=1}^5=\sum_{i,k=1}^5\sum_{j,l=1}^5$ in the definition of $H_q$, with $ik$ the indices of $\lambda,\lambda'$ and $jl$ the indices of $\rho,\rho'$. The sum over $ik$ goes into $L^{(0)}_q$, the sum over $jl$ into $R^{(0)}_q$.
	\item
	Furthermore, for the $\o \x$ and $\x \o$ channels we redistribute the weight
	of the diagram by the introduction of a parameter $w_\x>0$
        which multiples one of the factors $L^{(0)}_q$, $R^{(0)}_q$ and divides the other. We saw in \cref{sec:l-r-and-the-norm-of-kss} why this parameter is useful.
\end{itemize}

\subsubsection{Auxiliary space}\label{sec:auxiliary-space}

So we represented $H_q$ as a contraction of $L^{(0)}_q$ and $R^{(0)}_q$ for each $q$. We now wish to write the r.h.s. of \cref{eq:dis-as-sum}, which equals $\sum_q H_q$, as a contraction of $\sum_q L_q$ with $\sum_{q'} R_{q'}$. We could try to take $L_q \stackrel{?}{=} L^{(0)}_q$, $R_{q} \stackrel{?}{=} R^{(0)}_{q}$, but this would not work.
The diagonal terms $q=q'$ produce the five channels (good), but the cross terms $q \neq q'$
produce terms which do not exist in $\sum_q H_q$ (bad).

{ To avoid cross terms, we use the method of auxiliary vectors. We use the same auxiliary Hilbert space $V_{\rm aux}$ as in the linearized analysis, Eq.~\eqref{eq:Vaux0}.
The seven auxiliary vectors will be used for the same channels as there.}

For $q= z, \x \x, \o \x, \x \o$ (the $\o\o$ channel is considered separately below), we define
\beq
L_q = L^{(0)}_q \otimes e_q = \myinclude[scale=0.9]{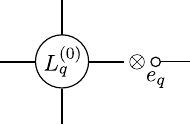},\qquad R_q = R^{(0)}_q \otimes e_q = \myinclude[scale=0.9]{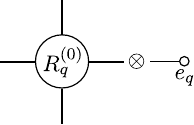}.
\label{eq:LRnot00}
\eeq
So, we think of $L_q$ (resp. $R_q$) as a 4-tensor, whose left (resp. right) leg lives in the Hilbert space
\beq
V_{3} := V_2\otimes V_{\rm aux} = V_1\otimes V\otimes V_1 \otimes V_{\rm aux}.
\label{eq:V3}
\eeq
Tensors $L_q$ and $R_q$ so defined have the same contraction as $L^{(0)}_q$ and $R_q^{(0)}$, i.e.~$H_q$. Furthermore, contractions between $L_q$ and $R_{q'}$, $q'\ne q$, vanish. Therefore, there are no unwanted cross terms, and we have ($\stackrel{V_3}{\text{\bf ---}}$ denotes contraction over $V_3$):
\beq
H_z + H_{\x\x} + H_{\o\x} + H_{\x\o} = (L_z + L_{\x\x} +L_{\o\x} + L_{\x\o})\stackrel{V_3}{\text{\bf ---}}(R_z + R_{\x\x} + R_{\o\x} + R_{\x\o}).
\label{eq:not00}
\eeq

\subsubsection{The $\o\o$ channel}\label{sec:the-oo-channel}
Just as in the linearized analysis, this needs a special treatment as this is the channel where we have exact
cancellations resulting from our choice of the disentangler $D$, see \eqref{eq:aa-prob} and \eqref{eq:Xdeflin}.

As a first step, we rewrite $L^{(0)}_{\o \o}$ and $R^{(0)}_{\o \o}$ in Fig.~\ref{fig:LR}, by splitting the middle horizontal legs into sectors $a$, summed over $\o, \d, \u, \r$:
\beq\label{eq:L00def}
\myinclude[scale=0.9]{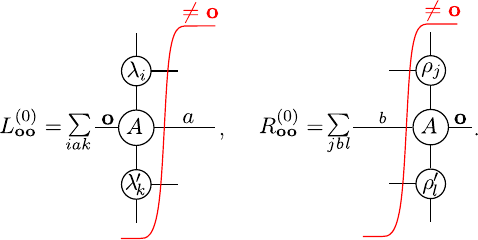} \codetag
\eeq
Now define the following set of groups of indices:
\begin{equation}
  \mathcal{E} = \{(2,\u,1), (1,\d,2), (4,\o,1), (1,\o,3)\}\,. \codetag
  \label{eq:Edef}
\end{equation}
Here is a crucial observation: In the contraction of $L_{\o \o}$ with $R_{\o \o}$ defined by \eqref{eq:L00def}, if we restrict the sums
to those terms with $(i,a,k) \in \mathcal{E}$ and $(j,b,l) \in \mathcal{E}$, then this partial sum
is zero. (The first two elements of $\mathcal{E}$ reproduce the two diagrams in \eqref{eq:aa-prob} and the second two elements of $\mathcal{E}$ reproduce the same diagrams with a minus sign.) This is the explicit realization of the cancellation mentioned at the beginning of Sec.~\ref{sec:decomposition-of-disentangler}.

In light of this observation let us split $L_{\o\o}^{(0)}=L_{\o\o}^{(1)}+L_{\o\o}^{(2)}$ where the two parts correspond to restricting the summation in \eqref{eq:L00def} to $(i,k,a) \notin \mathcal{E}$ (resp.~$(i,k,a) \in \mathcal{E}$):
\begin{equation}
	\label{eq:concl2}
  \raisebox{-2pt}{$L_{\o\o}^{(0)} = L_{\o\o}^{(1)}+L_{\o\o}^{(2)}$},\qquad\myinclude[scale=0.9]{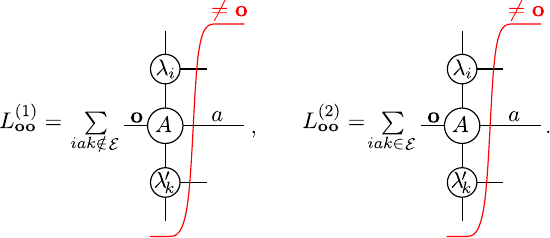}\codetag\\
\end{equation}
We also split $R_{\o\o}^{(0)} = R_{\o\o}^{(1)}+R_{\o\o}^{(2)}$ in the same fashion.

Recall that the contraction of $L_{\o\o}^{(0)}$ with
$R_{\o\o}^{(0)}$ produces the zero channel. Note that the
contraction of $L_{\o\o}^{(2)}$ with $R_{\o\o}^{(2)}$ is zero since this
contraction is precisely the sum over $\mathcal{E}$ which contains the terms that cancel.
Dropping this vanishing contraction, we are left with the decomposition of $\o\o$ into the sum of three
contractions over $V_2$:
\beq
H_{\o\o}
=L_{\o\o}^{(1)}\stackrel{V_2}{\text{\bf ---}} R_{\o\o}^{(1)}\ +\ L_{\o\o}^{(2)}\stackrel{V_2}{\text{\bf ---}} R_{\o\o}^{(1)}\ +\ L_{\o\o}^{(1)}\stackrel{V_2}{\text{\bf ---}} R_{\o\o}^{(2)}\,.
\label{eq:3terms}
\eeq
Given this equation, we now define $L_{\o \o}$ and $R_{\o \o}$ so that
\beq
H_{\o\o}= L_{\o\o}\stackrel{V_3}{\text{\bf ---}} R_{\o\o}.
\label{eq:needH00}
\eeq
To achieve this we use the auxiliary vector trick from Sec.~\ref{sec:auxiliary-space}. We define (For the first reading replace $w_{\o}$ and $P_Y$ by 1; see immediately below for the explanation of these factors.)
\beq
\myinclude[scale=0.9]{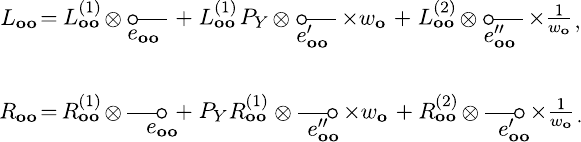}\codetag
\label{eq:L00R00def}
\end{equation}
In words, we consider the sums of Kronecker products of the left (resp. right) factors in \eqref{eq:3terms} with the auxiliary vectors $e_{\o\o}$, $e'_{\o\o}$, $e''_{\o\o}$, to eliminate the cross terms which would appear otherwise.\footnote{\label{note:nox}There are also no cross terms between $L_{\o \o}$, $R_{\o\o}$ and the $L_q$, $R_q$ tensors defined in \eqref{eq:LRnot00} for the other 4 channels.}

We now explain $w_\o$ and $P_Y$ in \eqref{eq:L00R00def}. The $w_{\o}>0$ is a reweighting parameter already discussed in \cref{sec:l-r-and-the-norm-of-kss}. As for $P_Y$, consider the following three subspaces in $V_2$:
\beq
\o \otimes (\u\oplus \d)\otimes \o,\quad \u\otimes \o\otimes\o,\quad \o \otimes \o\otimes\d\,.
\eeq
Let $Y$ be their direct sum and $P_Y$ the corresponding orthogonal projector. It's easy to check that the horizontal legs of $L_{\o\o}^{(2)}$ and $R_{\o\o}^{(2)}$ live in $Y$. Hence the result does not change if we multiply the tensors contracted with $L_{\o\o}^{(2)}$ and $R_{\o\o}^{(2)}$ by $P_Y$, which is what we did. As a matter of fact, this improvement is not essential but it leads to slightly better numerical estimates, so we implement it.

\subsubsection{Definition of $L$ and $R$}\label{sec:definition-of-l-and-r}

So far we represented the $\o\o$ channel as a contraction of $L_\o$ and $R_\o$, Eq.~\eqref{eq:needH00}, and the sum of 4 other channels as a contraction of the sum of $L_q$ and $R_q$ over those channels, Eq.~\eqref{eq:not00}. Given these and footnote \ref{note:nox}, we have
\beq\label{eq:LRdef}
\sum_{q} H_q = L \stackrel{V_3}{\text{\bf ---}} R,\qquad {L} := \sum_q L_q,\qquad {R} := \sum_q R_q,
\eeq
where the sum is over all 5 channels.

Recall that $\sum_q H_q$ equals the r.h.s.~of \eqref{eq:dis-as-sum}. So Eq.~\eqref{eq:LRdef} achieves the splitting. We will do one more minor step, to "erase" some unnecessary structure present in $V_3$. { This is analogous to the linearized analysis discussion around Eqs.~\eqref{eq:lin-kappa}-\eqref{eq:lin-L-kappa}. See also Q3 in Remark \ref{rem:QA} below.

We choose a real isometry $\varkappa:V_3 \to V$, which has the following properties:
\begin{align}
&\varkappa: e_{0} \otimes e_{0}\otimes e_{0} \otimes e_z \mapsto e_0 \label{eq:Kdef},\\
&\varkappa: (e_{0} \otimes e_{0}\otimes e_{0} \otimes e_z)^\perp \to \x.
\end{align}
Such an isometry exists by Lemma \ref{lem:isometry}. Following the proof of that lemma, an explicit construction can be given by enumerating all the basis elements of $V_3$. This can be done by using e.g.~multidimensional generalizations of the Cantor pairing function \cite{wikipedia_pairing_function}. We omit the details.}

We define tensors $L^\varkappa$ and $R^\varkappa$ contracting the $V_3$ leg of $L$ and $R$ with $\kappa$, i.e.
\beq
\label{eq:Lkappa}
L^\varkappa=L\stackrel{V_3}{\text{\bf ---}} \varkappa,\quad R^\varkappa=\varkappa \stackrel{V_3}{\text{\bf ---}} R\,. \codetag
\eeq
 Since $\varkappa$ is a real isometry, $\kappa^T \kappa=\mathds{1}_{V_3}$, the contraction of $L^\varkappa$ and $R^\varkappa$ over $V$ is the same as of $L$ and $R$ over $V_3$:
 \beq
 \sum_q H_q = L^\varkappa \stackrel{V}{\text{\bf---}}R^\varkappa\,,
 \label{eq:achieveLR}
 \eeq
This is the final form of splitting, which will be used below.

Let us split the tensors $L^\kappa$ and $R^\kappa$ into $A_*$ plus remainders. The component $L^{\kappa}_{0000}$ only has a contribution from $L_z$, Eq.~\eqref{eq:LkappaLz}. Using the definitions of $\lambda_i$ and $\lambda_i'$ in \cref{fig:lambdas} and the definition of $L^{(0)}_z$ in \cref{fig:LR}, we see that only the term with $i=1, k=1$ in the latter figure contributes. Then, using $A=A_*+b$, $b_{0000}=0$, we get $L^{\kappa}_{0000}=1$. Similarly, $R^{\kappa}_{0000} =1$. Thus we can write
\beq
\label{eq:LRsplit}
L^{\kappa}=\AHT+b_L,\qquad R^{\kappa}=\AHT+b_R,\qquad(b_L)_{0000} = (b_R)_{0000} = 0.
\eeq

\begin{remark} \label{rem:QA} Questions and answers.

\noindent Q1: The presented strategy for solving the equation in \cref{eq:AADD2LR} looks complicated. In numerical algorithms of tensor RG, one solves similar equations using singular value decomposition (SVD). Why can't we simply define $L$ and $R$ by performing an SVD of the r.h.s.~of \cref{eq:AADD2LR} as $U S V$, with $U,V$ unitary and $S$ nonnegative diagonal, and then defining $L_{\rm svd}=US^{1/2}$ and $R_{\rm svd}=S^{1/2} V$?

		\noindent A: This would not work for us for several reasons. First, tensors $L_{\rm svd}$ and $R_{\rm svd}$ are not guaranteed to be Hilbert-Schmidt in our infinite-dimensional setting. Second, since SVD involves matrix diagonalization, it would be nontrivial to estimate tensor elements of $L_{\rm svd}$ and $R_{\rm svd}$ in terms of tensor elements of $A$. In contrast, our definition of $L$ and $R$ involves nothing more than reshuffling and contracting tensor pieces, and the needed estimates will follow easily.

		\noindent Q2: What is the intuition behind the need to introduce the ``big'' Hilbert space $V_3$?

		\noindent A: The equation in \cref{eq:AADD2LR} has to be satisfied for arbitrary values of indices on 6 uncontracted legs. It is not so surprising that to solve this equation, the space of contracted indices should be in some sense "bigger" than $V$. In our infinite-dimensional setting, "bigger" means a complicated tensor product. In finite dimensions "bigger" would be literal, as the dimension of $V_3$ would have to be at least $\chi^3$ where $\chi$ is the dimension of $V$.

		\noindent Q3: Why erase almost all structure of $V_3$? What is special about the vector $e_{0} \otimes e_{0}\otimes e_{0} \otimes e_z$?

		\noindent A: We are studying the RG stability of the high-T fixed point. All tensors are split as $A_*$ plus remainders, and we only study how the norm of the remainders changes under the RG map. The basis vector $e_{0} \otimes e_{0}\otimes e_{0} \otimes e_z$ is the only component which contributes to the $A_*$ part of $L^\kappa$ and $R^\kappa$. It's not important how the other components are distributed over the tensor factors of $V_3$. They can be lumped into $\x$, in an arbitrary order, as this does not affect the norm of the remainder. { The order would have been important if we wanted, say, to discuss eigenvectors of the Jacobian of the map at $A_*$, but in this paper we do not do this.}

		\noindent Q4: Why are we keeping only two sectors $\o$, $\x$ on the contracted horizontal leg of $L^\kappa$, $R^\kappa$, and not four as for the horizontal legs of $A$?

		\noindent A: The subsequent 90 degree rotation will turn this horizontal leg into a vertical one.

	\end{remark}

\subsubsection{Hat-tensor estimates}\label{sec:hat-tensor-estimates}

We will now derive the hat-tensors $\widehat{L}^\kappa, \widehat{R}^\kappa$ for $L^\kappa,R^\kappa$. As with the gauge transformation, our graphical language makes this almost automatic.
\begin{itemize}
  \item We define the hat-tensors for $\lambda_3, \lambda_4, \lambda_5$ by putting hats on the $b$'s in \cref{fig:lambdas}. For $\lambda_5$, we also need to bound $1/(1+\sqrt{1-\alpha})^{1/2}$. We define an upper bound for $|\alpha|$:
  \begin{equation}\label{eq:hatalpha}
  	\widehat{\alpha} = \widehat{b}_{\o\o\d\x} \ \widehat{b}_{\d\o\o\x} \ \widehat{b}_{\u\x\o\o} \ \widehat{b}_{\o\x\u\o}, \codetag
  \end{equation}
  We will assume that $\widehat{\alpha}<1$. We then have $|1/(1+\sqrt{1-\alpha})^{1/2}|
  \le 1/(1+\sqrt{1-|\alpha|})^{1/2}
  \le  1/(1+\sqrt{1-\widehat{\alpha}})^{1/2}$. We use the latter bound in the definition of the hat-tensor.
  \item To obtain $\widehat{\lambda}_i', \widehat{\rho}_i, \widehat{\rho}_i'$, $i=3,4,5$, we first drop all the sign factors appearing in the definitions of the corresponding full-tensors and then repeat the same procedure as for $\lambda_3, \lambda_4, \lambda_5$.
  \item We define the check-tensors for $\lambda_1, \lambda_2$ by replacing the projector to $\o$ with the $2\times 2$ matrix with a single nonzero element $\o\o$ equal to $1$, the projector to $\x$ with the $2\times 2$ matrix with a single nonzero element $\x\x$ equal to $1$, and $e_0$ with a 4-dimensional vector with a single nonzero element $\o$ equal to $1$.
  \item To obtain $\widecheck{\lambda}_i', \widecheck{\rho}_i, \widecheck{\rho}_i'$, $i=1,2$, we first drop all the sign factors appearing in the definitions of the corresponding full-tensors and then perform the same replacements as for $\lambda_1, \lambda_2$.
  \item Hat-tensors for the orthonormal basis vectors of $V_{\rm aux}$ can be taken equal to the vectors themselves. (I.e.~we split $V_{\rm aux}$ into 7 one-dimensional sectors, one for each of its basis vectors.)
  \item It's now easy to check that the diagrams defining tensors $L_q$ and $R_q$ are allowed, and to define hat-tensors $\widehat{L}_q$ and $\widehat{R}_q$ using Lemma \ref{checklemma}.
  Summing those over $q$ we get hat-tensors $\widehat{L}$ and $\widehat{R}$.
\end{itemize}

  Note that $V_3$ has $N_3=N_1\times 4 \times N_1 \times 7$ sectors, where $N_1=4+4\times 4=20$ is the number of sectors of $V_1$.
  So $\widehat{L}$ and $\widehat{R}$ are 4-tensors with $4\times 2\times N_3 \times 2$ components.\footnote{ Many of these components are zero. Only 5 sectors of $V_1$ are actually used by $\lambda$ and $\rho$ ($\o$, $\d$, $\u$, $\d\otimes \u$, $\u\otimes\d$). So the number of used $V_3$ sectors is only $5 \times 4 \times 5 \times 7$. In the computer code that we use to compute the hat-tensors this is taken into account to speed up the computations.  }
On the other hand, the hat-tensors  $\widehat{L}^\kappa$ and $\widehat{R}^\kappa$ will only have $4\times 2 \times 2\times 2$ components.

There are two strategies to get $\widehat{L}^\kappa$ and $\widehat{R}^\kappa$ from $\widehat{L}$ and $\widehat{R}$. The simplest one would be to define a check-tensor $\widecheck{\kappa}$ and apply Lemma \ref{checklemma}. Below we describe a different strategy, which leads to much better hat-tensors. It takes advantage of the important fact that $\kappa$, being an isometry, maps sectors of $V_3$ to mutually orthogonal subspaces of $V$.

Defining hat-tensor for $L^\kappa$ means bounding the HS norm of its restriction $(L^\kappa)_{abcd}$. We consider the two cases $c=\o,\x$ separately. Denote by $\s_0$ the sector of $V_3$ spanned by $e_0\otimes e_0\otimes e_0\otimes e_z$. We have, for any $a,b,d$
\beq
\label{eq:LkappaLz}
L^\kappa_{ab\o d} = L_{ab\mathbf{s}_0 d} = (L_z)_{ab\mathbf{s}_0 d},
\eeq
where the first equation is by definition of $L^\kappa$ and the second follows because $L_z$ is the only tensor among $L_q$'s whose right leg lives in $\mathbf{s}_0$. Therefore we obtain the hat-tensor:
\beq
\widehat{L}^\kappa_{ab\o d} = (\widehat L_z)_{ab\mathbf{s}_0 d},\codetag
\eeq

Let now $(\s_r)_{r=1}^{N_3-1}$ be the remaining used sectors of $V_3$ in an arbitrary order. We have, for any fixed indices $i\in a,j\in b,l\in d$,
\begin{equation}
L^\kappa_{ij\x l} = \sum^{N_3-1}_{r=1} L_{ij \mathbf{s}_r l}\text{\bf---}\kappa_{\mathbf{s}_r \x}\,.
\end{equation}
The terms in the r.h.s. are vectors in the $\x$ subspace of $V$ which are mutually orthogonal for different values of $r$. This follows from the facts that sectors of $V_3$ are mutually orthogonal and that $\kappa$ being an isometry maps orthogonal vectors to orthogonal vectors. So the square of the norm of the l.h.s. equals the sum over $r$ of squares of the norms of the terms in the r.h.s. Summing the resulting equation over $i,j,l$ and taking the square root, we obtain:
\begin{equation}
	\label{eq:quad}
	\widehat{L}^\kappa_{ab\x d}
	= \left(
	\sum^{N_3-1}_{r=1} (\widehat{L}_{ab \mathbf{s}_r d})^2\right)^{1/2}\,. \codetag
\end{equation}
Note that for any given $\mathbf{s}_r$, the term in the r.h.s.~of this equation receives contributions from one of the three hat-tensors $\widehat{L}_{\x\x}$, $\widehat{L}_{\o\x}$, $\widehat{L}_{\x\o}$, or one of the three pieces of $\widehat{L}_{\o\o}$ in the r.h.s.~of \eqref{eq:L00R00def}. The $\widehat L_z$ does not contribute here.

\begin{remark}\label{rem:nonan}
	Because of several terms combined in quadrature as in \eqref{eq:quad}, $\widehat{L}^\kappa$ is not an analytic function of $\hat b$. On the other hand $L^\kappa$ is analytic function of $b$. There is no contradiction of course.
	\end{remark}

The discussion for $\widehat{R}^\kappa$ is completely analogous; we omit the details.

Recall that $L^\kappa$ and $R^\kappa$ are normalized and were split into $A_*$ plus remainders as in Eq.~\eqref{eq:LRsplit}. So their hat-tensors obtained as above can be written as
\beq
\label{eq:LRhatsplit}
\widehat{L}^{\kappa} = \hatAHT+\widehat{b}_L,\qquad\widehat{R}=\hatAHT+\widehat{b}_R,\qquad(\widehat{b}_{L})_{\o\o\o\o} = (\widehat{b}_R)_{\o\o\o\o} = 0.
\eeq

To summarize, in this section we defined the following functions:\footnote{As noted above, full-tensors $b_L,b_R$ have only two sectors $\o,\x$ on their contracted legs, not 4 sectors like $b$. Not to introduce further notation, we will continue to use $\mathbb{H}_0$ to denote the space in which these full-tensors live, as well as $\hat{\mathbb{H}}_0$ for the space of their hat-tensors. Hopefully this will not lead to a confusion.}
\begin{align}
& \widehat{\alpha}:\hat{\mathbb{H}}_0\to \mathbb{R}_{\ge 0},\\
& \widehat{b}_L, \widehat{b}_R :  \hat\Omega_{LR}  \to \hat{\mathbb{H}}_0,\quad \hat\Omega_{LR} =  \{\hat b\in \hat{\mathbb{H}}_0: \widehat{\alpha}<1 \}, \codetag \\
& b_L,b_R: \Omega_{LR}  \to \mathbb{H}_0, \quad \Omega_{LR} = \{b \in \mathbb{H}_0: \exists\ \widehat{b}\text{ for }b\text{ such that } \widehat{b} \in \widehat\Omega_{LR} \}\,.
\end{align}

\begin{proposition}\label{prop:bLbRms}
	The functions $\widehat{b}_L$ and $\widehat{b}_R$ are monotonic and subhomogeneous on the (downward closed) sets on which they are defined. Functions $b_L$ and $b_R$ are analytic, $b_L(0)=b_R(0)=0$.
\end{proposition}
\begin{proof}
	This proposition is analogous to Prop.~\ref{prop:summary-gauge}, so we will be brief. Let us discuss why $\hat b_L$ is subhomogeneous. Leaving aside $(\hat b_L)_{\o\o\o\o}=0$, the other elements of $\hat b_L$ are obtained by combining in quadrature as in \eqref{eq:quad} tensor elements of $\hat L$ which are all subhomogeneous, given by monomials in $\hat b$ of order $\ge 1$, times scalar factors monotonic in $\hat b$. Combining in quadrature preserves subhomogeneity.
	\end{proof}

Recall that throughout \cref{sec:disent}, $A$ was used as a shorthand for $A_1$, and $b$ was used as a shorthand for $b_1$, where the end result of the gauge transformation was $A_1 = A_* + b_1$. When we compose the four steps of the 2x1 map, we will have to restore the subscript $1$.

\subsection{Reconnection and rotation}\label{sec:reconnection-and-rotation}

The last steps of our RG map is the reconnection followed by rotation, Fig.~\ref{fig:bigfig}(c),(d).
After disentangling and splitting, our tensor network consists of 2x1 blocks as in the r.h.s. of Fig.~\ref{fig:bigfig}(b). Tensors $L,R$ were defined in Section \ref{sec:definition-of-l-and-r}; at the end of which section we also replaced them with $L^\kappa$, $R^\kappa$. In terms of these latter tensors, reconnection is graphically described by the following diagram:
\begin{equation}\label{eq:A2}
  \myinclude{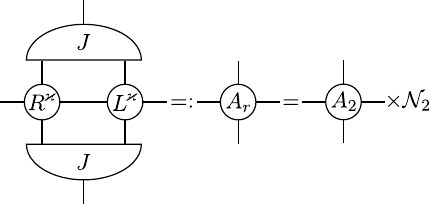}\,. \codetag
\end{equation}
Here:
\begin{itemize}
\item $J$ is an isometry from $V\times V$ onto $V$, see below;
\item $A_{r}$ ($r$ for reconnection)  is defined as the shown contraction;
\item Since $A_{r}$ is not in general normalized, we define the normalized $A_2=A_{r}/\mathcal{N}_2$, $\mathcal{N}_2=(A_{r})_{0000}$.
\end{itemize}
The reconnection step halves the horizontal size of the lattice. We have:
\beq
\boxed{Z(A_1,\ell_x\times \ell_y)=Z(A_{r},\ell_x/2\times \ell_y)=\mathcal{N}_2^{{\rm Vol}/2}
Z(A_2,\ell_x/2\times \ell_y)}\,.\label{eq:Zrec}
\eeq
After this, we perform the rotation step, Fig.~\ref{fig:bigfig}(d), which we copy here:
\begin{equation}\label{eq:Ap}
  \myinclude{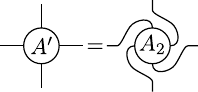}\,. \codetag
\end{equation}
This defines the final, normalized, tensor $A'$ as the 90 degree counterclockwise rotation of $A_2$. So clearly
\beq
\boxed{Z(A_2,\ell_x/2\times \ell_y)=Z(A', \ell_y\times \ell_x/2)}\,.\label{eq:Zrot}
\eeq
Combining Eqs.~\eqref{eq:Zgauge1},\eqref{eq:Zrec},\eqref{eq:Zrot}, we obtain \eqref{eq:Zpreserved} with $\mathcal{N}=\mathcal{N}_1^2 \mathcal{N}_2$. We record this as:
\begin{proposition}
The defined map $A\mapsto \calN A'$ is an RG map (for any values of the reweighting parameters $w_\x, w_\o>0$) in the sense that it preserves the tensor network partition function.
	\end{proposition}

{ Let us next discuss the isometry $J:V\otimes V\to V$. We use the same $J$ as in the linearized analysis. The "ingoing" vertical legs of $J$ are split into sectors corresponding to $V=\o\oplus \x$, while the "outgoing" leg is split into sectors $V=\o\oplus\d\oplus\u\oplus \r$. This leg is vertical, but it will become horizontal after rotation, so we use 4 sectors as on the horizontal legs of $A$. Then, we require that $J$ maps isometrically the 4 sectors of $V\otimes V$ onto the 4 sectors of $V$ as in Eq.~\eqref{eq:J_who_where} (see \cref{fig:J-sectors}). The rationale behind these assignments was already discussed in \cref{sec:lin-reconnection-and-rotation}.}

\begin{figure}[h]
	\centering
	\codetag\includegraphics{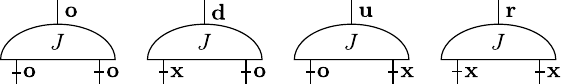}
	\caption{Nonzero sectors of the $J$ isometry. The ``ingoing'' legs are distinguished.}
	\label{fig:J-sectors}
\end{figure}

We next discuss the hat-tensors and check-tensors.

The isometry $J$ is not HS, so we use a check-tensor for it. The tensor $J$ has two distinguished legs, indicated by tick marks in \cref{fig:J-sectors}. The check-tensor $\widecheck{J}$ has elements equal to $1$ for the four elements corresponding to the four sectors in the figure. The rest of the elements are zero.

The contraction in \cref{eq:A2} which defines $A_{r}$ does not quite meet our definition of an allowed contraction (\cref{def:allowed}) since the two distinguished legs of $J$ are contracted with two different HS tensors. But there is an easy workaround by doing that contraction in two steps. We first contract $R^{\kappa}$ with $L^{\kappa}$, to get a HS tensor which is then contracted with $J$'s. Both steps are allowed contractions and Lemma \ref{checklemma} applies. We thus see that the hat-tensor $\widehat{A}_{r}$ is given by replacing $R^\kappa\to \widehat R^\kappa$, $L^\kappa\to \widehat L^\kappa$, $J\to \widecheck{J}$ in the l.h.s.~of \cref{eq:A2}.

The hat-tensor for $A_2=A_r/\calN_2$ is now obtained similarly to $A_1=A_g/\calN_1$ in the gauge transformation discussion. We start by writing
\beq
\widehat{A}_r = \widehat\calN_2 \widehat{A}_* + \widehat{b}_r,\qquad  (\widehat{b}_r)_{\o\o\o\o}=0. \codetag
\eeq
Here $\widehat{\calN}_2$, $\widehat{b}_r$ are functions of $\widehat{b}_L$ and $\widehat{b}_R$. $\widehat{\calN}_2$ is an upper bound on $|\calN_2|$, while a positive lower bound for $|\calN_2|$ is given by $|\calN_2|_{plb}=2-\widehat{\calN}_2$, positive assuming $\widehat{\calN}_2<2$. The hat-tensor for $A_2$ is then given by
\beq
\widehat{A}_2= \widehat{A}_* + \widehat{b}_2,\qquad
\widehat{b}_2= \widehat{b}_r/|\calN_2|_{plb}\,,
\eeq

Finally, we obtain $\widehat{A}'=\AHT + \widehat{b}'$ by rotating $\widehat{A}_2$ by 90 degrees counterclockwise.

 	We summarize the most important functions defined in this section:
 	\begin{align}
 		&\hat{\calN}_2:\{(\widehat{b}_L,\widehat{b}_R) \in \hat{\mathbb{H}}_0\times \hat{\mathbb{H}}_0\} \to \mathbb{R}_{\ge 0},\\
 		& \widehat{b}': \widehat{\Omega}_2 \to \hat{\mathbb{H}}_0 , \quad  \widehat{\Omega}_2=\{(\widehat{b}_L,\widehat{b}_R) \in \hat{\mathbb{H}}_0\times \hat{\mathbb{H}}_0 :\hat{\calN}_2<2\}\codetag \\
 		&  \calN_2:\{(b_L,b_R)\in \mathbb{H}_0\times \mathbb{H}_0\} \to\mathbb{C},\\
 		& b': \Omega_2 \to \mathbb{H}_0\,, \quad \Omega_2 = \{(b_L,b_R)\in \mathbb{H}_0\times \mathbb{H}_0: \exists\ \widehat{b}_L,\widehat{b}_R\text{ for }b_L, b_R\text{ such that } (\widehat{b}_L,\widehat{b}_R) \in \widehat\Omega_{2} \}\,.
 	\end{align}
 	The following proposition is analogous to Props.~\ref{prop:summary-gauge},\ref{prop:bLbRms}; we omit the proof.

 	\begin{proposition}\label{prop:summary-rec-rot}
$\hat{\calN}_2$ is monotonic. $\widehat{b}'$ is monotonic and subhomogeneous on the (downward closed) set $\Omega_2$. Functions $\calN_2$, $b'$ are analytic. We have $b'(0,0)=0$, $\calN_2(b_L,b_R)=1+O(b_L b_R)$, $\calN_2\ne 0$ in $\Omega_2$, and the bound:
\beq
\label{eq:N2dev}
|\calN_2-1|\le \hat{\calN}_2-1\,.
\eeq
 	\end{proposition}

\subsection{Master function}\label{sec:masterf}

All the steps of the 2x1 map have now been defined. Composing them, we define the map $b'(b)$ on the set $\Omega$ as in \eqref{eq:Omega} where the associated set $\hat\Omega$ is given by:
\begin{align}
\hat\Omega = \{\hat b\in \hat{\mathbb{H}}_0: &\ \hat\beta_h(\hat b)<1,\ \hat\beta_v(\hat b)<1,\ \hat\calN_1(\hat b)<2,
\nonumber\\
 &\ \hat \alpha(\hat b_1)<1\text{ where }\hat b_1= \hat b_1(\hat b),
 \nonumber\\
 &\ \hat{\calN}_2(\hat{b}_L, \hat{b}_R)<2\text{ where }\hat b_L= \hat b_L(\hat b_1),\hat b_R= \hat b_R(\hat b_1)
 \}. \codetag
\end{align}
The associated hat-tensor map $\hat{b}'(\hat b)=:\mathfrak{M}(\hat b)$ is the master function, defined on $\hat\Omega$.

The $\calN$-factor is defined on $\Omega$ by
\beq
\calN(b) = (\calN_1(b))^2 \calN_2(b_L,b_R)\text{ where }b_L= b_L(b_1(b)), b_R= b_R(b_1(b))\,. \codetag
\eeq
The following proposition is obtained by putting together Props.~\ref{prop:summary-gauge},\ref{prop:bLbRms},\ref{prop:summary-rec-rot}; we omit the proof.

\begin{proposition}\label{prop:summary-master}
	The master function $\mathfrak{M}(\hat b)$ is monotonic and subhomogeneous on the downward closed set $\hat\Omega$. Functions $\calN(b)$, $b'(b)$ are analytic on $\Omega$. We have $b'(0)=0$, $\calN(b)=1+O(b^2)$, $\calN(b)\ne 0$ in $\Omega$.
\end{proposition}

\subsection{Constraints on $w_{\x}$ and $w_\o$ from linear stability}\label{sec:leading}

 Now that the 2x1 map has been described in detail, we would like to come back to the question about the numerical values of parameters $w_{\x}$ and $w_\o$. We would like to require that the linearization of the RG map is a contraction in some norm, as discussed in \cref{sec:lin}. Let us use the leading part of the master function to control the linearization of the RG map. We write
\beq
b'(b)= b'{}^{(1)}+O(b^2)
\eeq
where $b'^{(1)}$ is linear in $b$. We can get a hat-tensor for $b'{}^{(1)}$ by extracting the leading part of the master function, as follows:
\beq
\label{eq:linpart}
\widehat{b}'^{(1)}=\lim_{\epsilon\to 0} \epsilon^{-1} \mathfrak{M}(\epsilon \widehat b) \,.\codetag
\eeq
Our computer code for the master function $\mathfrak{M}(\hat b)$ \cite{our-code} (see also App.~\ref{sec:accompanying-code}) is suitable for both numerical and symbolic calculations. Numerical mode will be used in Section \ref{sec:results} for controlling the map at the full nonlinear level.
Using the symbolic mode of the code, we easily evaluate \eqref{eq:linpart} and obtain \cite{our-code}:
\begin{equation}\label{eq:lead}
  \begin{aligned}
    \red{\widehat{b}'^{(1)}_{\u\x\o\o}} & =\left[
    {W\, \red{\widehat{b}^2_{\o\x\u\o} }}+
    \left( 1+1/{w_{\x}^2} \right) \widehat{b}_{\o\x\r\o}^{2}+\left( 1+ w_{\o}^{2}+1/{w_{\x}^2} \right) \widehat{b}_{\o\x\d\o}^{2}  \right]^{1/2},
    \\
   \red{\widehat{b}'^{(1)}_{\o\x\u\o}} & = \left[{W\,\red  \widehat{b}^2_{\o\o\d\x}}+\left( 1+1/{w_{\x}^2} \right) \widehat{b}_{\o\o\r\x}^{2}+\left( 1+ w_{\o}^{2}+1/{w_{\x}^2} \right) \widehat{b}_{\o\o\u\x}^{2}  \right]^{1/2},
    \\
   \red{\widehat{b}'^{(1)}_{\d\o\o\x}} & = \left[W\,{\red \widehat{b}^2_{\u\x\o\o}}+\left( 1+1/{w_{\x}^2} \right) \widehat{b}_{\r\x\o\o}^{2}+\left( 1+ w_{\o}^{2}+1/{w_{\x}^2} \right) \widehat{b}_{\d\x\o\o}^{2}  \right]^{1/2},
    \\
   \red{\widehat{b}'^{(1)}_{\o\o\d\x}} & = \left[W\, {\red \widehat{b}^2_{\d\o\o\x}}+\left( 1+1/{w_{\x}^2} \right) \widehat{b}_{\r\o\o\x}^{2}+\left( 1+ w_{\o}^{2}+1/{w_{\x}^2} \right) \widehat{b}_{\u\o\o\x}^{2}  \right]^{1/2},
    \\
   \red{\widehat{b}'^{(1)}_{\u\x\u\o}} & = \left[\left( 1+ w_{\o}^{2}+{1}/{w_{\x}^2} \right) \widehat{b}_{\o\x\d\x}^{2}
   +\left( 1+ w_{\o}^{2}+1/w_{\x}^2\right) \widehat{b}_{\o\x\u\x}^{2}+\left( 1+1/{w_{\x}^2} \right) \widehat{b}_{\o\x\r\x}^{2}\right]^{1/2},
    \\
   \red{\widehat{b}'^{(1)}_{\d\o\d\x}} & = \left[\left( 1+ w_{\o}^{2}+1/{w_{\x}^2} \right) \widehat{b}_{\d\x\o\x}^{2}+\left( 1+ w_{\o}^{2}+1/{w_{\x}^2} \right) \widehat{b}_{\u\x\o\x}^{2}+\left( 1+1/{w_{\x}^2} \right) \widehat{b}_{\r\x\o\x}^{2}\right]^{1/2},
    \\
   \red{\widehat{b}'^{(1)}_{\d\o\d\o}} & = \widehat{b}_{\o\x\o\x},                                                                                                                                                                                                      \\
   \red{\widehat{b}'^{(1)}_{\u\o\u\o}} & = \widehat{b}_{\o\x\o\x},
  \end{aligned} \codetag
\end{equation}
where we defined
\beq
W =  2 \left(1/w^2_{\x}+1/{w^2_{\o}} \right).
\eeq
Square roots in \eqref{eq:lead} are a consequence of combining in quadrature in \eqref{eq:quad}, which makes $\mathfrak{M}(\widehat b)$ non-analytic in $\widehat b$, see \cref{rem:nonan}. Checking \eqref{eq:lead} by pencil and paper would be a straightforward but tedious task; we haven't done it.

The special sectors of $\widehat{b}'^{(1)}$ and of $\hat{b}$ are shown in  \eqref{eq:lead} in red. (The special sectors were defined in \cref{eq:special}.) Only special sectors appear in the l.h.s. This means that the non-special sectors of $b'{}^{(1)}$ vanish. This is as discussed in \cref{sec:lin}.

The non-special sectors do appear in the r.h.s. of  \eqref{eq:lead}. But as discussed in \cref{sec:lin}, we may focus on the special-to-special part of the linearization, denoted there $K_{ss}$, and it's enough to ascertain that the operator norm (with respect to the HS norm of tensors) $\|K_{ss}\|<1$. Using \eqref{eq:lead}, we have the bound
\beq
\|K_{ss}\|\le \sqrt{W}\,.
\eeq
Therefore, for $W<1$ our map will be a contraction at the linearized level in the $\vvvert\cdot\vvvert$ norm described in  \cref{sec:lin}. Since the 2x1 map is analytic, it will also remain a contraction at the full nonlinear level in a sufficiently small neighborhood of the origin and in the same norm. So, once RG iterations get sufficiently close to the origin, further iterations will decrease exponentially fast.

Below we will always choose parameters $w_\x$ and $w_\o$ so that $W<1$. However we will no longer discuss the norm in which the map is a contraction. Instead, we will switch to controlling the map through the master function, as described in Section \ref{sec:control-of-the-2x1-map-via-hat-tensors}. This method can efficiently control both the early RG iterations, until the exponential decay sets in, and, thanks to Key Lemma \ref{lem:key}, later exponentially decaying iterations, saving the effort to exhibit a contracting norm.

\subsection{Symmetries}\label{sec:symmetries}

In conclusion, let us discuss symmetries of the 2x1 RG map, starting with the lattice symmetries.
The square lattice has $\mathbb{Z}_4$ discrete rotation symmetry and $\mathbb{Z}_2\times \mathbb{Z}_2$ reflection symmetries in the horizontal and vertical direction. The 2x1 map does not preserve these symmetries. For the rotation symmetry this is quite obvious, since the horizontal and the vertical directions are treated asymmetrically.

As for the reflection symmetries, they are violated by the assignment of minus signs for the $\rho_4,\rho_5$ tensors relative to the $\lambda_4$, $\lambda_5$ tensors (\cref{sec:decomposition-of-disentangler}). Notice however that when we defined the master function, we replaced all minus signs by plus signs. So the master function preserves the reflection symmetries. (Note that the vertical reflection exchanges $\u$ and $\d$.)

Let us discuss next global symmetries. For lattice spin models, a global symmetry acts on spin values in a way which preserves the energy of a spin configuration. When such a model is translated into a tensor network, the global symmetry is inherited by the resulting tensor. The tensor network partition function remains invariant when we act on the tensor by gauge transformations, Fig.~\ref{fig:bigfig}(a). The global symmetry group is formed by gauge transformations which leave the tensor invariant. In practice, consequences of a global symmetry group are best understood by decomposing the Hilbert spaces of vertical and horizontal legs into irreducible representations. Global symmetry then manifests itself through selection rules: a tensor element can be nonzero only if the tensor product of the four respective representations contains the trivial representation. We will see some examples in \cref{sec:results}.

In our description of the 2x1 map we did not keep track of the possible global symmetry $\mathcal{G}$. It is possible to generalize the discussion.  It is natural to assume that index $0$ transforms in the trivial representation. Then, one can show that the 2x1 map preserves the global symmetry. Namely, all tensors which entered the definition of the map can be consistently defined so that they are $\mathcal{G}$-invariant, provided that $b$ is $\mathcal{G}$-invariant. In particular the final tensor $b'$ will be $\mathcal{G}$-invariant.
To take advantage of this, it would be necessary to define a more fine-grained master function, dividing various sectors into subsectors corresponding to different irreps. This will not be done in this work.

\section{Computer-assisted bounds on the high-T phase}\label{sec:results}

In this section, we will use the 2x1 map defined in \cref{sec:construction} to prove results in Table \ref{tab:results} about the high-T phase of lattice models. We consider general tensors first, and a couple of specific models later.

\subsection{General tensors}

Let $\mathbb{O}$ be a set containing the origin in $\mathbb{H}_0$ (which may or may not be open). Consider a tensor network built out of tensor $A=A_*+b$, $b\in \mathbb{O}$. Fix some values of reweighting parameters of the 2x1 map. We will say that $\mathbb{O}$ is a \emph{basin of stability} if for any $b=b^{(0)}\in \mathbb{O}$:
\begin{itemize}
	\item
the 2x1 map can be iterated infinitely many times;
\item
the resulting sequence of tensors $b^{(i)}$, $i=1,2,\ldots$, tends to zero in the HS norm.
\end{itemize}

For $\delta>0$, we consider a closed neighborhood of the origin in $\mathbb{H}_0$ given by:
	\beq
	\mathbb{O}_\delta=\{b\in \mathbb{H}_0: \|b_{abcd}\|\le \delta \text{ for all $abcd\ne \o\o\o\o$}\}\,.
	\eeq
The following theorem is our first main result.
\begin{theorem}(Stability of the high-T fixed point)
	\label{th:stability}
Fix the reweighting parameters $w_\x=2.2$, $w_\o=2$. Then $\mathbb{O}_\delta$, $\delta=0.02$, is a basin of stability.
\end{theorem}

\begin{proof}
	To control the 2x1 map we use the master function as described in Section \ref{sec:control-of-the-2x1-map-via-hat-tensors}.
	We define a hat-tensor $\hat b^{(0)}\in \hat{\mathbb{H}}_0$ with $b^{(0)}_{\o\o\o\o}=0$ and all other components $\hat b^{(0)}_{abcd}=\delta$. We then start acting on $\hat b^{(0)}$ repeatedly with the master function $\mathfrak{M}$, generating a sequence of hat-tensors $\hat b^{(i)}$, $i=1,2,\ldots$. The idea is that $\hat b^{(0)}$ defines a box around $A_*$ which flows under RG, and we need to arrange that it flows to zero size.

We perform the first $i_0+1$ iterations numerically using our computer code. At every step the code checks that the master function is defined and if so, it evaluates the next $\hat b^{(i)}$. Suppose that all these evaluations went through, and the last two tensors satisfy \eqref{eq:bhatcase2} which we copy here:
\beq
\label{eq:start}
\hat b^{(i_0+1)}\le \lambda\, \hat b^{(i_0)}\text{ for some $\lambda<1$}. \codetag
\eeq
For $\delta=0.02$, $w_\x=2.2, w_\o=2$, our calculations \cite{our-code} show that this holds for $i_0=15$ with $\lambda=0.96$. (Our code evaluates $\hat b^{(i)}$ and checks \eqref{eq:start} using interval arithmetic so that the argument is rigorous.)
Given \eqref{eq:start}, Key Lemma \ref{lem:key} tells us that all further iterations of the master functions are also defined and tend to zero.

By Proposition \ref{prop:summary-master}, since the master functions iterates are defined, the 2x1 map iterates $b^{(i)}$, $i=1,2,\ldots$ are also defined for any $b^{(0)}$ for which $\hat b^{(0)}$ is a hat-tensor, i.e.~for any $b^{(0)} \in \mathbb{O}_\delta$, and moreover $\hat b^{(i)}$ is a hat tensor for $b^{(i)}$. Therefore, $b^{(i)}$ tends to zero in the HS norm.
\end{proof}

How did we arrive at the above numerical values of $\delta,w_\x,w_\o$? We took $w_{\x}$, $w_\o$ within the range $W<1$ but not too large, as the latter would unnecessarily enhance $K_{ns}$ and the nonlinear terms. We ran master function iterations for several values of $w_\x$ and $w_\o$ near 2, and for several values of $\delta$, producing plots like \cref{fig:0.02box} below. We then picked $w_\x=2.2$ and $w_\o=2$ as giving us the largest $\delta$ among those we tried, for which we could see exponential decrease set in after 10-20 iterations. Of course, varying $w_{\x}$ and $w_\o$ a bit, the theorem still holds for a somewhat different $\delta$. In the future, it would be interesting to optimize them systematically and to establish an even larger basin of attraction of the high-T fixed point; see \cref{sec:conclusions} for a discussion.

In \cref{fig:0.02box} we show the sequence of hat-tensors as a function of the RG step $i$, for $\delta,w_\x,w_\o$ as in the theorem. Note that many components of $\hat{b}^{(i)}$ behave nonmonotonically in the initial iterations, but eventually the exponential decay sets in, and condition \eqref{eq:start} can be verified. The initial nonmontonic behavior suggests that it would be challenging, or perhaps impossible, to prove Theorem \ref{th:stability} by exhibiting a norm in which the 2x1 map is a contraction in the whole of $\mathbb{O}_\delta$.

\begin{figure}[h]
	\centering
	\codetag\includegraphics[scale=0.7]{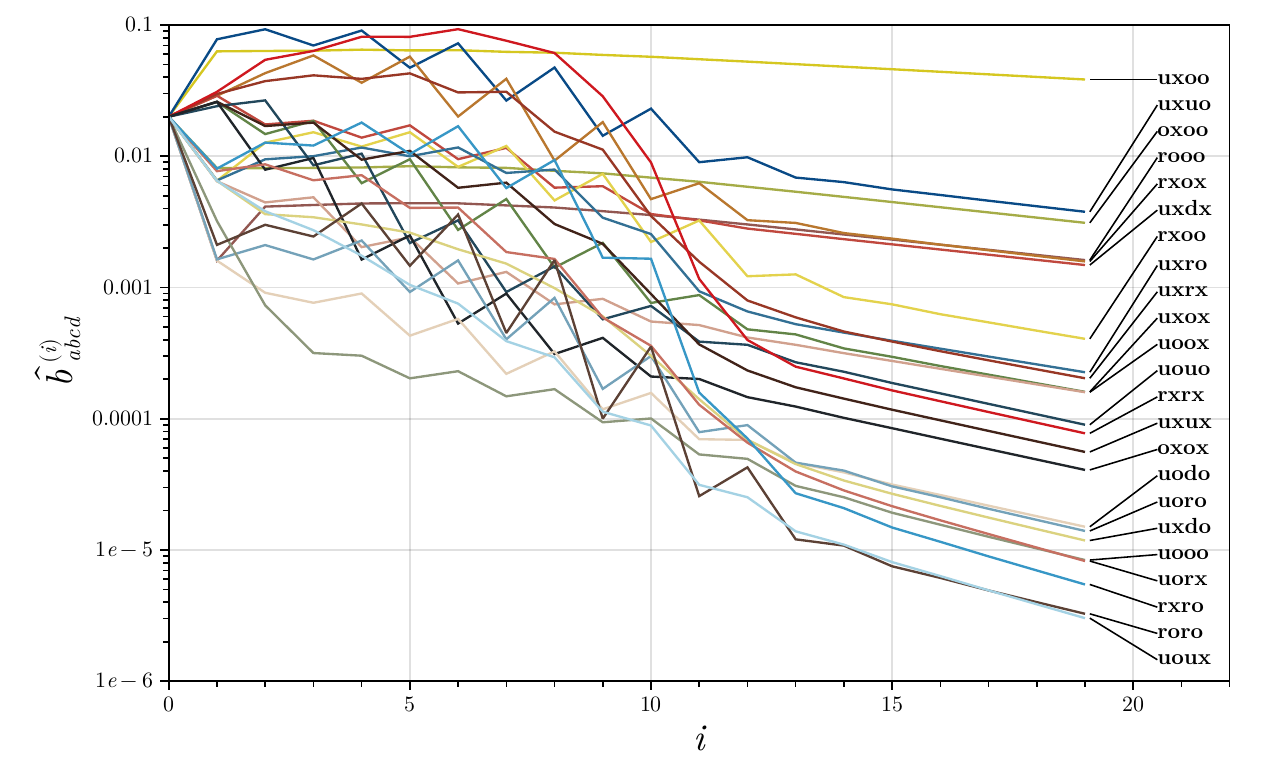}
	\caption{Master function iterates with the initial conditions $\widehat{b}^{(0)}_{\o\o\o\o} = 0$, $\widehat{b}^{(0)}_{abcd} = \delta = 0.02$ for all other components, and for the reweighting parameters $w_\x = 2.2$,  $w_{\o} = 2$. We only plot one sector per each group related by the lattice reflections. This is because the reflection symmetry of the initial $\widehat{b}^{(0)}$ is preserved by the master function, see \cref{sec:symmetries}. (Our interval arithmetic calculations are initialized with precision $10^{-21}$. After 20 iterations the precision is $10^{-17}$, i.e.~completely negligible on the scale of this plot.)
	}

	\label{fig:0.02box}
\end{figure}

We next discuss the free energy. The tensor network free energy per site is defined by \cite{paper2}\footnote{Sometimes the free energy is defined with the opposite sign.}
\beq\label{eq:freeen}
f(A) = \lim_{\ell \to \infty} f(A,\ell\times\ell),\quad f(A,\ell_x\times\ell_y)=\frac{1}{\ell_x\ell_y} \log Z(A,\ell_x\times \ell_y)\,.
\eeq
\codetag We will write $f(b)$ to denote $f(A_*+b)$ hoping that this will not lead to a confusion. For Hamiltonians with finite-range interactions, and hence for tensor networks obtained by translating partition functions of such Hamiltonians, the existence of such a limit is a classic result \cite{ruelle1999statistical}. For an arbitrary $A$ the limit may not always exist.

However, if the RG iterates tend to zero, then it's easy to show that the limit exists along a subsequence and moreover the free energy is analytic. Thus all such tensors belong to the same phase, which we call the high-T phase. More precisely we have the following result:

\begin{proposition}\label{prop:free} Suppose $\mathbb{O}$ is a basin of stability. Then for any $b\in \mathbb{O}$ the limit defining the free energy exists along the subsequence $\ell=2^k$. The free energy is given by the sum of the following series:
	\beq
	\label{eq:fbseries}
	f(b) = \sum_{i=0}^\infty 2^{-i-1}\log \calN(b^{(i)})\,.
	\eeq
	It is analytic in the interior of $\mathbb{O}$.
	\end{proposition}

\begin{proof} This is fully analogous to \cite{paper2}, Prop.~4.3, so we will be brief. (The RG map there had scale factor 4.) By Prop.~\ref{prop:summary-master}, $\calN(b)$ is analytic and nonzero in $\Omega$, and $\calN(b)=1+O(b^2)$. Since $b^{(i)}\to 0$, the series in \eqref{eq:fbseries} converges, and its sum is analytic in the interior of $\mathbb{O}$.
	Furthermore, by repeatedly using \eqref{eq:Zpreserved}, we show that $f(b,2^{k+1}\times 2^{k+1})$ equals the partial sum of \eqref{eq:fbseries} over $i\le 2k+1$ plus the "remainder" $2^{-2k} f(b^{(2k)},2\times 2)$. As $k\to\infty$, the remainder tends to zero, and we are done.
	\end{proof}

Combining this with Theorem \ref{th:stability}, we get

\begin{proposition}\label{prop:free1}
The free energy is analytic in the interior of $\mathbb{O}_{\delta}$, $\delta=0.02$.
\end{proposition}

It's interesting to discuss how these results compare to the cluster expansion. The cluster expansion \cite{friedli_velenik_2017} also provides an analytic representation for the free energy in the region where it converges. The cluster expansion computes the Taylor expansion of the free energy of a lattice model in polymer activities. As discussed in \cite[App.~B]{paper1}, the free energy of a tensor network can also be computed via a cluster expansion, with polymers being connected subsets of $k$ lattice points, polymer activity bounded by $\|b\|^k$, and the hardcore repulsion as the compatibility condition. Applying a standard criteria of cluster expansion convergence such as Koteck\'y-Preiss \cite{kotecky_preiss_1986}, our back-of-the-envelope estimate gives convergence for  $\|b\|<0.011$.\footnote{This cluster expansion is similar in structure to the Ising model in strong magnetic field discussed in \cite[Section 5.7.1]{friedli_velenik_2017}, replacing their polymer activities by the $\|b\|^k$ bound, $k=|S|$. Note that $k\ge 2$ in our case. We choose $a(S)=\alpha |[S]_1|$ where $|[S]_1|$ is the same as in \cite{friedli_velenik_2017}, and $\alpha$ is a free parameter we introduce to optimize the result. We use their Exercise 5.3 to bound the number of polymers. The result $\|b\|<0.011$ follows for $\alpha=0.1$.} We don't want to claim that this is the best the cluster expansion can do. But it does look like our RG method is already doing better than the cluster expansion in this general tensor network setting. Furthermore it can be significantly improved, as we will discuss in \Cref{sec:conclusions}.

Unlike the cluster expansion of the free energy, Eq.~\eqref{eq:fbseries} is not a Taylor expansion. Term number $i$ in this equation, if expanded in $b$, would contain terms of all orders. Eq.~\eqref{eq:fbseries} converges extremely rapidly for two reasons. First, because of the factor $2^{-i-1}$ and second because $b^{(i)}\to 0$ and $\calN=1+O(b^2)$. For any concrete $b^{(0)}$, one can take advantage of this fast convergence by evaluating a few terms of the expansion explicitly and estimating the rest.

For generic $b^{(0)}$, we can get the following bounds on the free energy in $\mathbb{O}_\delta$.
Let us write $\widehat\calN_1=1+\widehat n_1$, $\widehat\calN_2=1+\widehat n_2$. Then, using $ \calN=\calN_1^2\calN_2$ we have the lower and upper bounds
\begin{gather}
C_i^-\le \log \calN(b^{(i)})\le  C_i^+,\nonumber \\
C_i^\pm =2 \log [1\pm \widehat n_1(\widehat b^{(i)})]+\log [1\pm \widehat n_2(\widehat b^{(i)})]\,.\label{eq:bndpm} \codetag
\end{gather}
Note that $C_i^-$ is negative while $C_i^+$ is positive.

We can use \eqref{eq:bndpm} to bound every term in the series \eqref{eq:fbseries}. We split the resulting series into the head $i\le i_0$ plus the tail $i> i_0$, where $i_0$ is as in the proof of Theorem \ref{th:stability}. The head can be evaluated since $\widehat n_1(\widehat b^{(i)})$, $\widehat n_2(\widehat b^{(i)})$ are available in the process of iterating the master function. A tail term number $i$ can be estimated by $2^{i_0-i}$ times the last head term, since $\widehat n_1$ and $\widehat n_2$ are monotonic, and we know from Key Lemma \ref{lem:key} that $\widehat b^{(i)}\le \widehat b^{(i_0)}$. We thus obtain free energy bounds valid for any $b\in \mathbb{O}_\delta$:
\beq
\label{eq:explbnd}
\sum_{i=0}^{i_0}(1+\delta_{i,i_0}) 2^{-i-1}C_i^- \le f(b) \le \sum_{i=0}^{i_0}(1+\delta_{i,i_0})2^{-i-1}C_i^+\,, \codetag
\eeq
where $\delta_{i,i_0}$ is the Kronecker delta, $\delta_{i,i_0}=1$ for $i=i_0$ and is zero otherwise.

Bound \eqref{eq:explbnd} may appear not particularly impressive, for example it does not determine the sign of $f(b)$. We show it to demonstrate the general idea. As mentioned above, in a more meaningful treatment we need to evaluate the first few terms of \eqref{eq:fbseries} explicitly, and use a bound like \eqref{eq:explbnd} on the remaining terms. This is expected to give an extremely accurate rigorous estimate of free energy throughout the basin of stability.

\subsection{Ising model}

Let us examine next how well our method can detect the high-T phase in some familiar lattice models, starting with the (nearest-neighbor) Ising model defined by the partition function:
\begin{equation}
  Z = \sum_{\sigma_{x} = \pm 1} \exp
  \Bigl(\beta \sum_{\langle xy\rangle } \sigma_{x}\sigma_{y} \Bigr),
\end{equation}
where $x$ are vertices of the square lattice and $\sum_{\langle xy\rangle }$ means summation over nearest-neighbor pairs. The Ising model can be represented as a tensor network composed of the tensor $A=A(\beta)$ with bond dimension $2$ and the following nonzero tensor elements \cite{TEFR},\cite[App.~A.1]{paper1}
\begin{gather}
	\label{eq:Isingtensor}
 A_{0000} = \cosh(4\beta) + 3, \quad  A_{0101} =  A_{1111} =\cosh(4\beta) - 1, \quad  A_{0011} = {\sinh(4\beta)}, \codetag
\end{gather}
and rotations thereof.\footnote{To be precise, tensor network partition function $Z(A,\ell\times\ell)$ with periodic boundary conditions equals the Ising model partition functions on the square lattice rotated by 45 degrees, which contains $2\ell^2$ spins. This needs to be kept in mind when comparing the free energy per site.\label{foot:factor2}} The $\mathbb{Z}_2$ symmetry of the Ising model is reflected in $A$ as follows (see Section \ref{sec:symmetries}). The state 0 is $\mathbb{Z}_2$-even i.e.~transforms in the trivial representation. The state 1 is $\mathbb{Z}_2$-odd i.e.~transforms in the nontrivial representation. All elements of $A$ with an odd number of 1 indices are zero. However, we will not take advantage of the $\mathbb{Z}_2$ symmetry in the analysis below.

As usual we normalize the tensor by writing $A=A_{0000}(A_*+b(\beta))$ with $b(\beta)\in \mathbb{H}_0$. Note that $b(\beta)=O(\beta)$ as $\beta\to0$ and has nonnegative tensor elements which are monotonic functions of $\beta$.

 We would like to find $\bar\beta$ such that $A_*+b(\beta)$ for $\beta\le\bar\beta$ flows to $A_*$ under tensor RG. This will imply that the Ising model for $\beta\le\bar \beta$ is in the high-T phase. Already Theorem \ref{th:stability} could be used to get an estimate for $\bar\beta$, but we will instead get a much better result using the explicit form of $A$.

We find it natural to do a preliminary RG step. We define tensor $A^{(0)}$ (which will be the initial tensor for 2x1 map iterations) as:
\beq
\label{eq:preliminary}
\myinclude{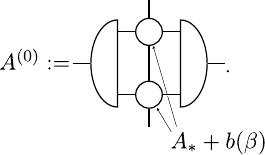}\codetag
\eeq
I.e.~$A^{(0)}$ is a contraction of two copies of $A_*+b$ put on top of each other, followed by an isometry on horizontal legs. We choose this isometry so that it maps 00 to 0,\footnote{This is a shorthand for $e_0\otimes e_0\mapsto e_0$.} 10 to 1, 01 to 2, 11 to 3.

 We think of $A^{(0)}$ having infinitely many indices, but only a finite number are "active" i.e.~give rise to nonzero components: 2 vertical and 4 horizontal. We define one horizontal and vertical sectors for $A^{(0)}$ per each active index (distributing inactive indices into $\x$, $\u$, $\d$, $\r$ so that each of these sectors is infinite-dimensional):
 \begin{align}
 	& \text{vertical:\ }&&\o = {\rm span}(e_0),\quad \x={\rm span}(e_1,\ldots),\nonumber\\
 	& \text{horizontal:\ }&&\o = {\rm span}(e_0),\quad \u={\rm span}(e_1,\ldots),\quad \d={\rm span}(e_2,\ldots),\quad
 	\r={\rm span}(e_3,\ldots)\,.
 \end{align}
 We see that $A^{(0)}$, unlike $A$, allows a natural sector decomposition $\o \oplus \d \oplus \u \oplus \r$ on the horizontal legs, which is the reason why we defined it.

$A^{(0)}$ is a normalized tensor. We write $A^{(0)}=A_*+b^{(0)}(\beta)$ with $b^{(0)}(\beta)\equiv b_{\rm Ising}^{(0)}(\beta)\in \mathbb{H}_0$. (We will use the "Ising" subscript only in formulations of main results.)

The tensor $b^{(0)}(\beta)$ is $O(\beta)$ for $\beta\to0$ and has nonnegative tensor elements which are monotonic functions of $\beta$. It inherits these properties from $b(\beta)$.

 Let $\hat b^{(0)}(\beta)$ be the minimal hat-tensor for $b^{(0)}(\beta)$. Because $b^{(0)}(\beta)$ is a nonnegative tensor and because every sector contains one active index, $\hat b^{(0)}(\beta)$ is obtained from $b^{(0)}(\beta)$ replacing active indices by sector labels.

Instead of the "square" neighborhood $\mathbb{O}_\delta$ in Theorem \ref{th:stability}, we will need more general "rectangular" subsets of $\mathbb{H}_0$ defined by
\beq
\mathbb{O}(\hat b)=\{b\in \mathbb{H}_0: \|b_{abcd}\|\le  \hat b_{abcd}\text{ for all $abcd\ne \o\o\o\o$}\}\,,
\eeq
where $\hat b$ is any hat-tensor with $\hat b_{\o\o\o\o}=0$. In other words, $\mathbb{O}(\hat b)$ is a set of all $b\in \mathbb{H}_0$ for which $\hat b$ is a hat-tensor. If all non-$\o\o\o\o$ components of $\hat b$ are positive, $\mathbb{O}(\hat b)$ is a closed neighborhood of the origin, otherwise it has empty interior.

The proof of Theorem \ref{th:stability} starts by considering the tensor $\hat b^{(0)}$ whose non-$\o\o\o\o$ components are all equal to $\delta$. We now repeat that proof with a different $\hat b^{(0)}$, namely $\hat b^{(0)}=\hat b^{(0)}(\bar\beta)$ defined as above. This tensor defines a rectangular box around $A_*$, whose size grows with $\bar\beta$. This box then flows under RG. We would like to take this initial box as large as possible, by increasing $\bar\beta$, while preserving the condition that it eventually flows to zero size. We ensure the latter by checking that master functions iterates initialized at $\hat b^{(0)}(\bar\beta)$ start converging exponentially after 10-20 steps. We can also tweak the reweighting parameters $w_\x$ and $w_\o$.

\begin{proposition}\label{eq:Isingbasin}
	(a) The set of tensors $\mathbb{O}(\hat b_{\rm Ising}^{(0)}(\bar\beta))$, $\bar\beta=0.12$, is a basin of stability for the 2x1 RG map with the reweighting parameters $w_\x=2.3$, $w_{\o} = 2$.\\[5pt]
	\noindent (b) The free energy is analytic on a neighborhood of $\mathbb{O}(\hat b_{\rm Ising}^{(0)}(\bar\beta))$.\\[5pt]
	\noindent (c) $b_{\rm Ising}^{(0)}(\beta) \in \mathbb{O}(\hat b_{\rm Ising}^{(0)}(\bar\beta))$ for any $0\le \beta\le \bar\beta.$
\end{proposition}

\begin{proof} Part (a) is proved by the strategy from the preceding paragraph. \Cref{fig:ising} shows the corresponding master function iterates initialized at $\hat b_{\rm Ising}^{(0)}(\bar\beta)$. Condition \eqref{eq:start} is satisfied for $i_0=15$ with $\lambda=0.95$ \cite{our-code}. This proves part (a) by the arguments analogous to Theorem \ref{th:stability}.

We next argue for (b). Note that $\mathbb{O}(\hat b^{(0)}(\beta))$ has empty interior, some non-$\o\o\o\o$ elements of $\hat b^{(0)}(\beta)$ being zero, although starting from $\hat b^{(1)}(\beta)$ they are all nonzero (see \Cref{fig:ising}).\footnote{\label{note:selrule}Let's explain why this happens, taking sector $\u\x\u\o$ as an example. Since we are not subdividing the infinite-dimensional sectors by symmetry, there are $\mathbb{Z}_2$-even indices in $\u$, $\x$. For the initial tensor $b^{(0)}$, those indices are inactive, but they become active at subsequent iterations. See also the discussion in \cref{sec:symmetries}.} So we cannot appeal directly to Proposition \ref{prop:free}, but there is an easy workaround. Let $\hat b^{(0)}(\bar\beta,\epsilon)$ be a regularization of $\hat b^{(0)}(\bar\beta)$ where all non-$\o\o\o\o$ elements were increased by an $\epsilon>0$. When $\epsilon$ is very small, the first $i_0$ master function iterates of $\hat b^{(0)}(\bar\beta,\epsilon)$ are very close to those of $\hat b^{(0)}(\bar\beta)$. The given proof of Proposition \ref{eq:Isingbasin} then implies that $\mathbb{O}(\hat b^{(0)}(\bar\beta,\epsilon))$ is also a basin of stability, for a sufficiently small $\epsilon>0$. The latter set has nonempty interior, which contains $\mathbb{O}(\hat b^{(0)}(\bar\beta))$, and on which free energy is analytic by Proposition \ref{prop:free}.

Finally, (c) follows from the fact, stated above, that $\hat b_{\rm Ising}^{(0)}(\beta)$ is monotonic in $\beta$.
\end{proof}

Let us recap. We have translated the Ising model into a tensor network. We have shown that in the range $\beta\le \bar\beta=0.12$ the free energy of the resulting tensor network is analytic under arbitrary perturbations of sufficiently small HS norm (which include perturbations which break $\mathbb{Z}_2$ symmetry such as the magnetic field). This proves that the Ising model for $\beta\le 0.12$ is in the high-T phase. Note that the Ising model critical point is located at $\beta_c = \frac{1}{2} \log(1 + \sqrt{2}) \approx 0.44$. Thus, our RG argument captured roughly one quarter of the full high-T phase $\beta<\beta_c$.

For the sake of the argument, let us compare to the cluster expansion. The best published cluster expansion estimate we know of is $\beta<0.151$ \cite[Sec.7.0.1]{Procacci}.\footnote{We are grateful to Alessandro Giuliani for mentioning this optimized result, which also uses a somewhat improved version of the Koteck\'y-Preiss condition \cite{kotecky_preiss_1986}. We mention here two older results. Ref.~\cite[Section V.7, Example 2]{simon2014statistical} used the cluster expansion to prove analyticity for $\beta e^\beta<1/192$. Ref.~\cite{GALLAVOTTI} used a different method for a lattice gas. Applied to the Ising model, their condition gives analyticity for $\beta<0.02$.} So this is a bit better than our result for the Ising model. Of course both methods can be optimized further, if needed. As discussed further in Section \ref{sec:conclusions}, we see the main virtue of our method is that, after appropriate generalization, it can be used around the critical point.

\begin{figure}[h]
	\centering
	\codetag\includegraphics[scale=0.7]{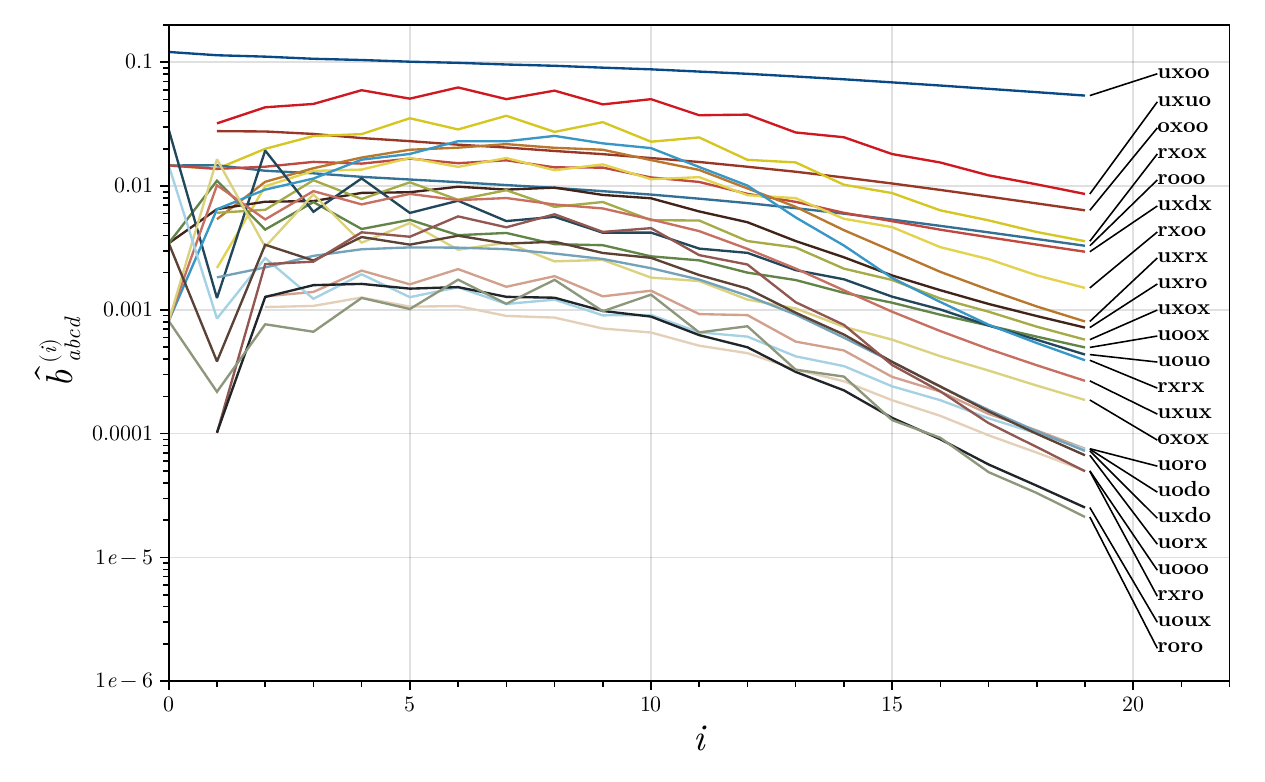}
	\caption{Master function iterates initialized at the tensor $\widehat{b}_{\rm Ising}^{(0)}(\bar\beta)$, $\bar\beta=0.12$. The reweighting parameters are: $w_{\o} = 2.0$ and $w_\x = 2.3$. Some of the trajectories appear only starting from the second step, as the initial tensor has these sectors identically zero.}
	\label{fig:ising}
\end{figure}

To conclude this section, let us use the master function iterates to get estimates of the free energy.
The free energy per site of the original tensor network obtained by translating the Ising model is given by
\beq
\label{eq:fAbeta}
f(A(\beta)) = \log(\cosh(4\beta)+3) + \frac{1}{2} f(b^{(0)}(\beta))\,,
\eeq
where the first term comes from the normalization of $A$ and the factor $1/2$ in front of the second term is because the preliminary RG step halved the number of tensors in the network. The free energy $f(b^{(0)})$ is computed by the series \eqref{eq:fbseries}. As discussed below Proposition \ref{prop:free1}, for a proper treatment we should evaluate the first few terms in \eqref{eq:fbseries} and estimate the rest. Since our purpose here is just to demonstrate the general idea, we will be content to simply bound all terms in $f(b^{(0)}(\beta))$ using the general formula \eqref{eq:explbnd}. Note that since $b^{(0)}(\beta)$ is a nonnegative tensor, the free energy is non-negative, and only the upper bound \eqref{eq:explbnd} is of interest. So we have:
\beq
\label{eq:fb0upper}
0\le f(b^{(0)}(\beta)) \le \sum_{i=0}^{i_0} (1+\delta_{i,i_0}) 2^{-i-1}C_i^+,
\end{equation}
where $C_i^+$ are given by \cref{eq:bndpm}. The bound on $f(b^{(0)}(\beta))$ and the resulting bounds on $f(A(\beta))$ are shown in \cref{fig:isingfbound} for $\beta \leq 0.12$. In the same plots we include the exact value of these quantities, extracted using Onsager's exact solution (see footnote \ref{foot:factor2}). For $f(A(\beta))$, the dominant contribution comes from the first term in \eqref{eq:fAbeta}. So the uncertainty due to $f(b^{(0)}(\beta))$ is hardly visible in much of the range of $\beta$.

\begin{figure}[h]
	\centering
	\codetag \raisebox{-0.6cm}{\includegraphics[width=0.45\textwidth]{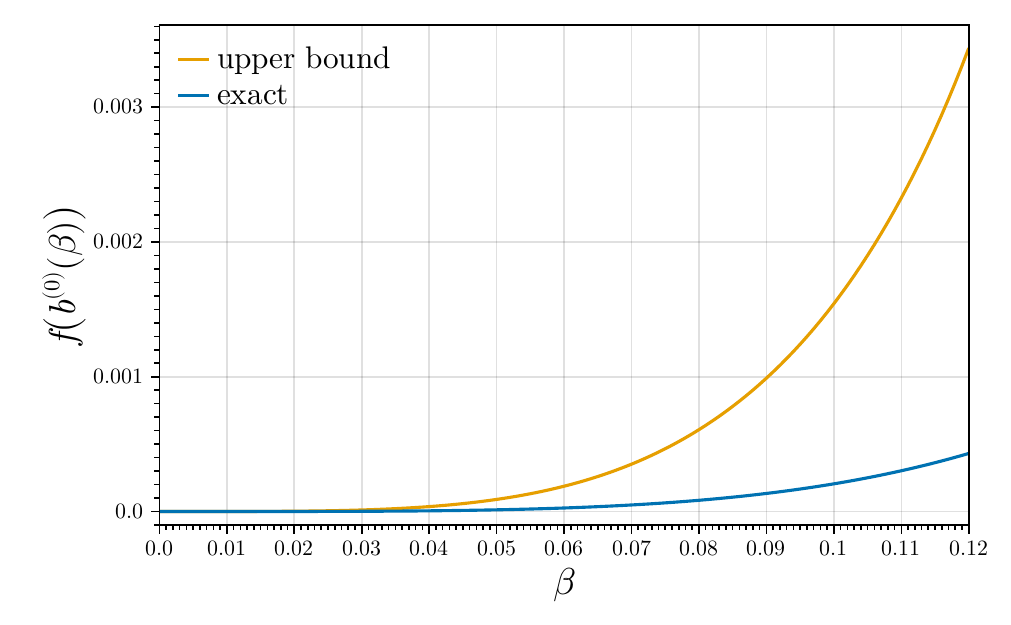}}(a)
	\raisebox{-0.6cm}{\includegraphics[width=0.45\textwidth]{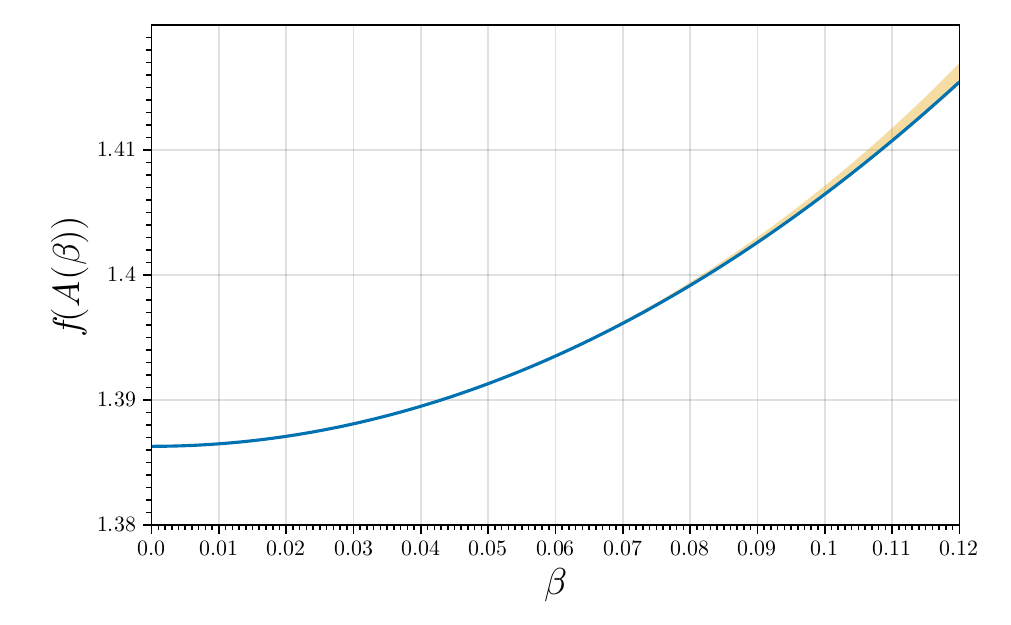}}(b)
	\caption{Bounds on the Ising model free energy compared to the exact values. (a) Orange: Upper bound \eqref{eq:fb0upper} on $f(b^{(0)}(\beta))$, computed for $\beta=0\ldots 0.12$ with step $10^{-3}$. For each $\beta$, master function iterations were performed up to $i_0=15$, checking the Key Lemma condition at the end. Blue: the exact value. (b) Shaded region: the allowed range of $f(A(\beta))$. Blue line: the exact value from Onsager's solution. }
	\label{fig:isingfbound}
\end{figure}

\subsection{XY model}\label{sec:xy-model}

Our last example is the (nearest-neighbor) XY model, defined by the following partition function:
\begin{equation}\label{eq:XY-pf}
  Z = \int \prod_{x} \frac{d \theta_{x}}{2\pi} \exp\Bigl( \beta \sum_{\langle xy\rangle} \cos(\theta_{x} - \theta_{y}) \Bigr),
\end{equation}
where $x$ are vertices of the square lattice, and the variables $\theta_{x}$ live on the unit circle $S^1\simeq[0,2\pi)$.
Since the space of spin states is continuous, tensor network representation of this model requires a tensor with infinitely many indices. This fits quite naturally into our framework.

We rewrite the partition function as a tensor network using the same method as for the Ising model in \cite{TEFR},\cite[App.~A.1]{paper1}, i.e.~rotating the square lattice by 45 degrees. Specifically, we encode the interaction among $\theta_1, \theta_2, \theta_3, \theta_4$ at the vertices of a single plaquette by the function
\begin{equation}\label{eq:XYcontinous}
  A_{\theta_1 \theta_2 \theta_3 \theta_4} = e^{\beta \left[ \cos(\theta_1 - \theta_2) + \cos(\theta_2 - \theta_3) + \cos(\theta_3 - \theta_4) + \cos(\theta_4 - \theta_1) \right] }.
\end{equation}
We view this function as a tensor on $V\times V\times V\times V$ where $V=L^2(S^1)$ with the measure $d\theta/(2\pi)$. The partition function \cref{eq:XY-pf} can then be obtained by placing copies of $A$ in every other plaquette in a checkerboard pattern and contracting neighboring tensors where contraction means integrating the product over $\theta\in S^1$. The tensor $A$ is Hilbert-Schmidt, the norm being given by
\beq
\|A\|=\left(\int \frac{d\theta_1}{2\pi}  \frac{d\theta_2}{2\pi} \frac{d\theta_3}{2\pi} \frac{d\theta_4}{2\pi}  |A_{\theta_1 \theta_2 \theta_3 \theta_4}|^2\right)^{1/2}<\infty\,.
\eeq

To connect with the theory developed in this paper, we introduce an orthonormal basis in $V$ and express $A$ in this basis. We use the orthonormal basis $e_n= e^{in\theta}$, $n \in \mathbb{Z}$. Thus our Hilbert space basis is numbered by $\mathbb{Z}$, not by $\mathbb{Z}_{\ge 0}$ as in the previous sections. Hopefully this will not lead to a confusion. The index $0$ will still play the same special role as in the previous sections. In this basis, our tensor is expressed as
\begin{equation}\label{eq:XYdiscrete1}
  A_{n_1 n_2 n_3 n_4} = \int \frac{d\theta_1}{2\pi} \frac{d\theta_2}{2\pi} \frac{d\theta_3}{2\pi} \frac{d\theta_4}{2\pi} \, A_{\theta_1 \theta_2 \theta_3 \theta_4} \, e^{-i(n_1 \theta_1 + n_2 \theta_2 + n_3 \theta_3 + n_4 \theta_4)}.
\end{equation}
To evaluate this, we use the Fourier series \cite[\href{https://dlmf.nist.gov/10.35.2}{(10.35.2)}]{NIST:DLMF}:
\begin{equation}\label{eq:expcos}
  e^{\beta \cos \theta } = \sum_{n \in \mathbb{Z}} I_n(\beta) e^{in\theta},
\end{equation}
where $I_n(\beta)$ is the modified Bessel function of the first kind. Substituting \eqref{eq:expcos} into \eqref{eq:XYcontinous}, we obtain after a straightforward computation:
\begin{equation}\label{eq:XYdiscrete2}
  A_{n_1 n_2 n_3 n_4} = \sum_{k \in \mathbb{Z}} I_k(\beta) I_{k + n_1}(\beta) I_{k + n_1 + n_2}(\beta) I_{k + n_1 + n_2 + n_3}(\beta) \, \delta_{n_1 + n_2 + n_3 + n_4}.\codetag
\end{equation}
Note that the XY model is invariant under the $U(1)$ global symmetry, under which the state $e_n$ has charge $n$. The $U(1)$ symmetry is visible in $A$ through the charge conservation condition $n_1 + n_2 + n_3 + n_4=0$ for the nonzero tensor elements.

Let us note that \cref{eq:XYdiscrete2} is not the only way to represent the XY model as a tensor network. Refs.~\cite{PhysRevE.89.013308, Jha_2020} provide an alternative method, which does not rotate the lattice by 45 degrees. It gives rise to a simpler-looking tensor
\beq
\label{eq:altXT}
 A^{\rm there}_{n_1 n_2 n_3 n_4} = \sqrt{I_{n_1}(\beta) I_{n_2}(\beta) I_{n_3}(\beta) I_{n_4}(\beta)} \, \delta_{n_1 + n_2 + n_3 + n_4}\,.
 \eeq
 Both \eqref{eq:XYdiscrete2} and \eqref{eq:altXT} are valid tensors representing the XY model partition function. They can be both used to get a bound on the high-T phase of the model. Below we work with \cref{eq:XYdiscrete2}, for two reasons. First, it was obtained by the same method we used of the Ising model, and it will be interesting to compare the result to the Ising model. Second, it happens to give a better bound than \eqref{eq:altXT}.\footnote{This was the result of our computations. A possible intuitive explanation is as follows. Because of the 45 degree lattice rotation, the tensor network built out of $\ell^2$ $A$'s corresponds to the XY partition function of $2\ell^2$ spins (see footnote \ref{foot:factor2} for the Ising). So passing from the spins to the tensors we have already performed a sort of "RG step" of scale factor $\sqrt{2}$. This is expected to bring us a bit closer to the high-T fixed point. Therefore, tensor RG iterations should have an easier time converging for $A$ than for $A^{\rm there}$ for which there is no such effect.}

 We write
 \beq
 \label{eq:bbeta}
 A=A_{0000} (A_*+b(\beta)),
 \eeq
 where $b(\beta)\in \mathbb{H}_0$. We define $\o$ and $\x$ as follows:
 \begin{equation}
 	\o = \mathrm{span}(e_0), \qquad \x = \mathrm{span}(\{e_n: n \ne 0\}).
 \end{equation}
 We use these sectors on both vertical and horizontal legs of $b(\beta)$. As a first step, we need to compute a hat-tensor $\hat b(\beta)$ for $b(\beta)$. $I_n(\beta)$ decreases exponentially with $|n|$. So the series \eqref{eq:XYdiscrete2} defining $A_{n_1 n_2 n_3 n_4}$ converges rapidly, and moreover the resulting tensor coefficients decrease exponentially for large $n_1,n_2,n_3,n_4$. For a purely numerical study, one could evaluate sufficiently many of these tensor coefficients, and thus get an approximate $\hat b(\beta)$. However since here we are aiming for a rigorous result, we need a rigorously valid $\hat b(\beta)$, which needs a bit of work to estimate tails. We describe our procedure in Appendix \ref{sec:hat-bbeta-for-the-xy-model}.

 Note that all states in the $\x$ sector carry non-zero $U(1)$ charge. So the hat tensor $\hat b(\beta)$ satisfies the selection rule $\hat b_{\o\o\o\x}=0$ (and rotations).

\codetag We next do the preliminary RG step expressed by the same diagram \eqref{eq:preliminary} as for the Ising model analysis. This defines the tensor $A^{(0)}=A_*+b_{\rm XY}^{(0)}(\beta)$ where $b_{\rm XY}^{(0)}(\beta)\in \mathbb{H}_0$. The isometry on the horizontal legs acts as follows:
\beq
\o\otimes\o\to \o,\quad \x\otimes\o\to \u, \quad \o\otimes\x\to\d, \quad \x\otimes\x\to\r.
\eeq
Thus, the horizontal Hilbert space of $A^{(0)}$ decomposes into 4 sectors $\o\otimes\u\otimes\d\otimes \r$.
We obtain the hat-tensor $\hat b_{\rm XY}^{(0)}(\beta)$ from $\hat b(\beta)$. Note that $(\hat b_{\rm XY}^{(0)})_{\o\o\o\o}=0$ because $\hat b_{\o\o\o\x}=0$.

\begin{remark}\label{rem:selruleXY}After the preliminary RG step, all states in sectors $\x,\u,\d$ (but not in $\r$) carry nonzero $U(1)$ charge. So $\hat b_{\rm XY}^{(0)}(\beta)$ like $\hat b(\beta)$ satisfies a selection rule: elements with three indices $\o$ and the fourth index $\x$, $\u$ or $\d$ vanish. After we start acting with the 2x1 map, some states of zero charge will enter also into $\x,\u,\d$, and this selection rule will no longer apply. This is visible in \Cref{fig:XY} below. See also Section \ref{sec:symmetries} and footnote \ref{note:selrule}.
	\end{remark}

\begin{figure}
	\centering
	\includegraphics[scale=0.7]{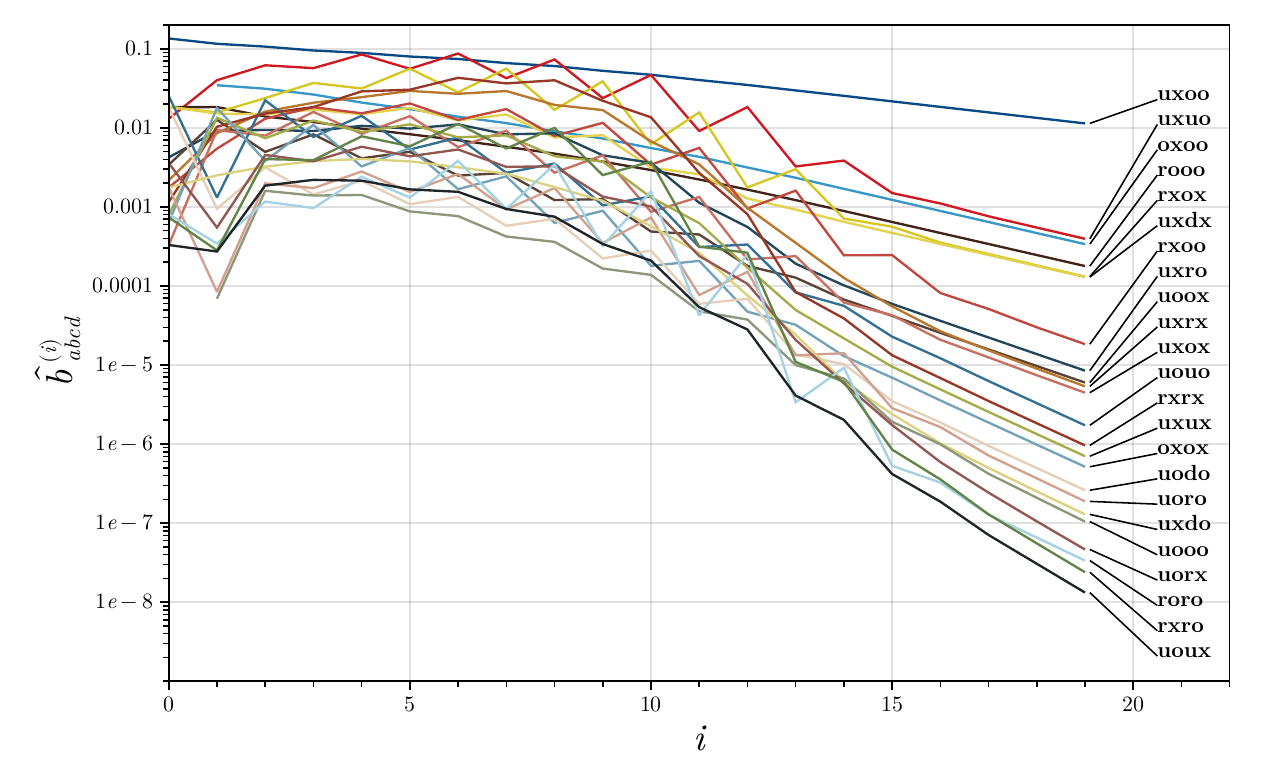}
	\caption{RG flow of the $\widehat{b}$ tensor for XY model initiated at $\beta=0.19$. The 2x1 map parameters are: $w_{\o} = 2.4$ and $w_\x = 2.3$. Some of the trajectories appear only starting from the second step, as the initial tensor has these sectors identically zero (see Remark \ref{rem:selruleXY}).}
	\label{fig:XY}
\end{figure}

We then proceed like in the Ising model analysis, replacing $\hat b_{\rm Ising}^{(0)}(\bar\beta)$ with $\hat b_{\rm XY}^{(0)}(\bar\beta)$. We tweak the values of the reweighting parameters a bit to improve $\bar\beta$. We thus obtain:
\begin{proposition}\label{eq:XYbasin}
	\noindent (a) The set of tensors $\mathbb{O}(\hat b_{\rm XY}^{(0)}(\bar\beta))$, $\bar\beta=0.19$, is a basin of stability for the 2x1 RG map with the reweighting parameters $w_\x=2.3$, $w_{\o} = 2.4$.\\[5pt]
	\noindent (b) The free energy is analytic on a neighborhood of this basin of stability.\\[5pt]
	\noindent (c) $b_{\rm XY}^{(0)}(\beta)\in \mathbb{O}(\hat b_{\rm XY}^{(0)}(\bar\beta))$ for all $0\le \beta\le \bar\beta-\delta$, $\delta=5\times 10^{-5}$.
\end{proposition}
\begin{proof} For (a),(b) we proceed as in Proposition \ref{eq:Isingbasin}(a),(b) for the Ising model. \Cref{fig:XY} shows the corresponding master function iterates. The hypothesis of \cref{lem:key} is satisfied for $i_0=15$ with $\lambda=0.86$ \cite{our-code}.

Part (c) is different from Ising. There, it was easy to see analytically that $\hat b(\beta)$, and hence $\hat b^{(0)}(\beta)$, were monotonic in $\beta$, so Proposition \ref{eq:Isingbasin}(c) was trivial with $\delta=0$. For XY, we don't have an easy proof of this monotonicity, due to more complicated expressions of tensor elements. (Numerically we do see monotonicity.) So we prove (c) by a computer-assisted argument \cite{our-code}. We split the  interval $[0,\bar\beta-\delta]$ into finitely many intervals $I_i=[\beta_i,\beta_{i+1}]$. For each interval $I_i$ we compute a tensor $\hat b_i^{(0)}$ which is a hat tensor for all $b^{(0)}(\beta)$, $\beta\in I_i$. This is done by running the procedure described in App.~\ref{sec:hat-bbeta-for-the-xy-model} in interval arithmetic.\footnote{To compute elements of the hat tensor $\hat b(\beta)$ we need to evaluate the norm of tensor elements in each sector, and normalize dividing by $A_{0000}$. Since all tensor elements are monotonic, Eq.~\eqref{eq:XYdiscrete2}, we obtain a valid hat tensor for all $\beta\in I_i$ evaluating the norm at $\beta=\beta_{i+1}$ and dividing by $A_{0000}$ at $\beta=\beta_{i}$. Our interval arithmetic code implements this logic automatically.} We then verify that $\hat b_i^{(0)}\le \hat b^{(0)}(\bar\beta)$ holds for all $i$.
\end{proof}

So we have proven that the XY model for $\beta \leq 0.18995$ lies in the high-T phase. The critical temperature for the XY model, where it transitions to a phase with algebraically decaying correlations (the BKT transition), is not exactly known. Various numerical techniques indicate that $\beta^{\text{XY}}_c \approx 1.12$ (see e.g.~\cite[Table 1]{Jha_2020} or \cite[Table I]{UedaBKT}), so our method captures approximately $17\%$ of the high-T phase in this case.
Note that the XY model is more disordered than the Ising model at the same $\beta$, in the sense that the spins have more room to fluctuate, and so it's natural that $\beta_c^{\rm XY}>\beta_c^{\rm Ising}$. It's reassuring that our method is also able to see this, as it gives a larger $\bar\beta$ for the XY than for the Ising. (There is also a rigorous bound $\beta_c^{\text{XY}} \geq 2\beta_c^{\text{Ising}}$ \cite{AIZENMAN1980281}.)

\section{Conclusions}\label{sec:conclusions}

We have introduced a new tensor RG map, but most importantly a new framework to do controlled renormalization group using tensor networks. The main characteristic of our RG map is that it comes equipped with a bounding box and a computable "master function" which shows how this bounding box varies from one RG step to the next. The RG flow happens in an infinite-dimensional space of tensors, but thanks to the bounding box we can control the RG flow via a finite calculation. Finding such a computer-assisted RG framework for lattice models has always been a dream, but it is perhaps for the first time that it is realized. This should allow myriad applications, once some generalizations described below are implemented.

In this paper we considered the two simplest concrete applications of this framework: proving that a lattice model is in the high-T phase and computing the model's free energy within the proven high-T domain. Of course there exist other means to do this. Perhaps the most standard tool is the cluster expansion, a.k.a.~the polymer expansion \cite{friedli_velenik_2017}. Other methods include \cite{DobrushinShlosman1985,Dobrushin1985,DobrushinShlosman1987} and \cite{Kennedy1990}. In this work we saw that in its simplest form presented here, our method goes somewhat further than simple applications of the cluster expansion for general models, while for specific models like the Ising model the optimized cluster expansion performs somewhat better.

{ For the future, another application within reach of our method is the analysis of the correlation functions on the lattice.
	This was discussed in Ref.~\cite{paper1} with the help of the cluster expansion, but it can be done by purely RG techniques, as follows. In the tensor network language, an $n$-point correlation function can be represented as a tensor network built of the tensor $A$ with $n$ tensors replaced by some ``defect'' tensors $B_1,\ldots ,B_n$ representing insertions of local operators. Then, one may apply the 2x1 map to such a network and track the flow of $B_i$'s in addition to that of $A$. The locality of the map guarantees that the $A$ flow remains unaffected by the presence of the defects, so $A$ will still approach the high-T fixed point tensor, which should also make the $B$ flow manageable for analysis. One should thus be able to recover the expected result that the connected correlation functions decay exponentially with the distance. }


Let us discuss the possible ways to improve the domain of applicability of our approach. There are a few improvements that should be relatively straightforward to implement. When there was a contraction of a left tensor and a right tensor, we introduced reweighting factors $w_\x$ and $w_\o$ using the trivial identity $w \frac{1}{w}=1$. A generalization is to introduce reweighting {\it matrices} $w=w_{AB}$ by inserting in the contraction the identity matrix in the form $w w^{-1}$ or $w^{-1} w$. Indices $A,B$ of these matrices will be of the form $A,B=abc$ where $a,b,c$ are the sectors cut through by the red lines in Fig.~\ref{fig:LR} and Eq.~\eqref{eq:concl2}. One can then optimize over invertible matrices $w$.

Our disentangler was chosen to cancel a certain class of diagrams to first order. Some such cancellation is essential to make the linearized map be a contraction, but the disentangler also produces terms that are not involved in this cancellation which can degrade our bounds. Using a less aggressive disentangler that does not fully cancel the selected diagrams could result in an improved domain of applicability.

We kept the parameters $w_\x$ and $w_\o$ constant along the RG trajectory. One can hope to improve the domain of applicability by allowing these parameters to vary for some finite number of initial steps of the RG flow. More generally, one can try varying other aspects of the RG map, e.g., the disentangler, for a finite number of RG steps.

Let us discuss now a more ambitious idea which should lead to a dramatic improvement of our method.
In following our tensor along the RG trajectory, we wrote the tensor as $A_* + b$, and the only information we kept track of for $b$ was the hat tensor $\hat{b}$. The approach is to express the tensor in the form $A=A_* + b_f + b$ where $b_f$ is a finite dimensional tensor and $b$ is infinite dimensional. We express the image of $A$ under the RG map in the form  $A_* + b_f' + b'$ where $b_f'$ is computed explicitly using a finite dimensional map that approximates the RG map. The infinite dimensional tensor $b'$ is then determined by requiring that $A_* + b_f' + b'$ is the image of $A_* + b_f + b$. The tensor $b$ is controlled by a hat tensor. The master function for that hat tensor is now a function of $b_f,b$. Since $b$ contributes to the renormalization of $b_f$, tensor $b_f'$ should be controlled by a set of intervals for its elements. After a finite number of iterations one can hope that $b_f$ will be small enough that we can switch to the map studied in this paper.

{ Our construction here was linked to a square lattice in 2D. Can it be extended to other 2D lattices or to 3D?
	Let us discuss these questions in turn.

	First, even if the original lattice spin model is defined say on 2D triangular or hexagonal lattice, we can rewrite it as a tensor network model on the square lattice by deforming the original lattice.  Then our results based on the 2x1 map apply, although this sacrifices the original lattice symmetry and may lead to suboptimal results. A different question is about RG maps adapted to the triangular or hexagonal lattice. For such numerical algorithms see \cite{Levin:2006jai, GILT, LoopTNR}. As for such rigorous RG maps, the method of disentangler maps on $V\otimes V$ used here does not seem to apply. An alternative idea which should be applicable is that of rearrange disentanglers \cite{Kennedy2022,Ebel:2024jbz}.

	As for 3D, the only current rigorous result is \cite{Ebel:2024jbz}, where the stability of the 3D high-T fixed point was established using rearrange disentanglers \cite{Kennedy2022}. Using hat-tensors, the algorithm of \cite{Ebel:2024jbz} could be used to obtain rigorous quantitative results about the size of the 3D high-T phase. One may also try to adapt the 2x1 scheme developed here to 3D, contracting the network in one direction at a time and rotating it appropriately. It remains to be seen if this can be done.\footnote{One suggestion would be to use the analogues of $\u$ and $\d$ sectors (one would need $4$ sectors in 3D) that would allow one to distinguish a class of special sectors that would hopefully be treatable by disentanglers.} 
}

Our long range goal is to use ideas from this paper to prove the existence of a critical fixed point tensor for some tensor RG map and control the map in a neighborhood of this fixed point. Here is a sketch of how such a proof might proceed. We express the tensor in the form $A=A_0 + b_f + b$ where $A_0$ and $b_f$ are finite dimensional tensors and $b$ is infinite dimensional. $A_0$ is an approximation to the critical fixed point. As above, the tensor $b_f$ is computed explicitly using a finite dimensional map that approximates the RG map, and $b'$ is determined by requiring that the RG map takes $A_0+b_f+b$ to $A_0+b_f^\prime+b^\prime$. ($A_0$ is kept fixed.) The tensor $b$ is controlled by a hat tensor. Our experience with existing numerical RG maps make us optimistic that one can get an approximate fixed point $A_0$ that is easily within $0.02$ of the exact fixed point. We would then prove that there exist $b_f,b$ such that $A_0+b_f+b$ is an exact fixed point. Most existing tensor network RG maps are not suitable for this program, in particular most use SVD which is not compatible with our HS norms. But we are hopeful that we can build a good RG map by modifying our 2x1 map.

\acknowledgments

SR and NE were supported in part by the Simons Foundation grant 733758 (Simons Bootstrap Collaboration). SR thanks the International Solvay Institute for Physics (Brussels) for hospitality while this work was being finalized. SR thanks the theoretical physicists of Belgium (as well as the architects, the painters, and the brewers) for the stimulating atmosphere. NE thanks Christophe Garban for the lecture series on BKT phase transitions that inspired the study of the XY model example for this paper.

\appendix

\section{Accompanying code}\label{sec:accompanying-code}

All our main results were obtained with computer assistance. We provide the {\tt Julia}~\cite{Julia-2017,Julia} code used in our analysis in the notebook {\tt HT.ipynb}, which is available in our GitHub repository~\cite{our-code}. The main part of the code, defining the master function, comprises only about 300 lines of code distributed across notebook cells 3--5. We have thoroughly documented and commented the code, and we hope this will facilitate its understanding.

Readers who wish to check whether their favorite lattice model lies in the high-temperature phase can follow the instructions in the {\tt README} file. A proof of such result requires only a few lines of code (see, for example, the first code cell in Section "General tensors," which proves \cref{th:stability}).

Several packages made our results possible. \cref{eq:lead} relies on {\tt TensorOperations.jl} \cite{TensorOperations.jl}, {\tt Symbolics.jl} \cite{gowda2022high}, and {\tt TaylorSeries.jl} \cite{benet_2025_15122713}. The proofs presented in \cref{sec:results} depend on {\tt TensorOperations.jl} \cite{TensorOperations.jl}, {\tt ArbNumerics.jl} \cite{ArbNumerics.jl}, {\tt SpecialFunctions.jl} \cite{SpecialFunctions.jl}, and {\tt HCubature.jl} \cite{HCubature.jl, Genz1980}. We use {\tt CairoMakie.jl}~\cite{DanischKrumbiegel2021} for plotting.

\section{$\hat b(\beta)$ for the XY model}\label{sec:hat-bbeta-for-the-xy-model}

In this appendix we explain how we get a rigorous hat-tensor $\hat b(\beta)$ for the XY tensor $b(\beta)$ defined by Eq.~\eqref{eq:bbeta}. We first present a series of lemmas containing estimates for the series defining $A_{n_1 n_2 n_3 n_4}$, and for the tails of the HS norm of $A$. We will then explain how we get $\hat b(\beta)$ using those lemmas.

Recall some useful facts about the Bessel functions. We have $I_n(\beta) = I_{|n|}(\beta)$. The standard power series representation of $I_n(\beta)$ implies that  $I_n(\beta)>0$ for $\beta>0$. For $n\ge 0$, it also implies a bound:
\beq
\label{eq:Iexp}
I_{n+1}(\beta)\le \frac{\beta/2}{n+1} I_n(\beta)\,.
\eeq
Using this several times, we get, for $n\ge 0$,
\beq
\label{eq:Iexp1}
I_{n}(\beta)\le \frac{(\beta/2)^{n}}{n!} I_0(\beta)\,.
\eeq
We will sometimes write $I_n$ for $I_n(\beta)$. In our estimates we will assume $\beta < 2$.
\begin{lemma}\label{lem:singleA}
	Suppose $k_-\le  0$, $k_+\ge 0$ are such that the four integers $k,k+n_1,k+n_1+n_2,k+n_1+n_2+n_3$ are all $\ge0$ for $k=k_+$ and all $\le 0$ for $k=k_-$. Denote $W_k= I_k I_{k+n_1}I_{k+n_1+n_2} I_{k+n_1+n_2+n_3}$. Then, for $\beta<2$,
	\beq
	\label{eq:Abnd}
	\sum_{k_-}^{k_+} W_k\le A_{n_1n_2n_3n_4}\le \sum_{k_-}^{k_+} W_k+ \frac{(\beta/2)^4}{1-(\beta/2)^4}(W_{k_+}+W_{k_-}).\codetag
	\eeq
	\end{lemma}
 \begin{proof} The l.h.s.~is obvious since all terms in the sum \eqref{eq:XYdiscrete2} defining $A_{n_1n_2n_3n_4}$ are positive. For the r.h.s., terms for $k> k_+$ can be estimated using \eqref{eq:Iexp} by $(\beta/2)^{4(k-k_+)}$ times $W_{k_+}$. Analogously for terms with $k<k_-$.
\end{proof}

 \begin{lemma}\label{lem:Abound2}
 	 Denote $n=\max(|n_1|,|n_2|,|n_3|,|n_4|)$. Then, for $\beta<2$,
 	\beq
 	A_{n_1n_2n_3n_4} \le \frac{(I_0(\beta))^4}{1-(\beta/2)^4} \frac{4n+1}{\Gamma(n/2+1)^2} (\beta/2)^{n}\,.
 	\eeq
 \end{lemma}
  \begin{proof} We pick $k_\pm=\pm 2n$. These $k_\pm$ satisfy the requirements of the previous lemma (note that $n_1+n_2+n_3=-n_4$ by charge conservation). So we can use the upper bound from that lemma. Consider any $W_k$ for $|k|\le 2n$. Suppose for definiteness that $n_1$ is one of the maximal indices, i.e.~$n=|n_1|$ (other cases are similar). Other indices $n_2,n_3,n_4$ may be anywhere in the range $[-n,n]$. Suppose also for definiteness that $n_1\ge 0$. Then we estimate
  	\beq
  	W_k\le I_k I_{k+n} I_0^2\le \frac{(\beta/2)^{|k|+|k+n|}}{|k|!|k+n|!} I_0^4,
  	\eeq
  	where in the second inequality we used \eqref{eq:Iexp1}. The maximum of the r.h.s. is realized for $k=-n/2$. Estimating all terms in $k\in[-2n,2n]$ by the maximum, and adding the last term in \eqref{eq:Abnd}, we get the stated result.
\end{proof}

\begin{lemma}\label{lem:counting} For a positive integer $n$, define the set $T_n$ of triples of integers $(x,y,z)$ which are all not larger than $n$ in absolute value and whose sum equals $n$:
	\beq
	T_n=\{ (x,y,z)\in \mathbb{Z}^3 : x,y,z\in[-n,n], x+y+z=n \}
	\eeq
	Then $|T_n|=2n^2+3n+1$.
\end{lemma}
This is an elementary counting exercise. We omit the proof.

 		\begin{lemma}\label{lem:tailnorm} Suppose $N\ge 3$. Then, for $\beta<2$,
 			\beq
 			\sum_{\max(|n_1|,|n_2|,|n_3|,|n_4|)\ge N} |A_{n_1n_2n_3n_4}|^2 \le 4 K_N \left[\frac{(I_0(\beta))^4}{1-(\beta/2)^4}\right]^2 \frac{(\beta/2)^{2N}}{1-(\beta/2)^2},\codetag
 			\eeq
 			where $K_n=2|T_n|\left(\frac{4n+1}{\Gamma(n/2+1)^2}\right)^2$.
 		\end{lemma}
 		\begin{proof}
 			We split the sum into $4$ groups corresponding to which $|n_i|$ is largest (the groups overlap because several $|n_i|$ may be equal). Consider for definiteness the group where $|n_4|$ is largest. Many terms in the group are zero because they don't satisfy charge conservation. By Lemma \ref{lem:counting}, the number of nonzero terms with a fixed $|n_4|=n$ is $2|T_n|$  (the factor 2 accounts for two signs of $n_4$). Bounding each nonzero term by Lemma \ref{lem:Abound2}, the sum of terms in the group is bounded by
 			\beq
 			\left[\frac{(I_0(\beta))^4}{1-(\beta/2)^4}\right]^2 \sum_{n\ge N}  K_n (\beta/2)^{2n} \le  K_N\left[\frac{(I_0(\beta))^4}{1-(\beta/2)^4}\right]^2 \frac{(\beta/2)^{2N}}{1-(\beta/2)^2},
 			\eeq
 			where we used that $K_n$ is monotonically decreasing for $n\ge 3$ as it is easy to check. We then multiply this estimate by $4$ to account for $4$ groups.
 			 \end{proof}

The nonzero components of the minimal hat-tensor $\hat b=\hat b(\beta)$ are given by
\begin{align}
\hat b_{abcd}=(A_{0000})^{-1}\|A_{abcd}\|,\quad
abcd= \o\o\x\x,\o\x\o\x,\o\x\x\x,\x\x\x\x\,,
\label{eq:bXYsectors}
\end{align}
and rotations thereof (the tensor $\hat b$ is rotation-invariant). To get a hat-tensor we have to provide upper bounds for all quantities in the r.h.s. These bounds are obtained by Lemmas \ref{lem:singleA},\ref{lem:tailnorm}. Moreover we can get bounds close to the exact values by playing with the parameters $k_\pm,N$ in those lemmas. To make this clear let $\eps>0$ be an accuracy parameter. In practice maybe $\eps=10^{-16}$ the machine precision. The bounds described below all go to the exact values as $\eps\to0$.

A lower bound on $A_{0000}$ (which gives an upper bound on its inverse) is given by the l.h.s.~of \eqref{eq:Abnd} taking $k_\pm$ so large that the r.h.s.~and the l.h.s.~differ by less than $\eps$.

Let $A_{abcd}$ be any of the $A$ sectors in the r.h.s.~of \eqref{eq:bXYsectors}. We choose $N$ so large that the r.h.s.~of \eqref{lem:tailnorm} is smaller than $\eps$. We then bound $\|A_{abcd}\|$ by
\beq
\label{eq:AabcdXY}
\|A_{abcd}\|\le \Biggl(\sum_{\mathop{|n_1|,|n_2|,|n_3|,|n_4|<N}\limits^{\scriptstyle n_1\in a, n_2\in b,n_3\in c, n_4\in d}}|A_{n_1n_2n_3n_4}|^2+\eps\Biggr)^{1/2}\,.
\eeq
So we are reduced to providing upper bounds for each $A_{n_1n_2n_3n_4}$ in the r.h.s.~of \eqref{eq:AabcdXY}. Since $N$ is fixed this is a finite number of tensor elements to consider. We bound them by the r.h.s. of \eqref{eq:Abnd} taking $k_\pm$ as large as that lemma requires so that it applies, and so that the two sides of \eqref{eq:Abnd} differ by less than $\eps$. We use interval arithmetic. In particular the Bessel functions $I_n(\beta)$ are evaluated in interval arithmetic. This is how we get a rigorous hat-tensor $\hat b(\beta)$.

\section*{Data availability statement}

The data supporting this study were generated using the publicly available code \cite{our-code}

\section*{Conflict of Interest statement}

The authors declare that they have no competing interests.

%
\providecommand{\href}[2]{#2}\begingroup\raggedright\endgroup

\end{document}